\documentclass[journal=psy]{CUP-JNL-DTM} %cup-journal}

\usepackage{graphicx}
\usepackage{multicol,multirow}
\usepackage{amsmath,amssymb,amsfonts}
\usepackage{mathrsfs}
\usepackage{amsthm}
\usepackage{rotating}
\usepackage{appendix}
\usepackage{ifpdf}
\usepackage[T1]{fontenc}
\usepackage{newtxtext}
\usepackage{newtxmath}
\usepackage{textcomp}
\usepackage{xcolor}
\usepackage{lipsum}
\usepackage[colorlinks,allcolors=blue]{hyperref}
\usepackage{longtable}
\usepackage{booktabs}
\usepackage{longtable}
\usepackage{array}
\usepackage{caption}
\usepackage{csquotes}
\usepackage{mathtools}
\usepackage{tabularx,booktabs,adjustbox,makecell}
\newcolumntype{Y}{>{\raggedright\arraybackslash}X} % wrapped, left-aligned X column
\usepackage{enumitem}
\usepackage[english]{babel}
\usepackage{algorithm}
\usepackage{tabularray}
\usepackage{float}
\usepackage{graphicx}
\usepackage{wrapfig}
\usepackage{codehigh}
\usepackage{anyfontsize}
\usepackage[normalem]{ulem}
\UseTblrLibrary{booktabs}
\usepackage{setspace}
\usepackage{subcaption}
\usepackage{floatflt}

\usepackage{threeparttable} % For better table notes
\NewTableCommand{\tinytableDefineColor}[3]{\definecolor{#1}{#2}{#3}}
\usepackage{amsmath}

\theoremstyle{plain}
\newtheorem{theorem}{Theorem}[section]
\newtheorem{proposition}[theorem]{Proposition}
\newtheorem{lemma}[theorem]{Lemma}
\newtheorem{corollary}[theorem]{Corollary}
\theoremstyle{definition}

\theoremstyle{remark}

\jname{Data/Math}
\articletype{STATISTICAL METHODS [preprint]}
\jyear{2026}
\begin{document}

\begin{Frontmatter}

\title[Article Title]{Refactor Analysis: Powerful and Predictive Evaluations of Factor Models and Dimensionality}

% There is no need to include ORCID IDs in your .pdf; this information is captured by the submission portal when a manuscript is submitted. 
% \author[1]{Anonymous}
\author[1]{Michael Hardy}
% \author[1]{Benjamin Domingue}
% \author[2]{Author Name3}

% \authormark{Anonymous}

% \address[1]{ %\orgdiv{GSE}, 
% \orgname{Anonymous Organization}, \orgaddress{\city{Everytown}, 
% % \postcode{Pincode}, 
% \state{Earth},  \email{someone@something.com}
% % \country{Country}
% }}

\authormark{Hardy}

\address[1]{ %\orgdiv{GSE}, 
\orgname{Stanford University}, \orgaddress{\city{Stanford}, 
% \postcode{Pincode}, 
\state{CA},  \email{hardym [ at ] stanford [ dot ] edu}
% \country{Country}
}}

\abstract{Unidimensional factor models justify some of the most consequential summaries in science---single scores, single ranks, and single leaderboards---yet unidimensionality is usually assessed indirectly, by fitting and evaluating models on \emph{images} of the data (e.g., correlation matrices) rather than on the response matrix itself. We introduce \textbf{Refactor analysis}, a data-first evaluation paradigm that converts a one-factor solution into a rank--1 prediction %$\widehat X=\widehat u\,\widehat v^\top$ 
of the original matrix %$X$ 
by estimating both respondent- and item-side structure from dual association images. We further introduce \textbf{Verifactor analysis}, which evaluates the same construction under bi-cross-validated (BCV) row--column partitions, targeting the random-respondent random-item generalization regime. In simulations where the data-generating mechanism is truly rank--1 and correlational, Refactor metrics align with classical unidimensionality indices, validating the approach. However, across 200 public dichotomous datasets, traditional fit and unidimensionality measures---though highly intercorrelated---are weakly related to data recoverability, especially out of sample. This gap exposes a methodological vulnerability: excellent image-based fit can coexist with poor data-level explanatory power. Finally, treating the association measure itself as a testable hypothesis, we compare $\phi$, tetrachoric, and quadrant correlation. Quadrant correlation emerges as a simple, interpretable, and remarkably robust alternative, yielding consistently stronger reconstruction and more stable behavior under sample-size variation than commonly used correlations. Together, Refactor and Verifactor shift unidimensionality assessment from ``does a one-factor model fit the correlation matrix?'' to the question that matters for measurement and benchmarking: \emph{does a one-factor dependence structure recover and generalize the observed responses?}}

\end{Frontmatter}

% To make reconstruction-based evaluation practical and interpretable for binary and ordinal responses, we report isotonic $R^2$ (best monotone variance explained) alongside a diagnostic suite spanning discrimination (AUC), calibration (likelihood / cross-entropy), nonlinear dependence (bias-corrected distance correlation), and geometric fidelity (matrix cosine). 

% \section*{Impact Statement}
% lol
% Some Data journals (DAP, DCE) require an `Impact Statement' section. Comment out this section if it is not required.

% Some math journals (FLO) require a table of contents. Comment out this line if no ToC is needed.
% \localtableofcontents

% \input{sections/0x_ncme_temp_intro}

\section{Introduction}
\label{sec:intro}
Physicists have had great success reducing the world's messy complexity through powerful mathematical simplifications—creating so-called ‘spherical cows’—which represent testable underlying laws. In measurement, we have performed similar reductions in pursuit of science and understanding, but these simplifications are typically based on statistical theory rather than underlying “physical” laws. In the era of Big Data, we now have access to datasets and processes that enable interrogation of our most prized spherical cows at scale: this study tests the methods and assumptions with which we evaluate unidimensional factor analyses. 

% Classical factor model methods and metrics for dichotomous items have been in use for well over 120 years and have had the allure of being both interpretable and computationally tractable. %With more data and compute, we can explore methods that retain or improve statistical interpretability whilst also providing insights that more closely align with the reality we are seeking to study. In short, we are asking: is it possible that we don’t have to sacrifice predictive power when modeling the underlying factor dependence structures?

%This is not a study to show that unidimensionality or factor analyses are inadequate representations of the world: this has been discussed at length for more than a century. Instead, w
% We propose new methods for evaluating the assumptions required of factor models fit and compare these new predictive measures of fit against extant metrics and methods. 
This study presents approaches to evaluating assumed or hypothesized latent relationships by measuring the recoverability of the data through the low-rank representations implied by factor models, some of the oldest methods of psychometrics at nearly 125 years old. We introduce Refactor Analysis, a paradigm that directly evaluates whether interpretable low-rank models can reconstruct the important variation found in the original response matrix and its out-of-sample predictive extension Verifactor Analysis which uses bi-cross-validated (BCV) block prediction \citep{owen_bi-cross-validation_2009}. After demonstrating Refactor analyses’ expected behavior via simulations, we then compare traditional and Refactor metrics across hundreds of empirical datasets using the Item Response Warehouse, testing alternative hypotheses of underlying relationships without assuming fixed knowledge of the items.  Specifically, we ablate the role of the elliptical (cow) correlations in these relationships.

We show that traditional Pearson correlations—both product-moment and tetrachoric—produce misleadingly high estimates of unidimensionality, both in simulation and empirically, when using traditional measures of model fit. We also offer a highly interpretable, computationally simpler correlational alternative—nearly uniformly more performant in both simulation and across datasets: the quadrant correlation \citep{mosteller_useful_1946,blomqvist_measure_1950}. With new tools for evaluating assumptions and models and the reintroduction of a highly interpretable and powerful correlation, we both identify and solve a core challenge in determining the unidimensionality of an instrument via factor models.

\subsection{Unidimensionality, a rank--1 hypothesis, is rarely tested where it lives}
\label{sec:intro_rank1}
Many scientific questions reduce to clearly interpreting a matrix of observations. In psychometrics, $X_{ij}$ may record whether person $i$ endorsed item $j$; in education, whether a student solved a problem; in biomedicine, whether a specimen expresses a marker; and in modern AI benchmarking, whether a model succeeds on a task or prompt. In each case we seek a compact explanation of systematic variation in the response matrix
% \[
$X \in \mathbb{R}^{n\times p}$,
% \]
where $n$ indexes observational units (people, respondents, models, systems) and $p$ indexes variables (items, tasks, prompts, features). A common, interpretable simplification is that one latent attribute is enough: a single ability, severity, quality, or general factor explains most meaningful variation. This is the \emph{unidimensionality} premise. It is scientifically consequential because it justifies reporting a single score, ranking, or ordering.

At its core, a unidimensional claim is a statement about the most important signal of a \emph{data matrix}: %$X \approx \widehat X \quad \text{with} \quad \mathrm{rank}(\widehat X)=1$:
\begin{equation}
X \approx \widehat X \quad \text{with} \quad \mathrm{rank}(\widehat X)=1 \qquad 
\widehat X = u v^\top,\ u\in\mathbb{R}^{n},\ v\in\mathbb{R}^{p}.
\label{eq:rank1_intro}
\end{equation}
That is, the observed patterns in rows and columns are largely explainable by the singular relationship between them, e.g., items and the respondents. We can simply represent it as the outer product of a single row vector and a single column vector.

However, classical practice rarely evaluates Eq.~\eqref{eq:rank1_intro} directly. Instead, factor analysis and many unidimensionality diagnostics operate on an \emph{image} of the data--typically an association matrix such as a covariance or correlation matrix,
$A_c = \mathcal{A}_c(X)\in\mathbb{R}^{p\times p}
\; %\qquad 
\text{where } \mathcal{A}_c(X) \in \{X^\top X,\ \mathrm{cor}(X),\ \rho_{\text{tet}}(X),\dots \} \subset \mathfrak{A} $,
where $\mathfrak{A}$ is the set of all applicable candidate association relationships, and then assess whether a one-factor model fits that image well. This is sensible if the chosen association matrix $\mathcal{A}_c(X)$ truly captures the signal of interest.
This complex assumption needs to be tested as it can fail, especially for binary/ordinal data and in modern regimes with random items, random respondents, and irregular aspect ratios.

\subsection{Image-based fit can be self-confirming}
\label{sec:intro_selfconfirming}

Factor analysis is defined by a second-order dependence structure. Under a standard one-factor model,
\begin{equation}
X = \Lambda f^\top + E,\qquad f\in\mathbb{R}^{n},\ \Lambda\in\mathbb{R}^{p},\ E\in\mathbb{R}^{n\times p},
\label{eq:fa_intro}
\end{equation}
the implied association image takes the familiar form for covariance $\Sigma$:
\begin{equation}
\Sigma(X) = \Lambda\Phi\Lambda^\top + \Theta,
\label{eq:sigma_intro}
\end{equation}
with $\Phi=\mathrm{Var}(f)$ and (typically) diagonal $\Theta$. In practice, both estimation and evaluation are often performed in the same image space: we choose an association operator %, typically Pearson/$\phi$ or tetra/polychoric (others include, covariances, Yule's $Q$, Yule's $Y$, Spearman's $\rho$, etc.), 
estimate a one-factor structure on that image, and then judge adequacy by fit indices computed on the same kind of image (i.e., evaluation is based on $\Sigma(X)$ rather than $X$ itself).
\begin{wrapfigure}{R}{0.65\textwidth}
    \centering
    \includegraphics[width=0.64\textwidth]{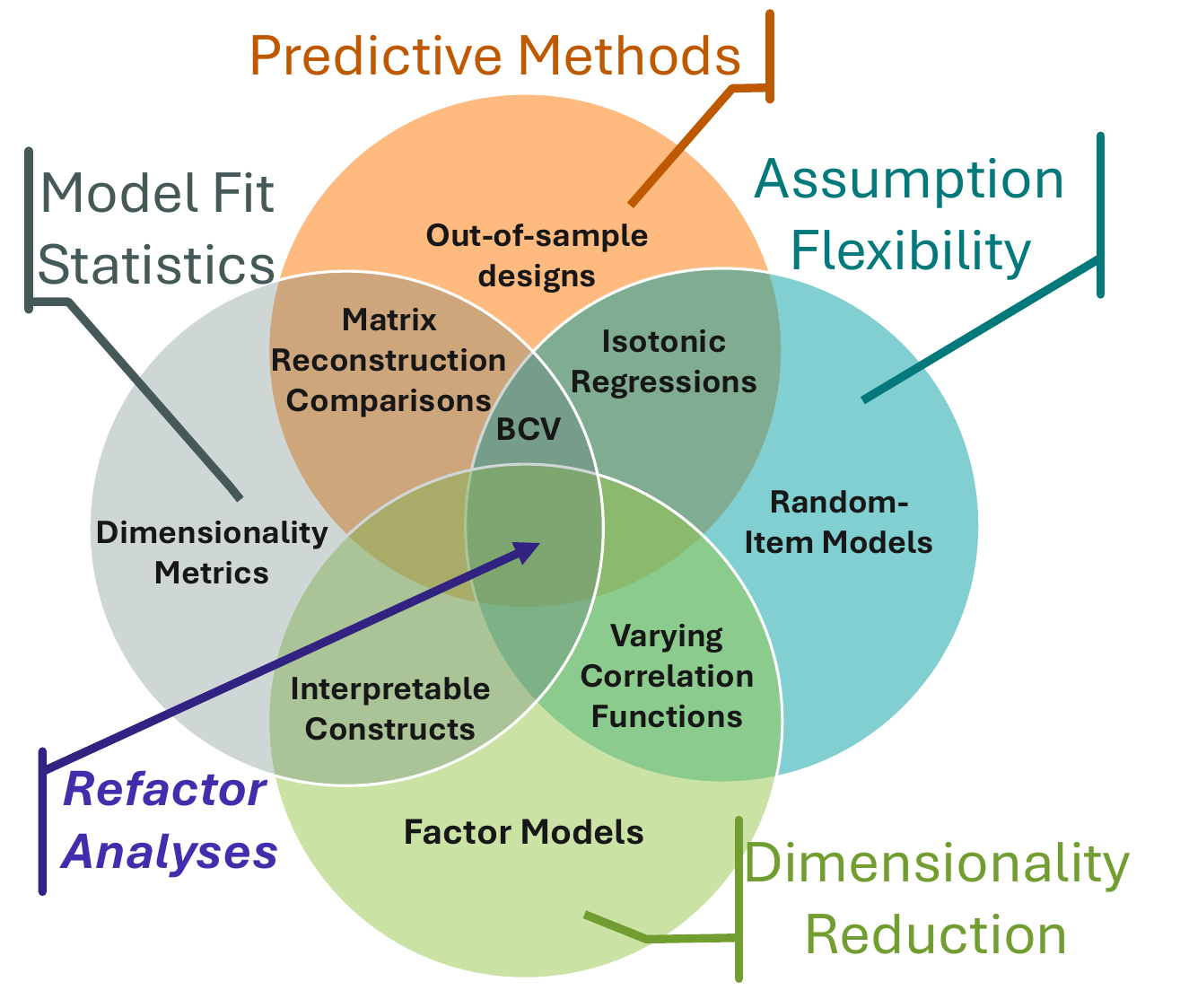}
    \caption{Refactor Analyses are useful for testing assumptions of factor models. Using a large number of datasets, we can test the general application of psychometric factor models.}
    \label{fig:refact_venn}
\end{wrapfigure}
For binary and ordinal data, ``correlation'' is a modeling choice rather than a fixed primitive. 
%Pearson/$\phi$ treats the $0\rightarrow 1$ step as a Euclidean distance; tetrachoric assumes thresholded latent normality; Yule's $Q$ emphasizes odds-ratio structure; Loevinger $H$ emphasizes monotone scalability; information-theoretic associations can emphasize uncertainty reduction. Each choice induces a different notion of rank--1 structure. 
Refactor/Verifactor make these choices empirically comparable by evaluating each induced rank--1 hypothesis on the same reconstruction task. Further, we illustrate that better signal recovery can be obtained by choice of correlational relationships. % \citep{hardy_star_corr_forthcoming}.

\paragraph{How  can this create a methodological vulnerability?} If the association operator itself imposes a particular notion of relationship, then high image-based fit can reflect internal coherence of the imposed relationship rather than faithful representation of the original response matrix. % 
To illustrate, let $\mathcal{A}\in\mathfrak{A}$ be a chosen association operator (e.g., $\phi$, tetrachoric, quadrant correlation) and let $A_c=\mathcal{A}_c(X)$ be the induced $p\times p$ item image. Standard image-based workflows (i) estimate a one-factor structure from $A_c$---for instance, by finding $\widehat v$ such that $A_c\approx \widehat v\,\widehat v^\top$---and then (ii) evaluate fit using statistics that are \emph{also functions of} $A_c$ (e.g., residual sums of squares on $A_c$, or indices derived from comparing $A_c$ to a model-implied image $\widehat\Sigma(\widehat v)$). Symbolically, this evaluates
% \[
$\text{fit}(\widehat v;\,X)\ \equiv\ \text{fit}\bigl(\widehat v;\,\mathcal{A}_c(X)\bigr)$,
% \]
rather than testing whether the rank--1 hypothesis $\widehat X=\widehat u\,\widehat v^\top$ actually recovers the response process $X$ in Eq.~\eqref{eq:rank1_intro} (formal treatment can be found in Appendix \ref{app:circularity}). When $\mathcal{A}$ is itself a modeling assumption---as it is for binary and ordinal data---this creates a risk of circularity: high apparent fit can reflect internal coherence of the \emph{imposed} association geometry (``does a one-factor model reproduce the image we constructed?'') rather than predictive fidelity to the original matrix (``does a one-factor representation recover the responses?''). Refactor breaks this loop by moving evaluation back to the data level: after estimating $\widehat u$ and $\widehat v$ from images, it judges adequacy by reconstruction quality $m(X,\widehat X)$ on the original matrix.

\begin{wrapfigure}{R}{0.56\textwidth}
    \centering
    \includegraphics[width=0.97\linewidth]{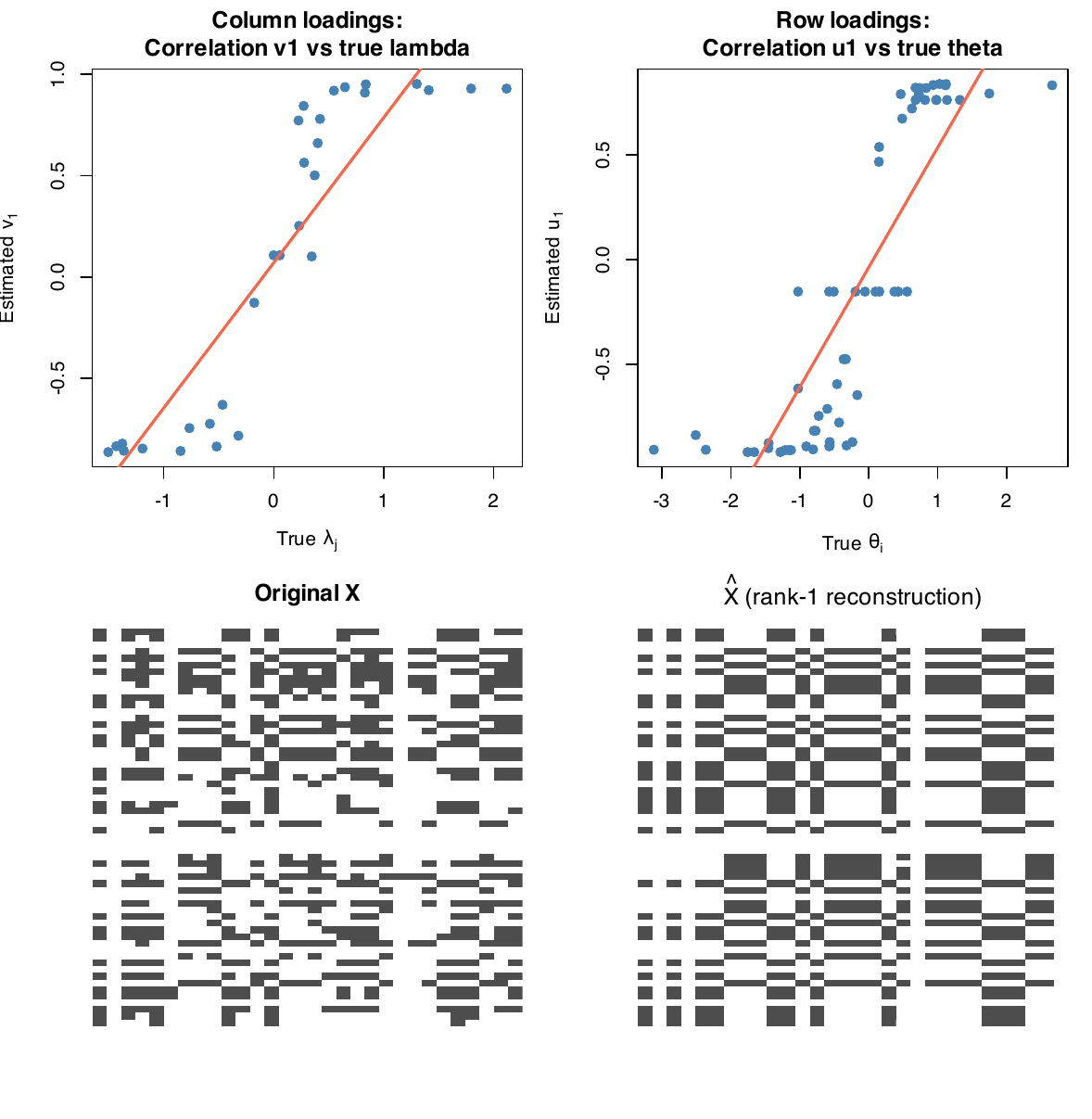}
    \caption{\textbf{Example Simulation}: \textbf{(top)} true unidimensional loadings for columns (left) and rows (right) vs their estimates, $\hat{v}$ and $\hat{u}$, respectively. \textbf{(bottom)} refactor rank--1 reconstruction where the data generating model reflects rank-1 tetrachoric correlations. The data and its reconstruction are compared $m(X,\hat{X})$ yielding a Refactor metric. While in this case the reconstruction has been transformed into binary for presentation, Refactor reconstructions are typically continuous: see Figure \ref{fig:reconstruction_diagram} for the continuous representations.}
    \label{fig:example_refactor_loadings}
\end{wrapfigure}

In short, we risk measuring how well a method reproduces its own assumptions. Solutions need to preserve the interpretive power of relationships between the rows and columns, which is why methods with uninterpretable factors, like Nonnegative Matrix Factorization \citep{lee_learning_1999}, that are optimized for specific ranks cannot fill this need in the measurement of latent constructs.\footnote{Methods like NMF and other less interpretable methods can, however, provide a helpful gage for more interpretable matrix reconstructions.} 

\subsection{The ``Refactor'' approach}
The {\bf Refactor} approach sits at the intersection of predictive model fit statistics and interpretable dimensionality reduction (see Fig. \ref{fig:refact_venn}). Refactor Analyses use the tools of random-item factor models to reconstruct response matrices and evaluate them based on the recoverability of the data. 

Refactor constructs a rank--$k$ representation from a chosen association image, but then evaluates it on the original response matrix. Concretely, as seen in Figure \ref{fig:example_refactor_loadings} for $k=1$ we obtain:
\begin{enumerate}\itemsep2pt
\item a column-side loading vector $\widehat v$ from an item-image $A_c=\mathcal{A}_c(X)$;
\item a row-side loading vector $\widehat u$ from an observation-image $A_r=\mathcal{A}_r(X)$;
\item a reconstructed matrix $\widehat X = \widehat u\,\widehat v^\top$;% (including any required calibration).
\item a measure of fit using metric that directly compares X and its reconstruction: $m(X,\hat{X})$
\end{enumerate}

\subsection{Refactor, Verifactor, and model evaluation}
\label{sec:intro_contributions}

A hallmark of a good scientific model is its ability to predict. This paper introduces two evaluation frameworks that shift dimensionality assessment from image fit to data recoverability and prediction via the above-described Refactor approach. Once we have constructed $\widehat X$, fit is then assessed by comparing $X$ and $\widehat X$. % using a \emph{suite} of reconstruction metrics that capture complementary notions of recovery: geometry (cosine), rank discrimination (AUC), ordinal agreement (Kendall $\tau$), calibration and information (cross-entropy / likelihood), and nonlinear dependence (distance correlation and partial distance correlation).

%\paragraph{(2) Verifactor analysis: refactoring under bi-cross-validation.}
Note that \textbf{Refactor} is in-sample as  it probes whether the hypothesized rank--1 structure can reproduce the data it was derived from. {\bf Verifactor} extends this to out-of-sample prediction using structured row and column held-out data, defined as bi-cross-validated (BCV) matrix partitions \citep{owen_bi-cross-validation_2009,owen_bi-cross-validation_2015}. By holding out blocks determined jointly by subsets of rows and columns, Verifactor targets the correct generalization question in crossed random designs: new respondents and new items simultaneously. This avoids the leakage and optimism that can occur when %factor scores or 
low-rank embeddings are computed using information from the same rows/columns being evaluated.
\begin{figure}[htb!]
    \centering
    \includegraphics[width=1\linewidth]{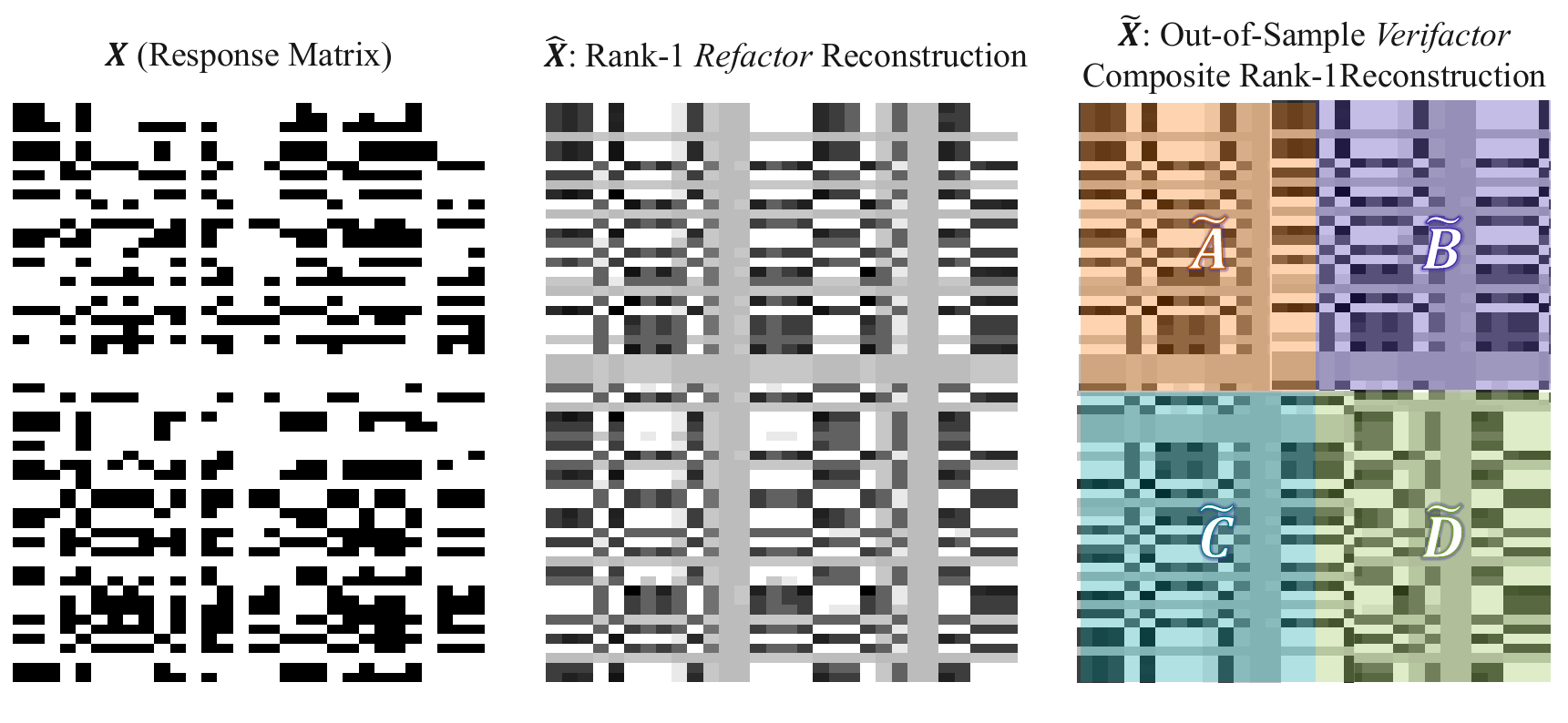}
    \caption{Example Response Matrix and Refactor and Verifactor Reconstructions}
    \label{fig:reconstruction_diagram}
\end{figure}
\vspace{-10pt}

\subsection{Implications for unidimensionality testing}
\label{sec:implications}

Refactor and Verifactor reframe unidimensionality as an empirical question of \emph{recoverability under a rank--1 representation}. This yields three concrete methodological implications that guide the empirical sections of the paper:

\begin{enumerate}\itemsep2pt
\item \textbf{From ``loading plausibility'' to ``data predictability'':} a one-factor model should be credited only insofar as it reconstructs $X$ (Refactor) and predicts held-out blocks of $X$ (Verifactor).
\item \textbf{Axis-respecting validation in crossed designs:} when both persons and items are random, Verifactor’s bi-cross-validation targets the correct generalization objective. %, avoiding entrywise CV artifacts.
\item \textbf{Correlational appropriateness is measurable:} by comparing association operators $\mathcal{A}$ through out-of-sample recoverability, we can diagnose when a correlational signal is present, absent, or when an alternative association better captures the intended construct.
\end{enumerate}

Refactor and Verifactor should not be seen as replacements for traditional factor-analytic diagnostics, but as complementary evidence: they directly test whether the hypothesized rank--1 structure is a faithful, predictive abstraction of the original data, and they do so in a way that is naturally compatible with modern random-by-random data regimes.

 \subsection{Outline}
 \label{sec:intro_roadmap}
First, we formalize the ``Refactor'' approach (Section \ref{sec:setup-refactor}) and establish Refactor and Verifactor as evaluation layers that sit atop standard factor-analytic estimators, respectively, in Sections \ref{sec:refactor_def} and \ref{sec:verifactor_def}, providing a more formal treatment in Appendices \ref{app:refactor-proof} and \ref{app:verifactor-proof}. In Section \ref{sec:test_corr}, we posit that these methods can compare correlational relationships used in factor analysis, and in Section \ref{sec:q}, we reintroduce a correlational coefficient $q^\prime$ that we will use as a comparison to Pearson and tetrachoric.  In Section \ref{sec:evalmeths}, we  introduce reconstruction metrics, emphasizing dual random-effects regimes. Multiple simulations ground the concepts by contrasting known data-generating mechanisms and association assumptions In the first simulations, our data generating model matches traditional correlations (Section \ref{sec:simple_sim}) and our second simulations we replicate a hierarchical factor structure from a recent study on unidimensionality (Section \ref{sec:unidim_reprod}). Finally, we evaluate over a large and diverse collection of public datasets, showing that traditional image-based unidimensionality measures often have weak relationship to the ability of the corresponding rank--1 model to reconstruct and predict the original response matrix (Section \ref{sec:empirical}). This motivates a re-interpretation of ``unidimensionality'' in modern scientific and benchmarking contexts: not merely as a property of an association image, but as a claim about recoverable structure in the data itself. Section \ref{sec:background}  acknowledges the long history of factor models and this study's position within it. 

% The long history of factor models (Section \ref{sec:background}) . We first establish Refactor and Verifactor as evaluation layers that sit atop standard factor-analytic estimators respectively in Sections \ref{sec:refactor_def} and \ref{sec:verifactor_def}, providing more formal treatment in Appendices \ref{sec:refactor-proof} and \ref{sec:verifactor-proof}. In Section \label{sec:test_corr} we posit that these methods can compare correlational relationships used in factor analysis and in Section \ref{sec:q} we reintroduce a correlational coefficient $q^\prime$ which we will use as a comparison to Pearson and tetrachoric.  We then introduce the reconstruction metric suite and its interpretation, emphasizing dual random-effects regimes. Simulations ground the concepts by contrasting known data-generating mechanisms and association assumptions. Finally, we evaluate over a large and diverse collection of public datasets, showing that traditional image-based unidimensionality measures often have weak relationship to the ability of the corresponding rank--1 model to reconstruct and predict the original response matrix. This motivates a re-interpretation of ``unidimensionality'' in modern scientific and benchmarking contexts: not merely as a property of an association image, but as a claim about recoverable structure in the data itself.
\begin{figure}[ht]
    \centering
    \includegraphics[width=1\linewidth]{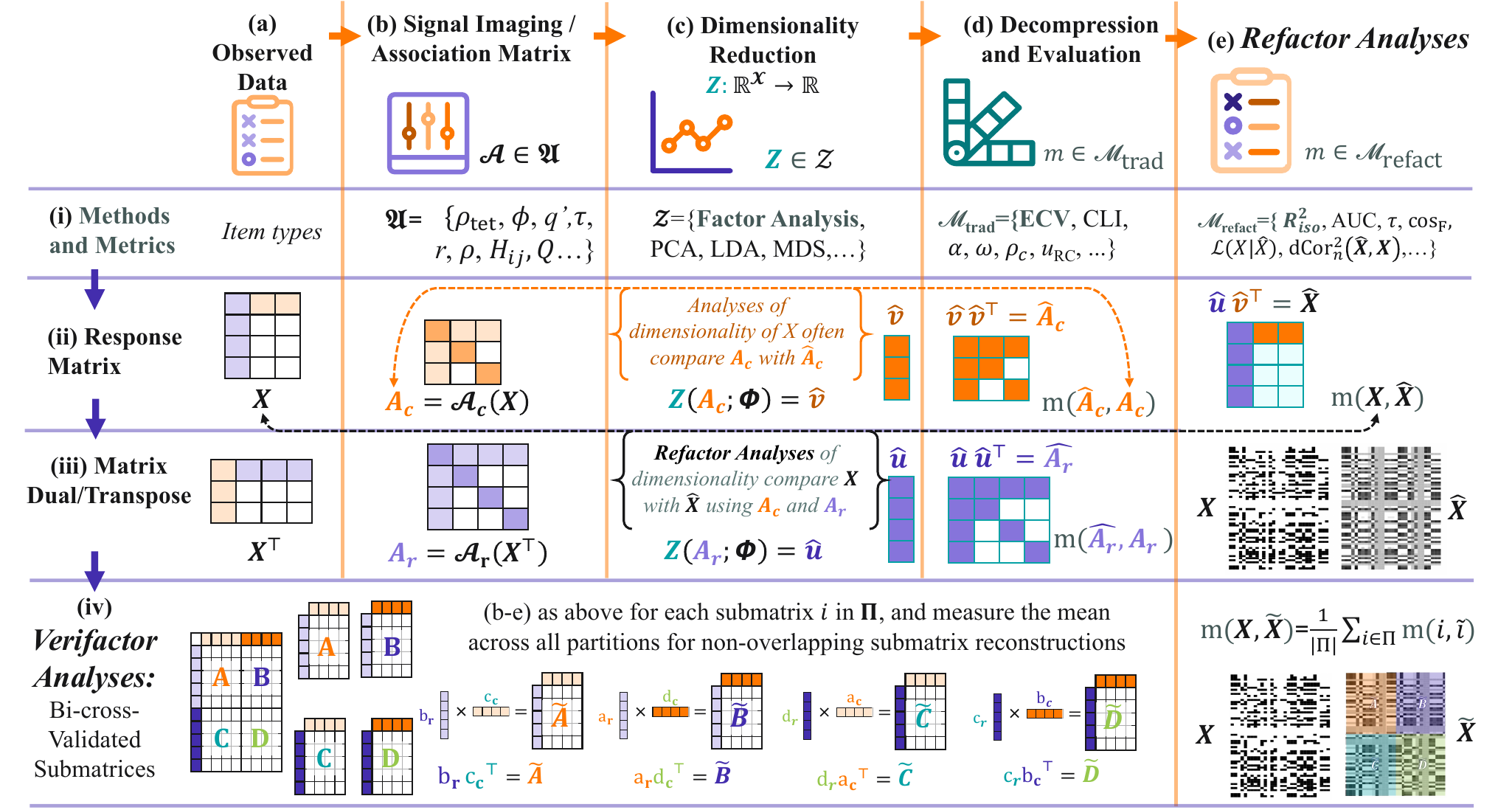}
    \caption{\textbf{Refactor and Verifactor Workflows}. \textbf{(a)} Starting from the observed response matrix $X \in \mathbb{R}^{n\times p}$, we \textbf{(b)} form an association matrix "image" $A$ (see Section \ref{sec:test_corr}) by capturing the signal of interest on both axes (see Section \ref{sec:setup-refactor}): \textbf{(ii)} $A_c$ (columns) and \textbf{(iii)} $A_r$ (rows). \textbf{(c)} Standard dimensionality reduction techniques $\textsf{Z}$ focus on this signal image derived from $\boldsymbol{X}$ (e.g., the covariance/correlation matrix $\boldsymbol{X}^T\boldsymbol{X}$) to produce column-space projection loadings, $\boldsymbol{\hat{v}}$. Refactoring extends this by performing a dual analysis on the matrix transpose to produce row-space projection loadings, $\boldsymbol{\hat{u}}$. \textbf{(d)} These two loading matrices are then used to reconstruct a prediction of the original data matrix, $\boldsymbol{\hat{u}}\boldsymbol{\hat{v}}^\top= \boldsymbol{\hat{X}}$ (see Section \ref{sec:refactor_def}. Finally, \textbf{(e)} Refactor Analyses evaluate the model by quantifying the correspondence between the observed data $\boldsymbol{X}$ and the refactored data $\boldsymbol{\hat{X}}$ using various \textbf{(i)} matrix comparison metrics, thereby assessing the model's ability to preserve the signal in the original data. Verifactor Analyses \textbf{(iv)} extends this paradigm to out-of-sample bi-cross-validated prediction by using limited information projections calculated from individual partitioned submatrices, $i \in \Pi$, of \textbf{X} to reconstruct low-rank approximations for held-out  submatrices (see Section \ref{sec:verifactor_def}).}
    \label{fig:refactor}
\end{figure}
\section{Methods}\label{sec:methods}

\subsection{Refactor and Verifactor leverage dual random effects and recoverability metrics}
\label{sec:setup-refactor}

Let $\mathcal{A}($X$)$ be an \emph{association operator} mapping $X$ to a symmetric image on rows or columns.\footnote{Two canonical choices are Gram matrices (square matrices that capture association by calculating an inner product of every possible pair)
\[
G_r(X)=XX^\top\in\mathbb{R}^{n\times n},\qquad G_c(X)=X^\top X\in\mathbb{R}^{p\times p},
\]
but in applications we also allow $\mathcal{A}$ to be constructed entrywise from ``Gram-functioning'' pairwise associations (e.g., a correlation matrix, a $\kappa$ matrix, etc.).} For concreteness we write the two images as
% \[
$A_r=\mathcal{A}_r(X)\in\mathbb{R}^{n\times n}, %\qquad 
A_c=\mathcal{A}_c(X)\in\mathbb{R}^{p\times p}$.
% \]
Let $\textsf{Z}($X$,k)$ be a dimensionality-reduction operator that returns the top-1\footnote{This can be extended to top-$k$ without loss of generality. Applying $\textsf{Z}$ to each image yields \emph{row-space} and \emph{column-space} loadings
\begin{equation}
B_k = \textsf{Z}(A_r,k)\in\mathbb{R}^{n\times k},\qquad
F_k = \textsf{Z}(A_c,k)\in\mathbb{R}^{p\times k}.
\label{eq:rowcol-loadings-footnote}
\end{equation} See \cite{owen_bi-cross-validation_2015} for further treatment of top-$k$ rank selection.} directions (e.g., eigendecomposition, minimum residual FA, robust variants). Applying $\textsf{Z}$ to each image yields \emph{row-space} and \emph{column-space} loadings
\begin{equation}
\hat{u} = B_1 =\textsf{Z}(A_r,1)\in\mathbb{R}^{n\times 1},\qquad
\hat{v} = F_1 = \textsf{Z}(A_c,1)\in\mathbb{R}^{p\times 1}.
\label{eq:rowcol-loadings}
\end{equation}
A \emph{refactoring map} is any rule $\mathcal{R}_k$ that produces a reconstructed matrix
\begin{equation}
\widehat{X_k} \;=\; \mathcal{R}_k(X; B_k,F_k)\in\mathbb{R}^{n\times p},
\label{eq:refactor-recon}
\end{equation}
intended to approximate $X$ if the hypothesized low-rank structure is correct, where $k = 1$ for rank-1.  

The construction of $\widehat{x_k}$ is the core of the Refactor approach. To summarize,
\begin{itemize} 
\item $X$ is first mapped to row- and column-based symmetric images via $\mathcal{A}$.
\item These images are then each reduced to a $1$-dimensional approximation via $\textsf{Z}$.
\item These approximations are combined to form $\widehat{X_k}$ using a refactoring map $\mathcal{R}_k$. 
\end{itemize}
After computation of $\widehat{X_k} $, Refactor analysis involves evaluation of the resemblance between $X$ and $\widehat{X_k}$; note that we are not using $\mathcal{A}_r$ or $\mathcal{A}_c$ for comparison but rather focusing on $X$ and its refactoring $\widehat{X_k}$ (see Figure \ref{fig:refactor}). To analyze resemblance, we consider some comparison metric $m_j(X,\widehat X_k)$ (Section~\ref{app:eval_metrics} and Appendix Table~\ref{tab:metrics}). 

% \subsubsection{Measuring Refactor Recovery}

% [what you now want to discuss is the choice of m\_j]. 

%In the special case where $B_k$ and $F_k$ are taken as the left and right singular vectors of $X$ (or equivalently, principal directions of $XX^\top$ and $X^\top X$), the natural refactoring map is the truncated SVD reconstruction $\widehat X_k = U_k\Sigma_k V_k^\top$. In more general settings (nonlinear associations, FA constraints, heteroscedasticity, sign/scale indeterminacy, and monotone response models), $\mathcal{R}_k$ is defined to respect the model class being evaluated (e.g., isotonic calibration of scores/predictions when only ordinal information is intended to be modeled).

% \subsection{Measuring Refactor Recovery}\label{sec:refact_measure}
The final step of Refactoring Analyses is to quantify the correspondence between the observed data $\boldsymbol{X}$ and its reconstruction $\boldsymbol{\hat{X}}$ (see Section \ref{sec:evalmeths}). There are many possible choices for $m_j$, and thus we compare them to extant metrics. This comparison provides a direct and interpretable measure of how well the low-rank structure captures the information in the original data. Because this task requires methods for comparing matrices that may be unfamiliar to many researchers, we introduce a suite of powerful metrics. Moving beyond simple element-wise correlations, these include global measures of matrix association, discussed at length\footnote{We would argue that a contribution of this work is the intentional selection of a diagnostic suite of metrics for evaluating the reconstruction. A single summary statistic may be insufficient to assess the correspondence between the observed data matrix $\boldsymbol{X}$ and its model-based reconstruction $\boldsymbol{\hat{X}}$. Therefore, in the appendix, we offer a multifaceted evaluation that provides a holistic profile of the model's performance. } in Appendix \ref{app:circularity}. In the main body of the paper we will show focus on one Refactor metric critical for evaluating the effects of varying association matrices on measures of unidimensionality: isotonic $R^2$ (see Section \ref{sec:methods_isotonic_r2}) which is the maximum proportion of variance explained, assuming a monotone unidimensional construct. We illustrate the evaluation processes below.

\subsection{Refactor Analyses}\label{sec:refactor_def}
% This section explains the core contributions of 
Refactor Analysis is a \emph{test of recoverability of $X$ from a rank--$k$ representation learned from an image}, rather than a test of fit on the image itself. A mathematically formal treatment of this section can be found in Appendix \ref{app:refactor-proof}. 

Let $\boldsymbol{X} \in \mathbb{R}^{n \times p}$ be the data matrix. A rank-1 approximation of $\boldsymbol{X}$ takes the form $\boldsymbol{\hat{X}} = \boldsymbol{u}\boldsymbol{v}^T$, where $\boldsymbol{u} \in \mathbb{R}^n$ and $\boldsymbol{v} \in \mathbb{R}^p$.  In this case, the vector $\boldsymbol{u}_1$ represents the first factor loadings based on the rows (e.g., persons) from the fitted factor model of $\boldsymbol{X}^T$. The vector $\boldsymbol{v}_1$ represents the first factor loadings based on the columns (e.g., items)from the fitted factor model of $\boldsymbol{X}$. The Refactor reconstruction, $\boldsymbol{\hat{X}}$, is the rank-1 matrix formed by the outer product of these two vectors.

This process allows for the use of various association matrices. Instead of $\boldsymbol{X}^T\boldsymbol{X}$, one can construct a matrix $\boldsymbol{A}$ whose entries $A_{ij}$ represent a chosen measure of association (e.g., Pearson, tetrachoric, quadrant correlation) between columns $i$ and $j$ of $\boldsymbol{X}$. The first factor loadings of $\boldsymbol{A}$ serves as the loading vector $\boldsymbol{v}$. $\boldsymbol{u}$ follows similarly.

The claim of unidimensionality, that a rank-1 matrix fits the data, is simultaneously a claim about \emph{row structure} (persons/observations align along one latent direction) and \emph{column structure} (items/features align along one latent direction). Image-only diagnostics often quantify only the column-side claim (e.g., via $X^\top X$), \textbf{\textit{while Refactor insists on reconstructing $X$ using both}}. 

%\subsubsection{Refactor fit as recoverability}
Let $\mathcal{M}_k$ denote a class of rank--$k$ reconstructions induced by a chosen image, estimator, and refactoring map. Define a \emph{Refactor functional} by
\begin{equation}
\mathrm{RF}_m(k;X)\;=\; m(X,\widehat X_k),
\qquad \widehat X_k\in\mathcal{M}_k,
\label{eq:refactor-functional}
\end{equation}
where $m(\cdot,\cdot)$ may be a loss (smaller is better) or an association (larger is better). % The paper instantiates $m$ with predictive (AUC, cross-entropy/likelihood), order-preserving (Kendall $\tau$), geometric (matrix cosine), and dependence-sensitive (bias-corrected squared distance correlation, independence-partialized squared distance correlation).
When the data are truly well described by a rank--$1$ signal, Refactor scores approach their optimal values. Conversely, when Refactor scores remain poor even as image-based fit looks adequate, the evidence points to a mismatch between the \emph{associational image} and the \emph{data-level signal} (e.g., nonlinear structure not captured by the chosen association, violations of conditional independence, mixtures, or strong idiosyncratic effects).

\subsection{Verifactor Analyses}\label{sec:verifactor_def}
% This section explains the core contributions of 
Verifactor Analysis is a \emph{test of recoverability of $X$ from a held-out rank--$k$ representations learned from an image}, rather than a test of fit on the image itself. The formal treatment of this section can be found in Appendix \ref{app:verifactor-proof}. 

Response matrices operate as 2-D samples indexed by observations (respondents) and variables (items), sampled where the scientific target concerns generalization over \emph{both} axes. Many data imputation methods exist for matrix completion, but each may require different assumptions on the underlying relationship we hope to investigate. In such matrix designs, resampling individual entries (or naive elementwise masking CV) breaks the true dependence structure induced by shared rows/columns and yields optimistic or otherwise misleading error estimates.\footnote{Solutions such as using the expected log predictive likelihood for missing respondents or items can be calculated exactly using the marginal maximum likelihood estimation \citep{casabianca_irt_2015} \emph{if the data generating model is known} and at least one set of variables are considered fixed \citep{stenhaug_predictive_2022}.} 
When the generating model is \textit{not known}, or when such assumed relationships are being tested in dual random effects contexts, resampling rows and columns is the appropriate asymptotic approximation \citep{mccullagh_resampling_2000}. BCV \citep{owen_bi-cross-validation_2009,owen_bi-cross-validation_2015} accomplishes this matrix prediction by holding out submatrices determined jointly by subsets of rows and columns.

%\subsubsection{Verifactor Analysis as BCV}
Partition (after row/column permutation) a matrix $X$ into blocks
\begin{equation}
X=
\begin{pmatrix}
A & B\\
C & D
\end{pmatrix} \qquad \text{and} \qquad \widetilde{X}=
\begin{pmatrix}
\widetilde{A} & \widetilde{B}\\
\widetilde{C} & \widetilde{D}
\end{pmatrix},
\label{eq:block-partition}
\end{equation}
where $\widetilde{X}$ is the Verifactor reconstruction of $X$ and where each submatrix is reconstructed based on the factor loadings of the submatrices belonging to the other diagonal. 

 In the case of reconstructing $\widetilde{A}$, 
 $\boldsymbol{\widetilde{A}} = \boldsymbol{u}_{B,1}\boldsymbol{v}_{C,1}^\top$. In this case, the vector $\boldsymbol{u}_{B,1}$ represents the first factor loadings based on the rows (e.g., persons) from the fitted factor model of $\boldsymbol{B}^\top$. The vector $\boldsymbol{v}_{C,1}$ represents the first factor loadings based on the columns (e.g., items)from the fitted factor model of $\boldsymbol{C}$. All submatrix reconstructions are created similarly.

%\subsubsection{Verifactor as out-of-sample recoverability}

Let $\Pi$ denote a random partition of rows and columns into folds. For each fold $(i,j)$, let $(A_{ij},B_{ij},C_{ij},D_{ij})$ be the corresponding block decomposition. Verifactor computes loadings (or low-rank structure) on held-in blocks, constructs a refactored predictor $\widehat A_{ij}$ for the held-out block, and aggregates fit:
\begin{equation}
\mathrm{VF}_m(k;X)
=
\frac{1}{|\Pi|}
\sum_{(i,j)\in\Pi}
m\!\left(A_{ij},\ \widetilde A_{ij}^{(k)}\right),
\qquad
\widetilde A_{ij}^{(k)} := \mathcal{V}_k(B_{ij},C_{ij},D_{ij}),
\label{eq:verifactor-functional}
\end{equation}
where $\mathcal{V}_k$ is a BCV-compatible predictor or a model-specific analog using estimated row/column loadings from $D$ via $B$ and $C$). In our analysis, rank is fixed at $k=1$ when the inferential target is unidimensionality, which avoids known monotonicity pathologies of one-way deletion CV for matrix factorization while preserving the interpretability of the hypothesis test.

Verifactor is especially well aligned with the random-rows/random-columns inferential target: it evaluates the \emph{generalizable} part of the rank--$k$ structure while cleanly separating irreducible noise. This stands in contrast to in-sample reconstruction, where overfitting and axis-specific dependence can inflate fit. Importantly for unidimensionality, Verifactor analysis can illustrate the extent to which the hypothesized relationship is truly unidimensional. Interestingly, unidimensionality as measured by traditional metrics, based on a covariance or correlation matrix, has no significant empirical relationship with this more honest measure of dimensionality (see Figure \ref{fig:fafit_vs_refact}).

\subsubsection{Verifactor prediction for two-way random-effects}
\label{sec:why_matrix_prediction_main}

In random observations $\times$ random variables regimes \citep{de_boeck_random_2008}, the scientific question concerns generalization to new rows and new columns. This makes the \emph{unit of resampling} (and thus the unit of evaluation) crucial. Entrywise resampling or prediction is generally misaligned because each entry shares its row and column with observed entries, enabling optimistic leakage.
% \subsubsection{Axis-respecting generalization and leakage control}
% \label{sec:axis_respecting}

Let $\Pi=(\mathcal{I},\mathcal{J})$ denote a fold, with held-out rows $\mathcal{I}\subset[n]$ and held-out columns $\mathcal{J}\subset[p]$. The held-out block is
$A=X_{\mathcal{I},\mathcal{J}}$. Verifactor constructs $\widehat A$ from
$D=X_{\mathcal{I}^c,\mathcal{J}^c}$ (and, depending on the completion rule, also $B=X_{\mathcal{I},\mathcal{J}^c}$ and $C=X_{\mathcal{I}^c,\mathcal{J}}$) but \emph{never} from rows in $\mathcal{I}$ together with columns in $\mathcal{J}$. This is the one sense in which Verifactor is an improvement over standard reconstruction assessments: it aligns evaluation with the intended generalization. 

\subsection{Evaluation Methods}\label{sec:evalmeths}
In this section, we discuss Refactor reconstruction recovery methods and metrics. Formal treatment of the content can be found in Appendix \ref{app:evaluation_motivation}.
\subsubsection{Refactor Analysis and Recoverability of Data}\label{sec:recoverability}

Refactoring introduces a new paradigm in factor model evaluation by refocusing the assessment from the model's abstract image to its direct explanatory power. The central premise, outlined in Figure \ref{fig:refactor}, is to leverage the full structure of the low-rank model to generate a prediction of the original data. 

The final step of Refactoring Analyses is to quantify the correspondence between the observed data $\boldsymbol{X}$ and the Refactored and Verifactored reconstructions, $\boldsymbol{\hat{X}}$ and $\boldsymbol{\widetilde{X}}$, respectively. 

This comparison provides a direct and interpretable measure of how well the low-rank structure captures the information in the original data, illustrated with an example in Figure \ref{fig:example_refactor_loadings}. Because this task requires methods for comparing matrices that may be unfamiliar to many researchers, we introduce a suite of powerful metrics. % Moving beyond simple element-wise correlations, these include global measures of matrix association, such as distance correlation and the RV coefficient, as well as more granular diagnostics that assess structural fidelity, such as the preservation of vector orientation as measured by cosine similarity. Together, these tools provide a richer, more diagnostically informative framework for evaluating model adequacy than is possible by inspecting the factor loadings alone.

\subsection{Spotlighted comparative evaluation metrics: ECV and isotonic $R^2$}
\label{sec:methods_primary_metrics}

The main figures are designed to make a single comparison transparent: \emph{when do traditional, image-based claims of unidimensionality agree with data-level recoverability under Refactor, and when do they diverge?} To keep this comparison interpretable across audiences and across datasets, we focus on metrics that share a common conceptual target: \textbf{how much of the systematic signal can be explained by a single dimension}. We therefore report (i) a standard unidimensionality index computed on the association image, \emph{explained common variance} (ECV), and (ii) a data-level reconstruction index computed on $X$ and $\widehat X$, \emph{isotonic $R^2$}. Both can be understood as proportion-of-variance style summaries, but they answer different questions.

\subsubsection{Traditional image-based unidimensionality: explained common variance (ECV)}
\label{sec:methods_ecv}

% \paragraph{What ECV measures.}
ECV quantifies the degree to which a \emph{common-factor} representation of the association image is dominated by the first factor. Intuitively, if the shared covariance among items is essentially one-dimensional, then the first common factor should account for most of the shared variance, and remaining common factors should be comparatively small.

% \paragraph{How ECV is computed.}
Let $A_c=\mathcal{A}_c(X)\in\mathbb{R}^{p\times p}$ be a chosen item association matrix (e.g., $\phi$ or tetrachoric correlations), and let a factor-analytic method produce an $m$-factor common-variance decomposition of the form
% \begin{equation}
$A_c \approx \Lambda\Lambda^\top + \Psi$,
% \label{eq:common_model}
% \end{equation}
where $\Lambda\in\mathbb{R}^{p\times m}$ contains factor loadings and $\Psi$ is a diagonal uniqueness matrix. ECV is then computed from the eigenvalues of the \emph{common variance matrix} $\Lambda\Lambda^\top$.\footnote{In practice we use minimum rank factor analysis (MRFA), which is designed to estimate the eigenvalues of common variance while minimizing rank under appropriate constraints \citep{ten_berge_numerical_1991,ten_berge_greatest_2004}.}
Let $\lambda_1\geq \lambda_2\geq \dots\geq \lambda_m>0$ denote the eigenvalues of $\Lambda\Lambda^\top$ (equivalently, the squared singular values of $\Lambda$). The explained common variance is
\begin{equation}
\mathrm{ECV}
\;:=\;
\frac{\lambda_1}{\sum_{r=1}^{m}\lambda_r}.
\label{eq:ecv}
\end{equation}
% \paragraph{Interpretation.}
$\mathrm{ECV}\in(0,1]$ is large when the first common factor dominates the remaining common factors. ECV is widely used as a unidimensionality indicator because it is relatively robust to the total number of items compared to older reliability heuristics, and it aligns with the conceptual claim that one latent dimension explains the shared signal \citep{rodriguez_evaluating_2016,sijtsma_use_2009, ten_berge_greatest_2004}. However, ECV remains an \emph{image-based} statistic: it is computed from a transformation $\mathcal{A}_c(X)$ and reflects the one-dimensionality of that image’s common variance, not necessarily the recoverability of the original response matrix.

\subsubsection{Refactor data-level unidimensionality and need for monotone metric: isotonic $R^2$}
\label{sec:methods_isotonic_r2}

% \paragraph{Need for monotone metric}
For binary and ordinal response matrices, a one-dimensional latent trait typically induces a \emph{monotone} relationship between the latent score and response probability, but not a linear one (e.g., thresholding, logistic/probit links, saturation). Therefore, evaluating $\widehat X$ by linear correlation or mean squared error can unfairly penalize a correct unidimensional structure when the link from latent signal to observed responses is nonlinear. To make the reconstruction evaluation compatible with the general measurement assumption of monotonicity, we use isotonic regression to compute the best possible monotone transformation of $\widehat X$ before measuring explained variance.

% \paragraph{Isotonic $R^2$.}
Here we define Isotonic $R^2$. Let $\widehat X\in\mathbb{R}^{n\times p}$ be a Refactor rank--1 reconstruction from some association operator and estimation procedure. Write $x=\mathrm{vec}(X)\in\mathbb{R}^{np}$ and $\hat{x}=\mathrm{vec}(\widehat X)\in\mathbb{R}^{np}$. Consider the class $\mathcal{G}$ of nondecreasing functions $g:\mathbb{R}\to\mathbb{R}$. Define the isotonic (best monotone) fit
\begin{equation}
g^\star
\;\in\;
\arg\min_{g\in\mathcal{G}}
\sum_{t=1}^{np}\bigl(x_t - g(\hat{x}_t)\bigr)^2,
\label{eq:isotonic_fit}
\end{equation}
and set $\tilde{x} = g^\star(\hat{x})$. The isotonic coefficient of determination is then
\begin{equation}
R^2_{\mathrm{iso}}(X,\widehat X)
\;:=\;
1-
\frac{\sum_{t=1}^{np}\bigl(x_t-\tilde{x}_t\bigr)^2}
{\sum_{t=1}^{np}\bigl(x_t-\bar{x}\bigr)^2},
\qquad \bar{x}=\frac{1}{np}\sum_{t=1}^{np}x_t.
\label{eq:isotonic_r2}
\end{equation}

% \paragraph{Interpretation.}
$R^2_{\mathrm{iso}}$ can be read as: \emph{how much variance in the observed responses can be explained by the reconstruction after allowing the best monotone calibration of its scale}. This makes it a natural ``recoverable signal'' index for unidimensional latent structure in binary/ordinal data: if $\widehat X$ correctly orders response propensities, isotonic regression will map that ordering to an optimal monotone approximation of $X$.

\subsubsection{Isotonic $R^2$ is the optimal monotone variance-explained score}
\label{sec:methods_optimality}

\begin{proposition}[Optimality of isotonic calibration for monotone fit]
\label{prop:isotonic_optimality}
Among all monotone transformations $g\in\mathcal{G}$ applied entrywise to $\widehat X$, isotonic regression achieves the minimal residual sum of squares in \eqref{eq:isotonic_fit}. Consequently, $R^2_{\mathrm{iso}}(X,\widehat X)$ is the \emph{largest} achievable $R^2$ obtainable from $\widehat X$ under the sole assumption that the relationship between latent signal and observed responses is monotone.
\end{proposition}

\begin{proof}
Equation \eqref{eq:isotonic_fit} is the least-squares projection of $x$ onto the closed convex cone\footnote{Isotonic regression is considered a cone regression because the set of all possible non-decreasing functions (the constraint space) forms a convex cone in vector space, represented here. Isotonic regression is a specific instance of a broader class of cone regression problems, which minimizes the distance (least squares) between observed data and a closed convex cone. In other words, finding the isotonic fit is equivalent to taking raw data points and ``projecting'' them onto the nearest point within the isotonic cone. See \cite{yang_contraction_2019}.} $\{g(\hat{x}): g\in\mathcal{G}\}$. Standard results for isotonic regression imply existence and (up to ties) uniqueness of the minimizer and its optimality in squared error. Since $R^2_{\mathrm{iso}}$ is a monotone transformation of the residual sum of squares, minimizing the residual maximizes $R^2_{\mathrm{iso}}$.
\end{proof}

\subsubsection{ECV and isotonic $R^2$ together clarify convergence versus divergence}
\label{sec:methods_why_these}

ECV and isotonic $R^2$ are intentionally paired because they operationalize two different, commonly conflated notions of unidimensionality:

\begin{itemize}\itemsep2pt
\item \textbf{ECV (image-based):} ``Is the \emph{common variance} of the association image essentially one-dimensional?''
\item \textbf{Isotonic $R^2$ (data-based):} ``Does the induced rank--1 structure \emph{recover the response matrix} up to an unknown monotone link?''
\end{itemize}

When the association operator matches the data-generating dependence (e.g., a genuinely correlational latent mechanism), we expect these measures to agree: high ECV should coincide with high isotonic $R^2$. When they diverge, the difference is informative: it indicates that a one-factor structure may fit the chosen image well while failing to recover the original responses (or vice versa), motivating further investigation of (i) the appropriateness of the association operator, (ii) marginal confounds, or (iii) non-correlational structure not captured by rank--1 images.

\subsection{Correlational Relationships as Testable}\label{sec:test_corr}
One benefit of the Refactor/Verifactor framing is that it treats the \emph{choice of association matrix} as a scientifically meaningful modeling decision. \emph{``Is a correlational relationship the right signal/abstraction for this dataset? Which correlational notion or relationship is most appropriate for explaining shared signal?''} These questions are central in applied psychometrics (binary/ordinal items; rater-like features) and in machine learning (implicit feedback, pairwise preferences, weak supervision), yet is not addressed by classical image-only fit statistics because they condition on a single, often default, association choice of standard correlation. As an example, a Pearson correlation with for binary data treats the distance between 0 and 1 as a meaningful distance; we ask whether that is a meaningful assumption when representing the latent relationship. 

\subsubsection{Recoverability of the response matrix under a rank--1 associative hypothesis}
% \label{sec:evaluation_target}

Let $X\in\mathbb{R}^{n\times p}$ be observed responses (possibly binary/ordinal). For a fixed association operator $\mathcal{A}$ and a fixed rank $k=1$, Refactor/Verifactor define an estimand that is more primitive than ``fit to correlations'': the extent to which a \emph{rank--1 associative signal} supports \emph{matrix prediction} of $X$.

Formally, an association operator $\mathcal{A}$ induces a hypothesis class of reconstructions
% \[
$\mathcal{H}_{\mathcal{A}}
=
\Bigl\{
\widehat X = \mathcal{R}_1(X;B_1,F_1):
(B_1,F_1) = \textsf{Z}(\mathcal{A}_r(X),\mathcal{A}_c(X))
\Bigr\}$,
% \]
and Verifactor evaluates the \emph{out-of-sample risk}
\begin{equation}
\mathcal{R}_{\mathcal{A}}(m)
\;:=\;
\mathbb{E}_{\Pi}\,\mathbb{E}\Bigl[
m\bigl(A_{\Pi},\widehat A_{\Pi}^{(\mathcal{A})}\bigr)
\Bigr],
\label{eq:assoc-risk-main}
\end{equation}
where $A_{\Pi}$ is the held-out block under a random row/column partition $\Pi$ and
$\widehat A_{\Pi}^{(\mathcal{A})}$ is predicted from held-in blocks using the rank--1 structure implied by $\mathcal{A}$. This creates a common evaluation currency across (i) association choices, (ii) estimators $\textsf{Z}$, and (iii) reconstruction metrics $m$.

% \paragraph{Interpretation.}
Equation~\eqref{eq:assoc-risk-main} turns the assumption of correlation as the implied relationship into a falsifiable claim: the correct associative abstraction is the one that supports stable prediction of held-out entries when both rows and columns are treated as random. This target is distinct from, and not implied by, image fit indices because those indices condition on the association image as sufficient statistics.

% In applied terms, Proposition~\ref{prop:non-equivalence} formalizes the practical warning suggested by Figures~\ref{fig:fafit_vs_verifact}--\ref{fig:fafit_vs_verifact}: \emph{``good''} CFI/TLI/$u$/$\rho_c$ does not imply that a rank--1 correlational mechanism usefully explains (or predicts) the original responses.

\subsection{Recommended Interpretable Correlational Relationships: $q^\prime$}\label{sec:q}
In this paper, we will test the relationships used in standard factor analysis against their ability to reconstruct the signal being represented. We also offer another correlational relationship that preserves interpretability while improving reconstructive power in the unidimensional setting: $q^\prime$ \citep{mosteller_useful_1946,blomqvist_measure_1950}.\footnote{We use $q^\prime$ as the variable symbol, originating with \cite{blomqvist_measure_1950}, who did the original work to understand the statistical properties of this relationship.} To simplify the many potential notational approaches to this relationship across many sources and re-discoveries, we define $q^\prime$ for two binary variables, $i$ and $j$, $q^\prime$ is the difference in the probabilities of agreement and disagreement: 
\begin{equation}
    q^\prime = \Pr(i=j) - \Pr(i\ne j) \;.
\end{equation}
In the dichotomous and unidimensional setting $q^\prime$--also known variously as \textit{quadrant correlation}, \textit{quadrant count ratio} (QCR), and Hamann similarity \citep{hamann_merkmalsbestand_1961,cheetham_binary_1969}--has additional appealing properties in addition to its simplicity.  $q^\prime$ is mathematically identical to propensity-adjusted bias-adjusted kappa (PABAK, \cite{byrt_bias_1993}), a version of Cohen's $\kappa$ \citep{cohen_coefficient_1960}. This leads to one interpretation of $q^\prime$ that may be more intuitive: each item or variable or observation is considered a "rater" trying to rate some underlying capability, with certain propensities and biases. The correlation between two items would represent their PABAK for the unobserved latent trait. Similarly, the correlation between two subjects across items would represent their bias and prevalence corrected understanding of a latent construct. Indeed, \cite{holley_note_1964}, discuss $q^\prime$ as an option in factor analysis, leaving the question of ``suitability'' open. The present study will demonstrate its suitability: across over 200 empirical and publicly available datasets, Figure \ref{fig:fafit_vs_refact} illustrates the relationships between several classical metrics of factor model unidimensionality with refactor recoverability metrics (see Section \ref{sec:evalmeths}). 

We conjecture that $q^\prime$ has been misunderstood in its simplicity\footnote{As an example, see how \cite{fuxman_bass_using_2013} dismiss the metric as inappropriate by naively classifying the metric as only a ``matching'' coefficient} and underutilized.\footnote{We also take modest delight in utilizing a correlational relationship deemed by \href{https://en.wikipedia.org/wiki/Quadrant_count_ratio}{Wikipedia} (and, anecdotally and relatedly, by modern AI models, cursorily checked) as ``not commonly used'' except as a ``tool in statistics education'' in order to understand Pearson correlations.} The main body of the present study will illustrate differences in factor analyses by comparing $q^\prime$ with Pearson $\phi$ and tetrachoric $\rho_t$ correlations, and we include reporting of additional measures of association in the appendix.

% \begin{figure*}
%     \centering
%     \includegraphics[width=1\linewidth]{figs/verifact_panels_tet.pdf}
%     \caption{\textbf{Traditional Measures of Unidimensionality vs. Verifactor Reconstruction Metrics using Tetrachoric Correlations}. There is no meaningful relationship between the traditional measures of unidimensionality explored in this paper and the ability of the relationship to reconstruct the imaged signal. Each point represents a separate publicly available dataset. Each column represents a traditional metric of unidimensionality based on factor analyses. Each row represents Refactor Reconstruction Metrics. X-axes represent the strength of the traditional metric for the dataset. Y-axes are the fit of the out-of-sample reconstructions via Verifactor Analyses. Fig \ref{fig:loess_uni_vs_verifact} is a nonlinear version of this same figure based on LOESS regression.}
%     \label{fig:fafit_vs_verifact}
% \end{figure*}

\section{Experiments}\label{sec:sims}
In this section we discuss simulated and empirical experiments conducted. The summary comparison of the results from all three experiments can be found in Figure \ref{fig:ecv_isor2}.

\subsection{Set-up and Variables}

\begin{figure*}
    \centering
    \includegraphics[width=1\linewidth]{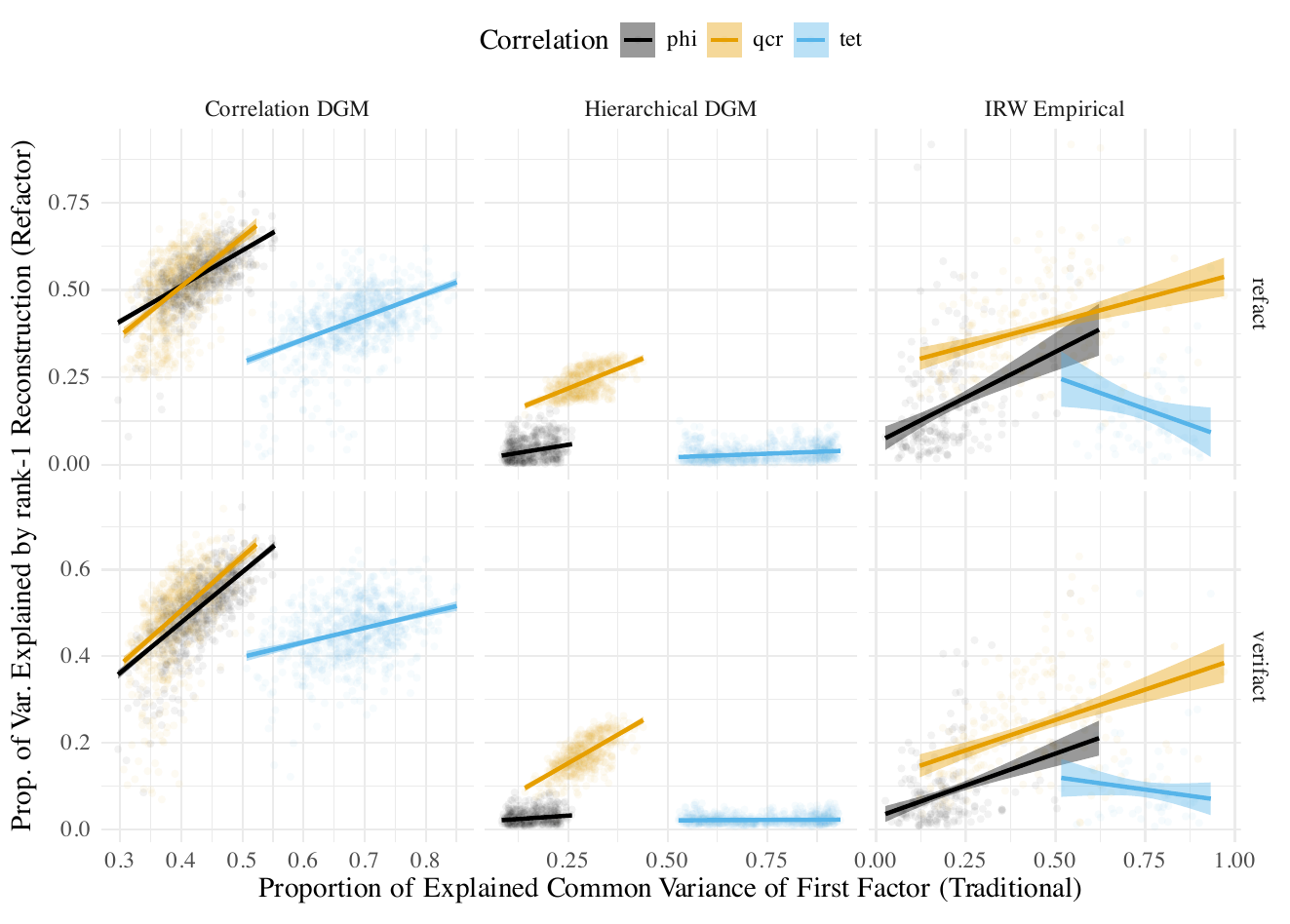}
    \caption{\textbf{Proportion of Explained Common Variance (ECV) vs Proportion of Variance Explained by best monotone rank-1 Reconstruction}: given the reconstruction, $\widehat{X}$ across three different conditions. \textbf{(left)} Simulation Study I: 1000 simulations where the underlying data generating models (DGM) are unidimensional tetrachoric correlations (see Section \ref{sec:simple_sim}. \textbf{(middle)} Simulation Study II: a hierarchical DGM with minor noise factors, following \cite{revelle_unidim_2025} (see Section \ref{sec:unidim_reprod}. \textbf{(right)} Empirical Study: 200 publicly available empirical datasets using the Item Response Warehouse (see Section \ref{sec:empirical}). \textbf{(top)} Refactor Analysis and \textbf{(bottom)} Verifactor out-of-sample bi-cross validated prediction. \textbf{(color)}represents different correlational relationships.}
    \label{fig:ecv_isor2}
\end{figure*}

\subsubsection{Association hypotheses and Refactor evaluation: Demonstration of Recoverability}
\label{sec:sim1_assoc_eval}

For each experiment, we fit rank--1 structures under three association operators used for dichotomous data:
\begin{enumerate}\itemsep2pt
\item \textbf{$\phi$ (Pearson on binary)}: treats $(0,1)$ as numeric and measures linear co-movement;
\item \textbf{Tetrachoric correlation}: assumes each binary item is a thresholded latent normal variable and estimates the latent correlation;
\item \textbf{Quadrant correlation (Mosteller--Blomqvist)}: a robust, highly interpretable measure based on concordance of signs around medians. % in the non-binary case).
\end{enumerate}
Each operator $\mathcal{A}^{(m)}$ induces item- and respondent-images, yields loadings $(\widehat u^{(m)},\widehat v^{(m)})$, and produces a rank--1 reconstruction $\widehat X^{(m)}=\widehat u^{(m)}\widehat v^{(m)\top}$ which is evaluated using Refactor metrics (e.g., AUC, Kendall $\tau$, cosine similarity, isotonic likelihood).

\begin{wrapfigure}{R}{0.39\textwidth}
    \centering
    \includegraphics[width=0.95\linewidth]{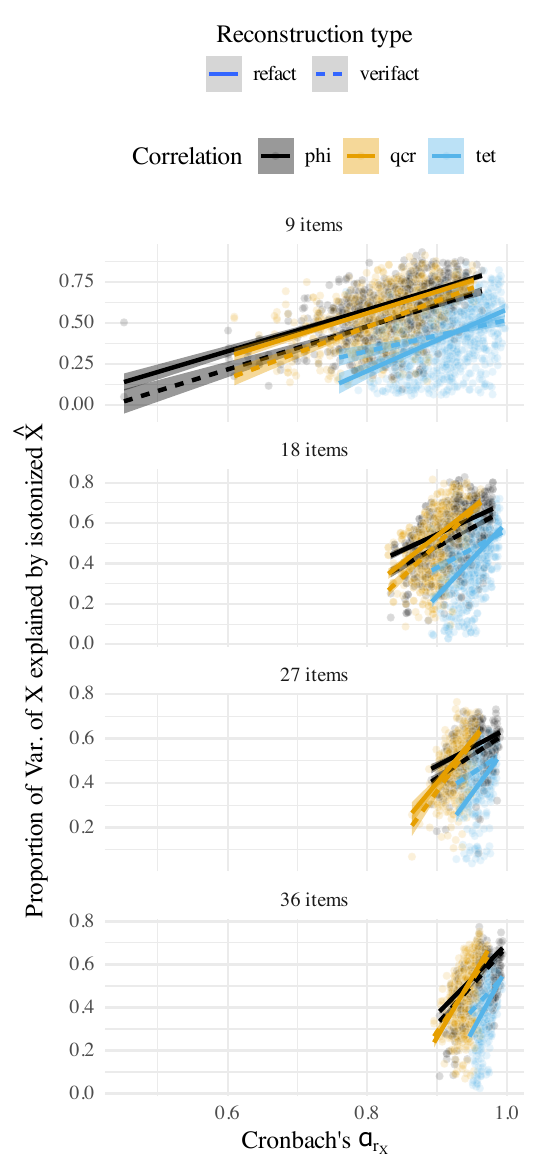}
    \caption{\textbf{Simulation I}: Correlation between Traditional Metrics and rank-1 recoverability} 
        \caption*{\footnotesize \textbf{Classical Unidimensionality (Cronbach's $\alpha_{\mathcal{A}_{X}}$) vs $\boldsymbol{R^2_{iso(\widehat{X})}}$ Proportion of Variance explained isotonized Refactor reconstruction}: %matrix cosine between $\mathbf{X}$ and $\widehat{\mathbf{X}}$: $\cos(\mathbf{X}, \mathbf{\widehat{X}}) = \langle \mathbf{X}, \mathbf{\widehat{X}} \rangle_F / \|\mathbf{X}\|_F \|\mathbf{\widehat{X}}\|_F$ 
    across three correlations. Robust $M$-estimator linear fit and confidence intervals are shown for 200 replications at each item set size. Cronbach's $\alpha$ displays its sensitivity to number of items. Response matrix sizes match those of Simulation II. Fully crossed linear relationships between classical measures of unidimensionality and refactoring measures are in Figure \ref{fig:basesimpanels}} 
    \label{fig:sim1_alpha_lik}
\end{wrapfigure}

\subsection{Simulation Study I: Refactor metrics align with classical indices under a correlational DGM}
\label{sec:simple_sim}

\subsubsection{Purpose and design}
\label{sec:sim1_purpose}
Our first simulation is intentionally a \emph{positive control}: the data-generating mechanism (DGM) is constructed so that a rank--1 latent correlational signal is the primary source of dependence. The purpose is not to challenge classical unidimensionality diagnostics, but to establish that when their assumptions are approximately correct, Refactor reconstruction metrics behave as expected and provide aligned evidence of unidimensional structure. In particular, we examine whether datasets with higher classical unidimensionality (e.g., $\omega$, CFI/TLI, or related one-factor fit indices computed from an association image) also yield higher Refactor reconstruction fidelity $m(X,\widehat X)$. This study illustrates the alignment of the traditional metrics of unidimensionality and Refactoring rank-1 metrics when the data generating mechanism matches $\mathcal{A}(X)$--the association relationship chosen for factor analysis.
% \subsubsection{Demonstration of Recoverability}\label{sec:simple_sim}
% First, 
We demonstrate that Refactoring is strongly correlated with classical metrics of unidimensionality, by simulating datasets where the data generating mechanism corresponds to a unidimensional, related via a tetrachoric correlational latent signal, where we expect to see that increasing traditional metrics of unidimensionality correspond with increasing Refactor rank-1 metrics of reconstruction.

% \subsubsection{Purpose and design}
% \label{sec:sim1_purpose}

% Our first simulation is intentionally a \emph{positive control}: the data-generating mechanism (DGM) is constructed so that a rank--1 latent correlational signal is the primary source of dependence. The purpose is not to challenge classical unidimensionality diagnostics, but to establish that when their assumptions are approximately correct, Refactor reconstruction metrics behave as expected and provide aligned evidence of unidimensional structure. In particular, we examine whether datasets with higher classical unidimensionality (e.g., $\omega$, CFI/TLI, or related one-factor fit indices computed from an association image) also yield higher Refactor reconstruction fidelity $m(X,\widehat X)$.

\subsubsection{Data-generating mechanism: a latent outer-product signal with heterogeneous thresholds}
\label{sec:sim1_dgm}

Let $N$ denote the number of respondents (rows) and $P$ the number of items (columns). We generate independent latent vectors and a rank--1 continuous signal matrix $Z$:
\[
\theta \in \mathbb{R}^{N},\quad \theta_k \sim \mathcal{N}(0,1),
\qquad
\lambda \in \mathbb{R}^{P},\quad \lambda_j \sim \mathcal{N}(0,1),
\qquad
Z = \theta \lambda^\top \in \mathbb{R}^{N\times P}
\]
Binary responses are formed by applying both item-specific and person-specific thresholds:
% \[
$X_{kj} = \mathbb{I}\{Z_{kj} > \tau_j\}\cdot \mathbb{I}\{Z_{kj} > \eta_k\}$,
% \]
where $\tau_j\sim\mathcal{N}(0,0.5)$ controls item difficulty (base rate) and $\eta_k\sim\mathcal{N}(0,1)$ controls respondent propensity. This DGM yields a dataset whose dependence structure is driven by a single latent outer-product term $Z$, but with realistic heterogeneity in marginal distributions and sparsity induced by thresholds.

\subsubsection{Expected outcome and interpretation}
\label{sec:sim1_expected}
We simulate 1000 iterations, using 36 items and 200 respondents (for direct comparability with the simulations in Section \ref{sec:unidim_reprod}). The reconstructions and results for a single iteration are shown in Figures \ref{fig:phi_dgm_vs_auc} and \ref{fig:phi_corr_sim}, respectively. 
Because the DGM is explicitly rank--1 in its latent signal, both traditional image-based indices and Refactor reconstruction metrics are expected to move together: datasets that are more clearly dominated by the rank--1 component (e.g., less threshold-induced attenuation or less marginal confounding) should show stronger one-factor fit on the association image and higher data-level recoverability.

Consistent with this expectation, we observe positive relationships between classical unidimensionality indices and Refactor reconstruction metrics across replications. This alignment establishes that Refactor is not a contrarian evaluation that ``penalizes'' correct factor structure; rather, it is a direct test of the same substantive claim (rank--1 adequacy), expressed at the level of the original response matrix.

At the same time, the simulation illustrates a subtle but important distinction between association choices. Under this correlational DGM, tetrachoric correlation tends to yield slightly higher values on traditional image-based indices yet slightly lower Refactor reconstruction values than $\phi$ and quadrant correlation, while preserving positive concordance overall. This pattern foreshadows our empirical findings: image-based fit can be inflated by how an association operator encodes dependence (especially under small samples or strong thresholding), whereas reconstruction metrics remain anchored to the predictive recoverability of $X$. The summary statistics using ECV and $R^2_{iso}$ can be found in Figure \ref{fig:ecv_isor2}, left column. A full panel of crossed traditional and Refactor metrics can be found in Figure \ref{fig:sim1full}.

\begin{wrapfigure}{R}{0.5\textwidth}
    \centering
    \includegraphics[width=0.49\textwidth]{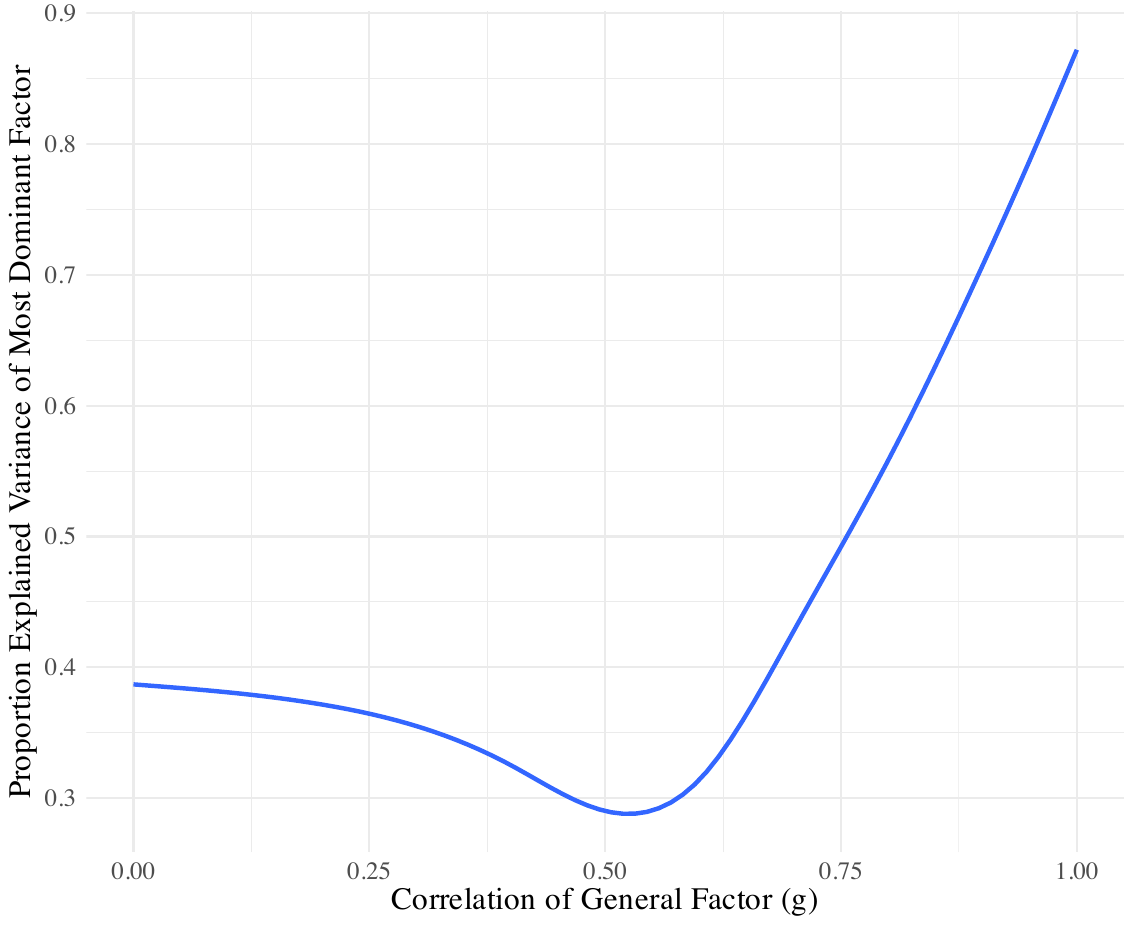}
    \caption{$g$ vs proportion of total variance explained by most dominant factor for measures of unidimensionality (this is distinct from ECV which is proportion of \textit{explained variance}, which ignores the proportion of signal due to the residual variance.)}
    \label{fig:g_vs_domf}
\end{wrapfigure}

\subsection{Simulation Study II: Replication of Unidimensionality Study}\label{sec:unidim_reprod}
Centering our simulations on comparability with extant measures of unidimensionality, we replicate the culminating simulations of \cite{revelle_unidim_2025}, providing a brief overview of the methods and new representations of the same findings. We use their \texttt{psych} package (\cite{revelle_psych_2024}) in \texttt{R} (\cite{r_core_team_r_nodate}) to calculate it and other measures of unidimensionality from data images. (Other tests of unidimensionality, not dependent on factor analysis, such as Cronbach's alpha, can be found in \cite{sijtsma_recognize_2024,mair_unidimensional_2015,van_abswoude_comparative_2004,sijtsma_use_2009, cronbach_my_2004}). In Section \ref{sec:unidim_reprod}, we replicate Revelle and Condon's culminating simulations,  with those based on tetrachoric and $\phi$ associations can be seen in Fig. \ref{fig:sim_phitetq}. The summary statistics using ECV and $R^2_{iso}$ can be found in Figure \ref{fig:ecv_isor2}, center column.

\subsubsection{Purpose and design}
\label{sec:sim2_purpose}
Our second simulation investigates further the concept of unidimensionality, replicating a recent study of unidimensionality \citep{revelle_unidim_2025}. We compare several traditional and rank-1 Refactor metrics for their ability to measure unidimensionality. We show that Refactor measures are sensitive to changes in unidimensionality as found in prior literature.

\subsubsection{Unidimensionality Simulation Studies}
Revelle and Condon demonstrate care and sophistication when representing noise within their simulations \citep{maccallum_representing_1991,maccallum_factor_2007}. They simulate data where a higher level ($g$) factor arises from the composite of three lower level factors \citep{jensen_what_1994} in order to demonstrate how unidimensional measures evolve as underlying factors increasingly correlate with $g$, and thus shift from being more multidimensional to more unidimensional.

% \begin{figure*}[p]
%     \centering
%     \includegraphics[width=1\linewidth]{figs/phi_tet_refact_trad.pdf}
%     \caption{Comparison of contemporary and Refactor measures of unidimensionality across factor model representations estimated with (\textbf{\textit{left}}) $\phi$ correlation and (\textbf{\textit{right}}) tetrachoric correlations. We note that many of these relationships should look concerning: tetrachoric has higher traditional/contemporary metric estimates (\textbf{\textit{bottom, y-axis}}) but lower predictive Refactor metric estimates (\textbf{\textit{top, y-axis}}) . $\phi$ shows insensitivity to sample size, but also insensitivity to a strengthened general factor (x-axis)  The data are simulated following the procedure used by \cite{revelle_unidim_2025} except using dichotomous data rather than categorical and do so using the \texttt{psych} package in \texttt{R} \cite{revelle_psych_2024, r_core_team_r_nodate}.}
%     \label{fig:phi_tet}
% \end{figure*}

Desiderata identified by the authors include finding metrics that, unlike Cronbach's $\alpha$, are not as sensitive to number of items or sample sizes. They demonstrate that $u_{RC}$ is far less sensitive to number of items than extant metrics, but is still sensitive to overall dataset size. % Thus, we add to their desiderata

We replicate their examples for dichotomous items, using both tetrachoric correlation and $\phi$/MCC,\footnote{the Pearson correlation on binary data is also known as the Matthews Correlation Coefficient or MCC} displayed in the bottom row of Fig. \ref{fig:sim_phitetq}. 

\subsubsection{Reevaluating the X axis}
For seeing how metrics evolve as unidimensionality increases, we need a slight modification of the x-axis. In the original study, it represented the correlation of a second-order $g$ with three lower level factors, each loading onto one third of the items. $g$ is distinct from the unidimensionality of the data simulated, as the strength of unidimensionality would be a measure of the strength of the most dominant signal, relative to the rest of the noise in the real-world usage where the data generation is not known perfectly. In the case of this simulation, the underlying unidimensionality simulated is nonmonotonic due to minor noise factors and imbalanced lower level factor magnitudes, as shown in Figure \ref{fig:g_vs_domf} where the coefficient of determination of the single most dominant signal is measured by the sum of the squared standardized Schmid Leiman loadings $\boldsymbol{\hat{\lambda}}$ used in the hierarchical data generating model \citep{schmid_development_1957,wolff_exploring_2005}: 
\begin{equation}
    \widetilde{R}^2 = f(\boldsymbol{\hat{\lambda}}) = \underset{i}{\operatorname{argmax}}  \sum_j \hat{\lambda}^2_{ij}\;.
\end{equation}

Thus, we replace the $g$ factor x-axis, which could bias our results in favor of relationships defined by correlation, for the proportion of variance explained by the most dominant factor, $\widetilde{R}^2$. In this case, until the correlation between $g$ and the first-order factors reaches 0.5, the dominant signal is not found across all items, but only across three of the items. Therefore we would expect reconstructions below this point to be poor and/or noisy. In terms of our Refactor evaluation metrics, we would expect to see clear differentiation between the signal recovery above and below that point, which in this case $\widetilde{R}^2 \approx 0.4$ and represented by a red dashed vertical line in Figure \ref{fig:sim_phitetq}. A full panel of crossed traditional and Refactor metrics can be found in Figure \ref{fig:sim2full}.

\begin{figure*}[h!]
    \centering
    \includegraphics[width=1\linewidth]{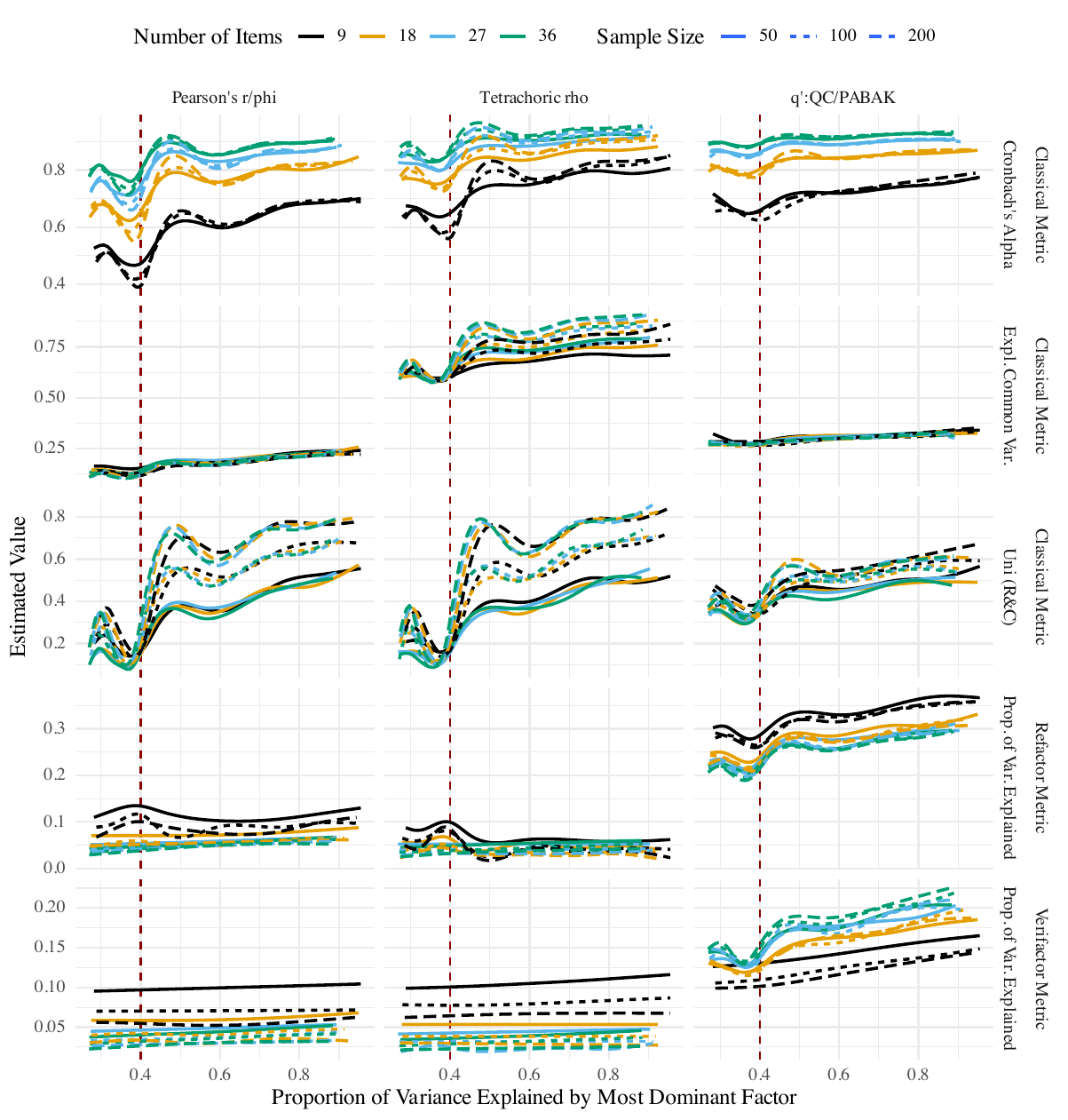}
    \caption{\textbf{Simulation II}: Comparison of Classical and Refactor Metrics of Unidimensionality on Simulated Data (Replication of \cite{revelle_unidim_2025}). X-axis represents the proportion of simulated signal represented by the most dominant dimension: sum of the squared standardized Schmid Leiman loadings $\boldsymbol{\hat{\lambda}}$ used in the hierarchical data generating modelhigher values indicate stronger dominance of a single latent factor. Y-axis represents the value of each respective metric. The bottom two rows show $R^2_{iso}$ for both Refactor and Verifactor reconstructions. All three correlation types show positive associations with the strength of unidimensionality along the x axis. However,  $q^\prime$ is the only relationship that shows strong positive relationship with simulated unidimensionality in reconstruction.  Additional metrics and comparisons are in Appendix \ref{apx:supp_figures}}.
    \label{fig:sim_phitetq}
\end{figure*}

\subsection{Empirical Study: 200 Datasets}\label{sec:empirical}
In the empirical study we calculate and contrast the various classical measures of unidimensional fit for factor models with refactor models across 200 empirical datasets. We would expect that these data to vary in their degree of unidimensionality and in the relational abstraction that unites any underlying signal. We measure the extent to which these data are unidimensional according to traditional measures of unidimensionality and contrast them with the rank-1 recoverability Refactor metrics. In Summative Figure \ref{fig:ecv_isor2}, these datasets represent the rightmost column. A full panel of crossed traditional and Refactor metrics can be found in Figure \ref{fig:empfull}.

\begin{wrapfigure}{R}{0.75\textwidth}
    \centering
    \includegraphics[width=0.74\textwidth]{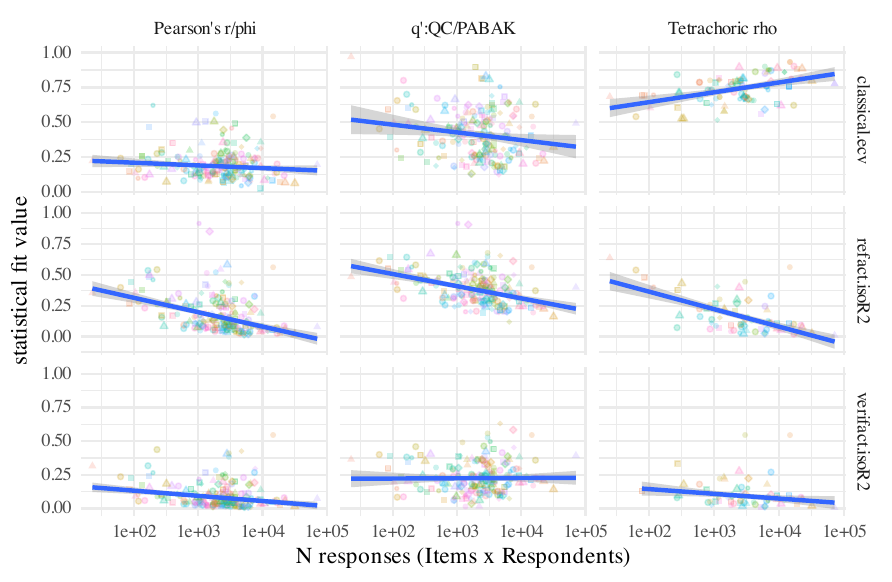}
    \caption{Metric Stability under increasing observations. Log-scale x-axis. Highlighting the importance of Verifactor analyses for larger samples, we see that it is relatively invariant. A shift from Refactor to Verifactor estimations across various reconstruction metrics for tetrachoric associations are in Figure \ref{fig:tet_all_metrics}.}
    \label{fig:emp_verifact_stability}
    
\end{wrapfigure}

\subsubsection{Purpose and design}
\label{sec:emp_purpose}

Our empirical experiment replaces bootstrapping of individual datasets with bootstrapping across real-world datasets, in order to capture the real variation and noise found in each dataset. Metrics are reported to show the relationship between traditional metrics and Refactor across three different hypothesized correlational relationships. Utilizing real-world data avoids the researcher-biases that accompany simulation construction. In fact, this difference can be seen visually by contrasting the partial squared distance correlation metric rows between this dataset, in Figure \ref{fig:empfull}, and the dataset from the previous simulation, in Figure \ref{fig:sim2full}. What is striking about the contrast is that it shows that even a more complicated simulation design, such as that in Simulation II with hierarchical and minor noise facets, fail to produce the kinds of dependence structures found in real-world data.

Having a more authentic representation of real-world noise allows for an objective measurement of metric stability. Under the assumption that the underlying relationship measured across these 200 datasets is orthogonal to the number of observations in each dataset, we would expect to see no relationship between the number of data points and each metric. Figure \ref{fig:emp_verifact_stability} shows how Verifactor estimates show the least sensitivity to the size of the sample, with $q^\prime$ having no statistical difference between the largest and smallest.

\section{Discussion}
\label{sec:discussion}

This work introduced \textbf{Refactor} and \textbf{Verifactor} analyses as data-first complements to traditional, image-based evaluation of factor models. The central message is that \emph{unidimensionality is a rank--1 claim about the response matrix $X$}, yet it is rarely evaluated where it lives. Our results show that classical unidimensionality indices can behave as expected under idealized correlational data-generating mechanisms, but that they can also become weakly informative---and sometimes systematically misleading---when the operative dependence in real data deviates from the correlation geometry assumed by the analyst. Refactor and Verifactor make these deviations observable by turning a factor model into a matrix prediction problem and evaluating the resulting reconstruction on $X$ (in-sample) and on held-out blocks (out-of-sample).

\subsection{Two notions of ``fit'': image coherence versus data recoverability}
\label{sec:discussion_two_fits}

Factor-analytic practice often treats fit as a property of a correlation matrix. Refactor/Verifactor separate two evaluands:
\begin{enumerate}\itemsep2pt
\item \textbf{Image coherence:} does a one-factor model explain the association image $A=\mathcal{A}(X)$ well?
\item \textbf{Data recoverability:} does the induced rank--1 representation yield a faithful reconstruction $\widehat X$ of the original matrix (Refactor), and does it predict held-out blocks (Verifactor)?
\end{enumerate}
These are related but not equivalent questions. Image coherence is necessary for recoverability when $\mathcal{A}$ is appropriate and sufficiently informative, but it is not sufficient when (i) the association operator discards information relevant to $X$, (ii) the response process is nonlinearly related to the latent structure, or (iii) the dominant structure is not well represented by the assumed correlation geometry. The summary pattern in Figure~\ref{fig:ecv_isor2} makes this distinction concrete: across simulations and across empirical datasets, the relationship between explained common variance (ECV) and isotonic reconstruction $R^2$ can tighten, loosen, or invert depending on whether the chosen associative abstraction matches the data.

\subsection{What the simulations establish (and why both are needed)}
\label{sec:discussion_sims}

\subsubsection{Simulation I: validity under a correlational DGM}
\label{sec:discussion_sim1}

Simulation~I is a positive control: the DGM is constructed so that a rank--1 correlational signal governs dependence. In this regime, traditional indices and Refactor recoverability \emph{should} agree, and they do. This finding is important for interpretability: Refactor does not ``punish'' unidimensionality; rather, it measures the same claim at the level of $X$ by asking whether a rank--1 model can recover the response matrix up to a monotone link.

At the same time, Simulation~I already exposes a subtle methodological issue that foreshadows the empirical results. Tetrachoric correlation yields systematically higher \emph{image-based} unidimensionality (e.g., higher ECV and related indices), yet does not yield uniformly stronger \emph{data-level} recoverability than $\phi$ or $q'$. In other words, tetrachoric can make the \emph{image} look more one-dimensional without producing correspondingly better reconstruction. This is precisely the kind of discrepancy that Refactor is designed to surface: it distinguishes ``one-dimensional in the chosen image space'' from ``one-dimensional in the response process.''

\subsubsection{Simulation II: hierarchical structure and the diagnostic value of $q'$}
\label{sec:discussion_sim2}

Simulation~II (replicating \citealp{revelle_unidim_2025} under dichotomization) introduces a more realistic structure: a dominant factor plus minor factors and noise. In this setting, the most informative outcome is not that all methods correlate with the underlying strength of a dominant dimension---many do---but that \textbf{quadrant correlation $q'$ produces reconstructions that respond to increases in unidimensional dominance in the way a measurement scientist would expect}. In particular, as the dominant dimension becomes stronger (and, in our reparameterization, becomes more globally expressed across items), Verifactor isotonic $R^2$ under $q'$ increases cleanly and monotonically, while reconstructions under $\phi$ and tetrachoric are comparatively less diagnostic.

This matters because Simulation~II is closer to the applied case: real instruments and benchmarks are rarely ``pure factor''; they include mixtures of nuisance dimensions, local dependence, and heterogeneity in marginals. The simulation therefore positions $q'$ not merely as an alternative correlation coefficient, but as a \emph{more stable associative geometry} for rank--1 recovery under plausible deviations from idealized latent-normal assumptions.

\subsection{Empirical study: stability, sample size, and why Verifactor is essential}
\label{sec:discussion_empirical}

The empirical study replaces researcher-controlled noise with naturally occurring heterogeneity across 200 public datasets. Two findings are especially consequential.

\paragraph{(1) Recoverability metrics are less entangled with sample size than image-fit indices.}
Traditional indices and some image-derived measures vary systematically with the number of observations and items, even when the scientific question (whether the instrument is effectively unidimensional) should not depend on sample size per se. In contrast, Refactor metrics are more stable, and Verifactor metrics are the most stable (Figure~\ref{fig:emp_verifact_stability}). This behavior is desirable under a random-respondent, random-item perspective: metrics should reflect \emph{structure} rather than opportunistic properties of estimation in a particular sample.

\paragraph{(2) The ``best'' association is data-dependent, and $q'$ is consistently competitive.}
Across datasets, the association operator that yields the strongest recoverability varies. This is a feature, not a bug: it implies that ``correlation'' is an empirical hypothesis about the signal of interest, and different measures encode meaningfully different notions of dependence for binary data. The empirical results show that quadrant correlation is nearly uniformly more desirable than $\phi$ and tetrachoric across the reconstruction metrics emphasized in the main body, supporting its reintroduction as a practical default for dichotomous factor modeling when robustness and interpretability are required. Tetrachoric correlations are particularly vulnerable to blindspots in the absence of Refactor metrics, as there is no statistical justification for assuming a tetrachoric relationship between items when the relationship is not already known \citep{gronneberg_partial_2020}. 

\subsection{Why quadrant correlation works well here (interpretability and robustness)}
\label{sec:discussion_qprime}

Quadrant correlation $q'$ \citep{mosteller_useful_1946,blomqvist_measure_1950} has two properties that match the present problem.

\paragraph{Robustness to marginal distortions and heavy tails.}
Binary response matrices commonly exhibit strong base-rate variation (item difficulty; respondent propensity) and skew. Measures that are highly sensitive to marginal imbalance or rely on strong latent-normal assumptions can overstate or understate dependence in ways that inflate image-based unidimensionality. Quadrant correlation depends on concordance relative to central tendencies and therefore tends to be less driven by extreme marginals and less brittle under dichotomization.

\paragraph{Compatibility with monotone latent structure.}
Many measurement models imply that the probability of a positive response is monotone in a latent score; they do not imply linearity in the observed scale. By pairing $q'$ with isotonic reconstruction evaluation, the analysis emphasizes what is scientifically stable across a broad class of monotone links: ordinal structure and recoverable dependence, rather than the specific geometry induced by a particular correlation estimator.

These features provide a principled explanation for why $q'$ yields reconstructions that are both (i) sensitive to strengthening unidimensional signal in simulation and (ii) comparatively stable across heterogeneous empirical datasets.

\subsection{Practical implications: recommended workflow}
\label{sec:discussion_workflow}

The results suggest a concrete workflow for applied users.

\begin{enumerate}\itemsep2pt
\item \textbf{Start with an explicit rank--1 target.} Treat unidimensionality as the hypothesis that $X$ is approximately recoverable by $\widehat X = \widehat u\,\widehat v^\top$ up to an appropriate link.
\item \textbf{Evaluate multiple association hypotheses.} For dichotomous data, compare at least $\phi$, tetrachoric, and $q'$ rather than defaulting to a single image.
\item \textbf{Use ECV and isotonic $R^2$ together.} ECV quantifies one-dimensionality of common variance in the image; isotonic $R^2$ quantifies one-dimensional recoverability of the response process. Divergence between them is diagnostic.
\item \textbf{Prefer Verifactor for generalization claims.} When rows and columns are both stochastic (random respondents and random items), out-of-sample block prediction is the appropriate evidentiary standard.
\end{enumerate}

\subsection{Limitations and scope}
\label{sec:discussion_limitations}

Three limitations bound the interpretation of our findings.

\paragraph{First, Refactor/Verifactor do not replace substantive validity.}
A strong rank--1 reconstruction can indicate coherent shared structure, but it does not by itself establish that the dimension corresponds to the intended construct. Conversely, weak reconstruction can indicate multidimensionality, local dependence, or a mismatch between the chosen association and the signal.

\paragraph{Second, reconstructions depend on the estimator and calibration.}
We focused on rank--1 structures and monotone calibration (isotonic regression) to keep the claim interpretable and robust. Other estimators and links may be more appropriate in particular domains, and extending the framework to explicitly model different links is a promising direction.

\paragraph{Third, the empirical study aggregates across datasets.}
While this provides breadth and reduces simulation bias, it does not replace careful, domain-specific case studies where item content, scoring rules, and known substructures are available.

\section{Conclusion and Future Directions}

\paragraph{Rethinking Generalization: From Static Items to Stochastic Systems}
The Refactor and Verifactor frameworks offer more than just a new set of metrics; they offer a new lens through which to view model validity. By shifting the focus from the abstract fit of a model's image to the concrete, predictive fidelity of the model for the data itself, these methods provide a more stringent, interpretable, and diagnostically rich evaluation of unidimensionality and correlational appropriateness. For psychometricians, this provides a powerful tool to move beyond checking fit indices to validating that a test's purported structure truly explains subject response patterns. For computer scientists and machine learning practitioners, it provides a robust methodology for assessing latent structure in high-dimensional, sparse, or irregular data where traditional statistical assumptions may not hold. This framework opens the door to a more critical, empirical, and ultimately more useful science of latent variable modeling.

This is achieved through a dual analysis of the data matrix, deriving not only the column-space loadings ($\boldsymbol{A}$) from the covariance matrix ($\boldsymbol{X}^T\boldsymbol{X}$) but also the row-space loadings ($\boldsymbol{B}$) from the Gram matrix ($\boldsymbol{X}\boldsymbol{X}^T$). These two matrices, representing the principal axes of the data from both orientations, are then used to "decompress" or refactor a complete, model-based prediction of the original data matrix, $\boldsymbol{\hat{X}}$. This refactoring step transforms the dimensionality reduction model into a generative, predictive one.

Standard factor analysis implicitly assumes a fixed-effects model for items; the items are what they are, and we sample respondents. This view is increasingly challenged by random effects frameworks that treat both items and persons as sampled from larger populations \citep{de_boeck_random_2008}. Verifactor analysis operationalizes this modern perspective by treating both rows (observations) and columns (variables) as stochastic units for resampling. This provides a far more rigorous test of model generalization than traditional approaches.

Furthermore, Verifactor's construction avoids a subtle but critical form of information leakage common in other cross-validation schemes. Many methods might use factor scores estimated from the full dataset to predict held-out responses. These scores, however, have been "exposed" to the very data they are meant to predict. Verifactor's prediction of block $\boldsymbol{A}$ via $\boldsymbol{\hat{A}} = \boldsymbol{B}(\boldsymbol{\hat{D}}^{(1)})^\dagger\boldsymbol{C}$ is constructed using *only* information from the disjoint "training" blocks. It is a true out-of-sample test of the latent structure's ability to generalize, not just the model's ability to interpolate.

\paragraph{Rethinking Subjects: From Psychometrics to Cyber-metrics}
This is particularly critical in modern applications, such as evaluating large language models on benchmarks where the number of test examples ($n$) can be much smaller than the number of features or tasks ($p$). In such $n \ll p$ regimes, the empirical covariance matrix $\boldsymbol{X}^T\boldsymbol{X}$ is rank-deficient and unstable, rendering traditional fit indices meaningless or misleading. Verifactor's bi-cross-validation on submatrices is more robust to this setting, as the prediction of $\boldsymbol{A}$ relies on the structure learned from $\boldsymbol{D}$, which can remain well-conditioned.

% But as we have access to more data and compute, we have the opportunity to explore methods that retain the ease of interpretability whilst also providing insights that more closely align with the reality we are seeking to study. Is it possible that when trying to understand human measurement, the latent constructs are consistently more accurately measured by even simpler relationships than Pearson and tetrachoric correlations? Is it possible that we don’t have to sacrifice predictive power when modeling the underlying dependence structure?

The implications of possible improvements to our most foundational models may reverberate beyond the scope of a single study. If there is any echo to be heard, in the dawn of the AI Age, we need more than ever to have our most sensitive microphones and best measurement scientists scrutinizing the signal. We conjecture that there will likely differences between AI and human behaviors that psychometric tools are best positioned to detect and that these tools may need some updates to detect them.

Indeed, we offer initial findings in this direction: in Table \ref{tab:verifeols}, we add 19 AI benchmarks to the 200 empirical datasets from IRW, and estimate the difference between the Refactor and Verifactor estimates for both humans and AI. Greater discrepancy between these values, across 17 different hypothesized measures of association, would suggest that Refactor would overfit to the data relative to Verifactor, even after controlling for size and submatrix and bootstrap composition. We see this trend: Verifactor is consistently lower. However, we see across all recoverability metrics, that human datasets show significantly more discrepancy than AI datasets, AI datasets show less unpredictable variation in a first dimension than the human datasets. Explorations such as this warrant future studies. The above findings highlight another important future direction: developing the theory and interpretation for allowing $\mathcal{A}_c \ne \mathcal{A}_r$, equipping researchers with guidance for factorially more flexibility with which to specify underlying latent relationships. 

\paragraph{Popping Spherical Cows}
Unidimensionality is an empirical claim about a matrix of responses, but conventional evaluation primarily inspects a transformed image of that matrix. Refactor and Verifactor re-center evaluation on \emph{recoverability}: a one-dimensional factor model is supported to the extent that it reconstructs $X$ (Refactor) and predicts held-out submatrices (Verifactor) under a random-respondent random-item perspective. Across simulations and 200 real datasets, this shift reveals a consistent and actionable insight: classical image-based indices can overstate unidimensionality for dichotomous data, while quadrant correlation paired with reconstruction-based evaluation provides a robust, interpretable alternative that better tracks recoverable signal.

This study proposes \textit{Refactor} methods for understanding and testing one of our most prized “spherical cows”: the methods with which we test and fit unidimensional factor analyses of dichotomous item responses. Classical methods and metrics have been in use for well over 120 years and have the allure of both being interpretable and computationally tractable. Without sacrificing interpretability nor computability, we can increase predictive power and test the imposed assumptions of the relationships.

\section{Related Work}\label{sec:background}
\subsection{Historical Note about Factor Analyses}\label{sec:historical}

The history of unidimensional factor analysis is not neutral, and commentary is warranted \citep{maccallum_factor_2007}. Mathematically, it was developed ahead of its time, but practically Spearman constructed it in pursuit of a harmful eugenicist theory of ``general intelligence'' \citep{spearman_general_1904, oconnor_spearman_2021}. %Its history with assessment is not neutral, and commentary is warranted .

In this paper we contrast the traditional methods and metrics of unidimensional factor analysis with data-centric methods. We present evidence that challenges conventional wisdom, by fitting models with reduced assumptions to identify the nature of underlying relationships. This study offers evidence from simulations and over 200 human datasets that the uses of unidimensional factor analysis via Pearson's correlations lead to overconfident estimates that are further obscured when diagnosed with traditional metrics of fit. We also present solutions, based on fewer assumptions, that meet the expectations of both traditional and predictive measures of model fit.

 We are not the first to present evidence against these original models. Even within 10 years, others had noted the discrepancy between his theory and the actual data \citep{thomson_hierarchy_1916}, but the power of the method to unify and interpret correlational relationships ensured its use and attention for centuries. The many measurement scientists that have worked on \citep{thurstone_multiple_1931,ll_thurstone_vectors_1935}: these models have done more to acknowledge that ``no model is completely faithful'' \citep{cudeck_model_1991}. In attempts to address some of the problems we discuss, some recommended Pearson's tetrachoric correlations, over product-moment relationships \citep{pearson_i_1900,wherry_factor_1944}, a popular choice that is becoming increasingly difficult to justify \citep{gronneberg_partial_2020}. Others sought methods to improve and test for model appropriateness \citep{muthen_testing_1988, maccallum_representing_1991}. 

While not our initial objective, our results challenge 120 years of practice. Or rather, they are challenging the notion that one must accept the statistical assumptions of the previous 120 years to benefit from interpretable dimensionality reduction.

\subsection{Traditional Factor Model Methods and Metrics}\label{sec:traditionalmeth}
Traditional methods for assessing the latent structure of a data matrix, such as factor analysis, have historically focused on the properties of the model's \textit{image}. For an observed data matrix $\boldsymbol{X} \in \mathbb{R}^{n \times p}$, a factor model posits a lower-rank structure $\boldsymbol{X} \approx \boldsymbol{\textsf{Z}}\boldsymbol{F}^T + \boldsymbol{E}$, where $\boldsymbol{\textsf{Z}} \in \mathbb{R}^{n \times k}$ are latent scores, $\boldsymbol{F} \in \mathbb{R}^{p \times k}$ are factor loadings, and $\boldsymbol{E}$ is an error term. Model evaluation, particularly the test for unidimensionality ($k=1$), then proceeds by analyzing the properties of the estimated loadings $\boldsymbol{\hat{F}}$ or the associated covariance matrix image (e.g., $\boldsymbol{\hat{F}}\boldsymbol{\hat{F}}^T$). While this approach is foundational, it provides an indirect and abstract view of model adequacy. It assesses the characteristics of the model's components rather than the model's fundamental ability to account for the signal in the observed data from which it was derived, consistent with findings from recent model comparison studies \citep{%domingue_intermodel_2021 ,
zhang_intermodel_2023}.

This work introduces a paradigm that shifts the evidentiary basis for model fit from the abstract image to the data itself. We formalize and motivate two novel evaluation frameworks: \textbf{Refactor} analysis, which assesses a model's fidelity in reconstructing the original data, and \textbf{Verifactor} analysis, which evaluates its predictive power on held-out data partitions (see Figure \ref{fig:refactor}). These methods provide a direct and rigorous test of the hypothesis that the signal of interest in a data matrix is well-represented by a low-rank dependence structure. In binary factor models, Pearson's $r$/$\phi$/Kendall's $\tau$ (equivalent with dichotomous items), tetrachoric correlation, and Yule's $Q$ are often imposed as the assumed structure \citep{revelle_psych_2024}.

% \subsubsection{Predictive Metrics for Factor Analysis}\label{sec:relatedwork}
% Despite its long history, the area of data-level predictions for factor models is relatively under-explored.  Nevertheless, good research has been done, inclu

% % \subsection{Related Work}
% The original extension from Nonnegative Matrix Factorization (NMF) to Factor analysis was motivated for the removal of
% unwanted variation from microarray experiments \citep{owen_bi-cross-validation_2015}. Factor analyses themselves originally stem from the (problematically motivated \cite{oconnor_spearman_2021}) statistical methods from \cite{spearman_general_1904}. Suppose that the ideal data for estimating a model consists of a sample of (X, Z), but the researcher only observes X. Our objective is to identify the latent variable(s) Z under the most general conditions.

% We identify latent variables in nonparametric models with nonlinear generating processes based on the so-called Hu-Schennach Theorem \citep{zheng_nonparametric_2025,hu_instrumental_2008,hu_econometrics_2017}, even when confronted with non-negligible noise. 

% \subsection{Covariance, Gram-functioning and \textit{Image} Matrices}
% We report results arising from Pearson's $r$/$\phi$/Kendall's $\tau$ (equal with dichotomous items), tetrachoric correlation, Yule's $Q$, Cramér's $V$, Loevinger's $H_{ij}$, and Cohen's $\kappa$, where we treat individual features as ``raters'' under the assumption that they rate the same construct. However, findings are robust to this choice and we illustrate this in sensitivity analyses below.

\subsubsection{Extant Metrics of Unidimensionality}
\cite{revelle_unidim_2025} provide a clear overview of current image-based dimensionality analysis, including a compelling case for their metric of unidimensionality, which we denote $\text{Uni} =u_{RC}$, which is the product of a measure of congeneric fit $\rho_c$ and a tau-equivalence-like measure of correlational homogeneity, $\tau_{RC}$. Other metrics they discuss include Cronbach's $\alpha$, McDonald's $\omega$, CFI, and several versions of Gutmann's $\lambda_i$-related metrics \citep{revelle_psych_2024}. %and we would add to that some of the more common indices such as RMSE, TLI, SRMR.

\subsection{Theoretical Motivations and Constraints}
The preference for low-rank linear models in psychometrics is driven by theory and interpretability, but it represents a strong assumption about the nature of latent traits. In practice, this assumption is never perfectly theoretically met; it just needs to be ``good enough'' for many applications \citep{domingue_intermodel_2021,domingue_intermodel_2024}. The challenge of dimensionality reduction is profound, as it involves mapping a high-dimensional space of observations to a low-dimensional latent space. Topological theorems reveal the inherent limitations of this process. For narrative flow we refer to the appendices for formal treatment and relevant mathematical proofs.

A one-to-one function $g: \mathbb{R}^p \to \mathbb{R}^k$ with $k < p$ cannot be continuous and highlights a key tension: if a latent trait truly reduces the dimensionality of the observed data, the mapping from the high-dimensional observation space to the low-dimensional latent space must be discontinuous. Factor analysis resolves this by imposing a linear generative model, $\boldsymbol{X} \approx \boldsymbol{\textsf{Z}}\boldsymbol{F}^T$, which is continuous but not one-to-one. This linearity is a powerful simplification, but its validity is an empirical question. The presence of complex, nonlinear relationships between items and the latent trait will violate this assumption, leading to poor reconstruction in a Refactor analysis.

While the identification of latent variables is possible even in nonparametric settings under certain conditional independence assumptions \citep{hu_instrumental_2008,hu_econometrics_2017,yalcin_nonlinear_2001,amemiya_model_1997}, such theorems do not guarantee that the resulting structure is linear or low-rank. \textbf{Refactor} and \textbf{Verifactor} analyses thus serve as essential empirical tools. They test the practical adequacy of the linear, low-rank simplification that is implicitly assumed in factor-analytic tests of unidimensionality, providing a grounded, data-centric verdict on whether a correlational relationship truly captures the signal of interest.

\begin{Backmatter}

\paragraph{Acknowledgments}
We are grateful for the assistance of Benjamin Domingue.

\printbibliography %[title={References}]

@article{revelle_unidim_2025,
    title = {Unidim: {An} index of scale homogeneity and unidimensionality},
    issn = {1939-1463},
    shorttitle = {Unidim},
    doi = {10.1037/met0000729},
    journal = {Psychological Methods},
    author = {Revelle, William and Condon, David},
    year = {2025},
    note = {Place: US
Publisher: American Psychological Association},
}

@misc{revelle_psych_2024,
    title = {psych: {Procedures} for {Psychological}, {Psychometric}, and {Personality} {Research}},
    copyright = {GPL-2 {\textbar} GPL-3 [expanded from: GPL (≥ 2)]},
    shorttitle = {psych},
    url = {https://cran.r-project.org/web/packages/psych/index.html},
    urldate = {2024-07-09},
    author = {Revelle, William},
    month = jun,
    year = {2024},
}

@misc{r_core_team_r_nodate,
    address = {Vienna, Austria},
    title = {R: {A} {Language} and {Environment} for {Statistical} {Computing}},
    url = {https://www.r-project.org/},
    urldate = {2024-07-09},
    publisher = {R Foundation for Statistical Computing},
    author = {R Core Team},
}

@article{ramsay_matrix_1984,
    title = {Matrix correlation},
    volume = {49},
    copyright = {1984 The Psychometric Society},
    issn = {1860-0980},
    url = {https://link.springer.com/article/10.1007/BF02306029},
    doi = {10.1007/BF02306029},
    language = {en},
    number = {3},
    urldate = {2025-08-22},
    journal = {Psychometrika},
    author = {Ramsay, J. O. and ten Berge, Jos and Styan, G. P. H.},
    month = sep,
    year = {1984},
    note = {Company: Springer
Distributor: Springer
Institution: Springer
Label: Springer
Publisher: Springer-Verlag},
    pages = {403--423},
}

@article{yanai_proposition_1980,
    title = {A {PROPOSITION} {OF} {GENERALIZED} {METHOD} {FOR} {FORWARD} {SELECTION} {OF} {VARIABLES}},
    volume = {7},
    issn = {0385-7417},
    url = {https://link-springer-com.stanford.idm.oclc.org/article/10.2333/bhmk.7.7_95},
    doi = {10.2333/bhmk.7.7_95},
    number = {7},
    urldate = {2025-09-11},
    journal = {Behaviormetrika},
    author = {Yanai, Haruo},
    year = {1980},
    note = {Publisher: The Behaviormetric Society},
    pages = {95--107},
}

@misc{szekely_partial_2014,
    title = {Partial {Distance} {Correlation} with {Methods} for {Dissimilarities}},
    url = {http://arxiv.org/abs/1310.2926},
    doi = {10.48550/arXiv.1310.2926},
    urldate = {2025-08-11},
    publisher = {arXiv},
    author = {Szekely, Gabor J. and Rizzo, Maria L.},
    month = jul,
    year = {2014},
    note = {arXiv:1310.2926 [stat]},
}

@article{szekely_energy_2017,
    title = {The {Energy} of {Data}},
    volume = {4},
    issn = {2326-8298, 2326-831X},
    url = {https://www.annualreviews.org/doi/10.1146/annurev-statistics-060116-054026},
    doi = {10.1146/annurev-statistics-060116-054026},
    language = {en},
    number = {1},
    urldate = {2025-08-14},
    journal = {Annual Review of Statistics and Its Application},
    author = {Székely, Gábor J. and Rizzo, Maria L.},
    month = mar,
    year = {2017},
    pages = {447--479},
}

@book{kendall_rank_1962,
    title = {Rank {Correlation} {Methods}},
    language = {en},
    publisher = {Hafner Publishing Company},
    author = {Kendall, Maurice George},
    year = {1962},
    note = {Google-Books-ID: 9ni4AAAAIAAJ},
}

@article{sijtsma_recognize_2024,
    title = {Recognize the {Value} of the {Sum} {Score}, {Psychometrics}’ {Greatest} {Accomplishment}},
    volume = {89},
    issn = {0033-3123},
    url = {https://www.ncbi.nlm.nih.gov/pmc/articles/PMC11588849/},
    doi = {10.1007/s11336-024-09964-7},
    number = {1},
    urldate = {2025-08-11},
    journal = {Psychometrika},
    author = {Sijtsma, Klaas and Ellis, Jules L. and Borsboom, Denny},
    year = {2024},
    pmid = {38627311},
    pmcid = {PMC11588849},
    pages = {84--117},
}

@article{van_abswoude_comparative_2004,
    title = {A {Comparative} {Study} of {Test} {Data} {Dimensionality} {Assessment} {Procedures} {Under} {Nonparametric} {IRT} {Models}},
    volume = {28},
    issn = {0146-6216},
    url = {https://doi.org/10.1177/0146621603259277},
    doi = {10.1177/0146621603259277},
    language = {EN},
    number = {1},
    urldate = {2025-06-25},
    journal = {Applied Psychological Measurement},
    author = {van Abswoude, Alexandra A. H. and van der Ark, L. Andries and Sijtsma, Klaas},
    month = jan,
    year = {2004},
    note = {Publisher: SAGE Publications Inc},
    pages = {3--24},
}

@article{sijtsma_use_2009,
    title = {On the {Use}, the {Misuse}, and the {Very} {Limited} {Usefulness} of {Cronbach}’s {Alpha}},
    volume = {74},
    issn = {0033-3123, 1860-0980},
    url = {https://www.cambridge.org/core/journals/psychometrika/article/on-the-use-the-misuse-and-the-very-limited-usefulness-of-cronbachs-alpha/72E9A648D5324412AF5506701B6BE325},
    doi = {10.1007/s11336-008-9101-0},
    language = {en},
    number = {1},
    urldate = {2025-07-29},
    journal = {Psychometrika},
    author = {Sijtsma, Klaas},
    month = mar,
    year = {2009},
    pages = {107--120},
}

@misc{filzmoser_pcapp_2024,
    title = {{pcaPP}: {Robust} {PCA} by {Projection} {Pursuit}},
    copyright = {GPL (≥ 3)},
    shorttitle = {{pcaPP}},
    url = {https://cran.r-project.org/web/packages/pcaPP/index.html},
    urldate = {2025-06-25},
    author = {Filzmoser, Peter and Fritz, Heinrich and Kalcher, Klaudius and Todorov, Valentin},
    month = aug,
    year = {2024},
}

@article{szekely_energy_2013,
    title = {Energy statistics: {A} class of statistics based on distances},
    volume = {143},
    issn = {0378-3758},
    shorttitle = {Energy statistics},
    url = {https://www.sciencedirect.com/science/article/pii/S0378375813000633},
    doi = {10.1016/j.jspi.2013.03.018},
    number = {8},
    urldate = {2025-07-21},
    journal = {Journal of Statistical Planning and Inference},
    author = {Székely, Gábor J. and Rizzo, Maria L.},
    month = aug,
    year = {2013},
    pages = {1249--1272},
}

@article{indahl_similarity_2018,
    title = {A similarity index for comparing coupled matrices},
    volume = {32},
    copyright = {Copyright © 2018 John Wiley \& Sons, Ltd.},
    issn = {1099-128X},
    url = {https://onlinelibrary.wiley.com/doi/abs/10.1002/cem.3049},
    doi = {10.1002/cem.3049},
    language = {en},
    number = {10},
    urldate = {2025-09-11},
    journal = {Journal of Chemometrics},
    author = {Indahl, Ulf G. and Næs, Tormod and Liland, Kristian Hovde},
    year = {2018},
    note = {\_eprint: https://analyticalsciencejournals.onlinelibrary.wiley.com/doi/pdf/10.1002/cem.3049},
    pages = {e3049},
}

@article{robin_proc_2011,
    title = {{pROC}: an open-source package for {R} and {S}+ to analyze and compare {ROC} curves},
    volume = {12},
    issn = {1471-2105},
    shorttitle = {{pROC}},
    url = {https://doi.org/10.1186/1471-2105-12-77},
    doi = {10.1186/1471-2105-12-77},
    number = {1},
    urldate = {2025-06-25},
    journal = {BMC Bioinformatics},
    author = {Robin, Xavier and Turck, Natacha and Hainard, Alexandre and Tiberti, Natalia and Lisacek, Frédérique and Sanchez, Jean-Charles and Müller, Markus},
    month = mar,
    year = {2011},
    pages = {77},
}

@incollection{mair_unidimensional_2015,
    title = {Unidimensional {Scaling}},
    copyright = {Copyright © 2015 John Wiley \& Sons, Ltd. All rights reserved.},
    isbn = {978-1-118-44511-2},
    url = {https://onlinelibrary.wiley.com/doi/abs/10.1002/9781118445112.stat06462.pub2},
    language = {en},
    urldate = {2025-06-16},
    booktitle = {Wiley {StatsRef}: {Statistics} {Reference} {Online}},
    publisher = {John Wiley \& Sons, Ltd},
    author = {Mair, Patrick and Leeuw, Jan De},
    year = {2015},
    doi = {10.1002/9781118445112.stat06462.pub2},
    note = {\_eprint: https://onlinelibrary.wiley.com/doi/pdf/10.1002/9781118445112.stat06462.pub2},
    pages = {1--3},
}

@article{ten_berge_greatest_2004,
    title = {The {Greatest} {Lower} {Bound} to the {Reliability} of a {Test} and the {Hypothesis} of {Unidimensionality}},
    volume = {69},
    copyright = {https://www.cambridge.org/core/terms},
    issn = {0033-3123, 1860-0980},
    url = {https://www.cambridge.org/core/product/identifier/S003331230002353X/type/journal_article},
    doi = {10.1007/BF02289858},
    language = {en},
    number = {4},
    urldate = {2025-09-12},
    journal = {Psychometrika},
    author = {Ten Berge, Jos M. F. and Sočan, Gregor},
    month = dec,
    year = {2004},
    pages = {613--625},
}

@article{de_boeck_random_2008,
    title = {Random {Item} {IRT} {Models}},
    volume = {73},
    issn = {1860-0980},
    url = {https://doi.org/10.1007/s11336-008-9092-x},
    doi = {10.1007/s11336-008-9092-x},
    language = {en},
    number = {4},
    urldate = {2025-01-28},
    journal = {Psychometrika},
    author = {De Boeck, Paul},
    month = dec,
    year = {2008},
    pages = {533--559},
}

@article{owen_bi-cross-validation_2009,
    title = {Bi-{Cross}-{Validation} of the {SVD} and the {Nonnegative} {Matrix} {Factorization}},
    volume = {3},
    issn = {1932-6157},
    url = {https://www.jstor.org/stable/30244256},
    number = {2},
    urldate = {2026-02-02},
    journal = {The Annals of Applied Statistics},
    author = {Owen, Art B. and Perry, Patrick O.},
    year = {2009},
    note = {Publisher: Institute of Mathematical Statistics},
    pages = {564--594},
}

@article{mccullagh_resampling_2000,
    title = {Resampling and {Exchangeable} {Arrays}},
    volume = {6},
    issn = {13507265},
    url = {https://www.jstor.org/stable/3318577?origin=crossref},
    doi = {10.2307/3318577},
    language = {en},
    number = {2},
    urldate = {2026-02-03},
    journal = {Bernoulli},
    author = {McCullagh, Peter},
    month = apr,
    year = {2000},
    pages = {285},
}

@misc{owen_bi-cross-validation_2015,
    title = {Bi-cross-validation for factor analysis},
    url = {http://arxiv.org/abs/1503.03515},
    doi = {10.48550/arXiv.1503.03515},
    urldate = {2026-02-04},
    publisher = {arXiv},
    author = {Owen, A. B. and Wang, J.},
    month = nov,
    year = {2015},
    note = {arXiv:1503.03515 [stat]},
}

@article{spearman_general_1904,
    title = {"{General} {Intelligence}," {Objectively} {Determined} and {Measured}},
    volume = {15},
    issn = {0002-9556},
    url = {https://www.jstor.org/stable/1412107},
    doi = {10.2307/1412107},
    number = {2},
    urldate = {2026-02-04},
    journal = {The American Journal of Psychology},
    publisher = {University of Illinois Press},
    author = {Spearman, C.},
    year = {1904},
    pages = {201--292},
}

@misc{oconnor_spearman_2021,
    type = {The {British} {Psychological} {Society}},
    title = {Spearman {Medal} is retired},
    url = {https://www.bps.org.uk/psychologist/spearman-medal-retired},
    language = {en},
    urldate = {2026-02-04},
    journal = {BPS},
    author = {O'Connor, Daryl},
    month = mar,
    year = {2021},
}

@misc{zheng_nonparametric_2025,
    title = {Nonparametric {Factor} {Analysis} and {Beyond}},
    url = {http://arxiv.org/abs/2503.16865},
    doi = {10.48550/arXiv.2503.16865},
    urldate = {2025-08-14},
    publisher = {arXiv},
    author = {Zheng, Yujia and Liu, Yang and Yao, Jiaxiong and Hu, Yingyao and Zhang, Kun},
    month = mar,
    year = {2025},
    note = {arXiv:2503.16865 [cs]},
}

@article{hu_instrumental_2008,
    title = {Instrumental {Variable} {Treatment} of {Nonclassical} {Measurement} {Error} {Models}},
    volume = {76},
    issn = {0012-9682},
    url = {https://www.jstor.org/stable/4502059},
    number = {1},
    urldate = {2026-02-04},
    journal = {Econometrica},
    publisher = {[Wiley, Econometric Society]},
    author = {Hu, Yingyao and Schennach, Susanne M.},
    year = {2008},
    pages = {195--216},
}

@article{hu_econometrics_2017,
    series = {Measurement {Error} {Models}},
    title = {The econometrics of unobservables: {Applications} of measurement error models in empirical industrial organization and labor economics},
    volume = {200},
    issn = {0304-4076},
    shorttitle = {The econometrics of unobservables},
    url = {https://www.sciencedirect.com/science/article/pii/S0304407617300830},
    doi = {10.1016/j.jeconom.2017.06.002},
    number = {2},
    urldate = {2026-02-04},
    journal = {Journal of Econometrics},
    author = {Hu, Yingyao},
    month = oct,
    year = {2017},
    pages = {154--168},
}

@article{ten_berge_numerical_1991,
    title = {A numerical approach to the approximate and the exact minimum rank of a covariance matrix},
    volume = {56},
    issn = {1860-0980},
    url = {https://doi.org/10.1007/BF02294464},
    doi = {10.1007/BF02294464},
    language = {en},
    number = {2},
    urldate = {2026-02-04},
    journal = {Psychometrika},
    author = {ten Berge, Jos M. F. and Kiers, Henk A. L.},
    month = jun,
    year = {1991},
    pages = {309--315},
}

@misc{domingue_intermodel_2021,
    title = {The {InterModel} {Vigorish} ({IMV}) as a flexible and portable approach for quantifying predictive accuracy with binary outcomes},
    url = {https://osf.io/gu3ap},
    doi = {10.31235/osf.io/gu3ap},
    language = {en-us},
    urldate = {2025-01-14},
    publisher = {OSF},
    author = {Domingue, Ben and Rahal, Charles and Faul, Jessica and Freese, Jeremy and Kanopka, Klint and Rigos, Alexandros and Stenhaug, Ben and Tripathi, Ajay},
    month = jun,
    year = {2021},
}

@article{domingue_intermodel_2024,
    title = {The {InterModel} {Vigorish} as a {Lens} for {Understanding} (and {Quantifying}) the {Value} of {Item} {Response} {Models} for {Dichotomously} {Coded} {Items}},
    volume = {89},
    issn = {0033-3123, 1860-0980},
    url = {https://www.cambridge.org/core/journals/psychometrika/article/abs/intermodel-vigorish-as-a-lens-for-understanding-and-quantifying-the-value-of-item-response-models-for-dichotomously-coded-items/F61C75F6F945A5B13F73C6128EB83998},
    doi = {10.1007/s11336-024-09977-2},
    language = {en},
    number = {3},
    urldate = {2026-02-10},
    journal = {Psychometrika},
    author = {Domingue, Benjamin W. and Kanopka, Klint and Kapoor, Radhika and Pohl, Steffi and Chalmers, R. Philip and Rahal, Charles and Rhemtulla, Mijke},
    month = sep,
    year = {2024},
    pages = {1034--1054},
}

@article{stenhaug_predictive_2022,
    title = {Predictive {Fit} {Metrics} for {Item} {Response} {Models}},
    volume = {46},
    issn = {0146-6216},
    url = {https://doi.org/10.1177/01466216211066603},
    doi = {10.1177/01466216211066603},
    language = {EN},
    number = {2},
    urldate = {2026-02-10},
    journal = {Applied Psychological Measurement},
    publisher = {SAGE Publications Inc},
    author = {Stenhaug, Benjamin A. and Domingue, Benjamin W.},
    month = mar,
    year = {2022},
    pages = {136--155},
}

@article{steiger_structural_1990,
    title = {Structural {Model} {Evaluation} and {Modification}: {An} {Interval} {Estimation} {Approach}},
    volume = {25},
    issn = {0027-3171},
    shorttitle = {Structural {Model} {Evaluation} and {Modification}},
    url = {https://doi.org/10.1207/s15327906mbr2502_4},
    doi = {10.1207/s15327906mbr2502_4},
    number = {2},
    urldate = {2026-02-10},
    journal = {Multivariate Behavioral Research},
    publisher = {Routledge},
    author = {Steiger, James H.},
    month = apr,
    year = {1990},
    note = {\_eprint: https://doi.org/10.1207/s15327906mbr2502\_4},
    pages = {173--180},
}

@article{casabianca_irt_2015,
    title = {{IRT} {Item} {Parameter} {Recovery} {With} {Marginal} {Maximum} {Likelihood} {Estimation} {Using} {Loglinear} {Smoothing} {Models}},
    volume = {40},
    issn = {1076-9986, 1935-1054},
    url = {https://journals.sagepub.com/doi/10.3102/1076998615606112},
    doi = {10.3102/1076998615606112},
    language = {en},
    number = {6},
    urldate = {2026-02-10},
    journal = {Journal of Educational and Behavioral Statistics},
    author = {Casabianca, Jodi M. and Lewis, Charles},
    month = dec,
    year = {2015},
    pages = {547--578},
}

@article{cronbach_my_2004,
    title = {My {Current} {Thoughts} on {Coefficient} {Alpha} and {Successor} {Procedures}},
    volume = {64},
    issn = {0013-1644},
    url = {https://doi.org/10.1177/0013164404266386},
    doi = {10.1177/0013164404266386},
    language = {EN},
    number = {3},
    urldate = {2026-01-19},
    journal = {Educational and Psychological Measurement},
    publisher = {SAGE Publications Inc},
    author = {Cronbach, Lee J. and Shavelson, Richard J.},
    month = jun,
    year = {2004},
    pages = {391--418},
}

@article{heller_consistent_2013,
    title = {A consistent multivariate test of association based on ranks of distances},
    volume = {100},
    issn = {0006-3444},
    url = {https://doi.org/10.1093/biomet/ass070},
    doi = {10.1093/biomet/ass070},
    number = {2},
    urldate = {2025-01-16},
    journal = {Biometrika},
    author = {Heller, Ruth and Heller, Yair and Gorfine, Malka},
    month = jun,
    year = {2013},
    pages = {503--510},
}

@article{busing_monotone_2022,
    title = {Monotone {Regression}: {A} {Simple} and {Fast} {O}(n) {PAVA} {Implementation}},
    volume = {102},
    copyright = {Copyright (c) 2022 Frank M. T. A. Busing},
    issn = {1548-7660},
    shorttitle = {Monotone {Regression}},
    url = {https://doi.org/10.18637/jss.v102.c01},
    doi = {10.18637/jss.v102.c01},
    language = {en},
    urldate = {2025-07-03},
    journal = {Journal of Statistical Software},
    author = {Busing, Frank M. T. A.},
    month = may,
    year = {2022},
    pages = {1--25},
}

@article{josse_measuring_2016,
    title = {Measuring multivariate association and beyond},
    volume = {10},
    issn = {1935-7516},
    url = {https://pmc.ncbi.nlm.nih.gov/articles/PMC5658146/},
    doi = {10.1214/16-SS116},
    urldate = {2025-10-17},
    journal = {Statistics surveys},
    author = {Josse, Julie and Holmes, Susan},
    year = {2016},
    pages = {132--167},
}

@misc{josse_measures_2014,
    title = {Measures of dependence between random vectors and tests of independence. {Literature} review},
    url = {http://arxiv.org/abs/1307.7383},
    doi = {10.48550/arXiv.1307.7383},
    urldate = {2025-08-30},
    publisher = {arXiv},
    author = {Josse, Julie and Holmes, Susan},
    month = aug,
    year = {2014},
    note = {arXiv:1307.7383 [stat]},
}

@article{jensen_what_1994,
    title = {What is a good g?},
    volume = {18},
    issn = {0160-2896},
    url = {https://www.sciencedirect.com/science/article/pii/0160289694900299},
    doi = {10.1016/0160-2896(94)90029-9},
    number = {3},
    urldate = {2026-02-14},
    journal = {Intelligence},
    author = {Jensen, Arthur R. and Weng, Li-Jen},
    month = may,
    year = {1994},
    pages = {231--258},
}

@article{maccallum_representing_1991,
    address = {US},
    title = {Representing sources of error in the common-factor model: {Implications} for theory and practice},
    volume = {109},
    issn = {1939-1455},
    shorttitle = {Representing sources of error in the common-factor model},
    doi = {10.1037/0033-2909.109.3.502},
    number = {3},
    journal = {Psychological Bulletin},
    publisher = {American Psychological Association},
    author = {MacCallum, Robert C. and Tucker, Ledyard R.},
    year = {1991},
    pages = {502--511},
}

@incollection{maccallum_factor_2007,
    title = {Factor {Analysis} {Models} as {Approximations}},
    booktitle = {Factor {Analysis} at 100},
    publisher = {Routledge},
    author = {MacCallum, Robert C. and Browne, Michael W. and Cai, Li},
    year = {2007},
    note = {Num Pages: 24},
}

@article{wherry_factor_1944,
    title = {Factor {Pattern} of {Test} {Items} and {Tests} as a {Function} of the {Correlation} {Coefficient}: {Content}, {Difficulty}, and {Constant} {Error} {Factors}},
    volume = {9},
    copyright = {https://www.cambridge.org/core/terms},
    issn = {0033-3123, 1860-0980},
    shorttitle = {Factor {Pattern} of {Test} {Items} and {Tests} as a {Function} of the {Correlation} {Coefficient}},
    url = {https://www.cambridge.org/core/product/identifier/S0033312300044513/type/journal_article},
    doi = {10.1007/BF02288734},
    language = {en},
    number = {4},
    urldate = {2026-02-16},
    journal = {Psychometrika},
    author = {Wherry, Robert J. and Gaylord, Richard H.},
    month = dec,
    year = {1944},
    pages = {237--244},
}

@article{gronneberg_partial_2020,
    title = {Partial {Identification} of {Latent} {Correlations} with {Binary} {Data}},
    volume = {85},
    issn = {0033-3123, 1860-0980},
    url = {https://www.cambridge.org/core/journals/psychometrika/article/abs/partial-identification-of-latent-correlations-with-binary-data/B3AF024BA0F5689BFB3224D261218480#article},
    doi = {10.1007/s11336-020-09737-y},
    language = {en},
    number = {4},
    urldate = {2025-08-20},
    journal = {Psychometrika},
    author = {Grønneberg, Steffen and Moss, Jonas and Foldnes, Njål},
    month = dec,
    year = {2020},
    pages = {1028--1051},
}

@article{muthen_testing_1988,
    title = {Testing the {Assumptions} {Underlying} {Tetrachoric} {Correlations}},
    volume = {53},
    issn = {0033-3123, 1860-0980},
    url = {https://www.cambridge.org/core/journals/psychometrika/article/abs/testing-the-assumptions-underlying-tetrachoric-correlations/08A71979113F986F1164B359E6D99020},
    doi = {10.1007/BF02294408},
    language = {en},
    number = {4},
    urldate = {2025-07-30},
    journal = {Psychometrika},
    author = {Muthén, Bengt and Hofacker, Charles},
    month = dec,
    year = {1988},
    pages = {563--577},
}

@article{cudeck_model_1991,
    address = {US},
    title = {Model selection in covariance structures analysis and the "problem" of sample size: {A} clarification},
    volume = {109},
    issn = {1939-1455},
    shorttitle = {Model selection in covariance structures analysis and the "problem" of sample size},
    doi = {10.1037/0033-2909.109.3.512},
    number = {3},
    journal = {Psychological Bulletin},
    publisher = {American Psychological Association},
    author = {Cudeck, Robert and Henly, Susan J.},
    year = {1991},
    pages = {512--519},
}

@article{thurstone_multiple_1931,
    title = {Multiple factor analysis.},
    volume = {38},
    number = {5},
    journal = {Psychological review},
    publisher = {Psychological Review Company},
    author = {Thurstone, Louis Leon},
    year = {1931},
    pages = {406},
}

@article{yalcin_nonlinear_2001,
    title = {Nonlinear {Factor} {Analysis} as a {Statistical} {Method}},
    volume = {16},
    issn = {0883-4237},
    url = {https://www.jstor.org/stable/2676693},
    number = {3},
    urldate = {2026-02-16},
    journal = {Statistical Science},
    publisher = {Institute of Mathematical Statistics},
    author = {Yalcin, Ilker and Amemiya, Yasuo},
    year = {2001},
    pages = {275--294},
}

@inproceedings{amemiya_model_1997,
    address = {New York, NY},
    title = {Model fitting procedures for nonlinear factor analysis using the errors-in-variables parameterization},
    isbn = {978-1-4612-1842-5},
    doi = {10.1007/978-1-4612-1842-5_10},
    language = {en},
    booktitle = {Latent {Variable} {Modeling} and {Applications} to {Causality}},
    publisher = {Springer},
    author = {Amemiya, Yasuo and Yalcin, Ilker},
    editor = {Berkane, Maia},
    year = {1997},
    pages = {195--210},
}

@article{lee_learning_1999,
    title = {Learning the parts of objects by non-negative matrix factorization},
    volume = {401},
    number = {6755},
    journal = {nature},
    publisher = {Nature Publishing Group UK London},
    author = {Lee, Daniel D and Seung, H Sebastian},
    year = {1999},
    pages = {788--791},
}

@article{mosteller_useful_1946,
    title = {On {Some} {Useful} "{Inefficient}" {Statistics}},
    volume = {17},
    issn = {0003-4851},
    url = {https://www.jstor.org/stable/2236081},
    number = {4},
    urldate = {2025-10-09},
    journal = {The Annals of Mathematical Statistics},
    publisher = {Institute of Mathematical Statistics},
    author = {Mosteller, Frederick},
    year = {1946},
    pages = {377--408},
}

@article{blomqvist_measure_1950,
    title = {On a {Measure} of {Dependence} {Between} two {Random} {Variables}},
    volume = {21},
    issn = {0003-4851, 2168-8990},
    url = {https://projecteuclid.org/journals/annals-of-mathematical-statistics/volume-21/issue-4/On-a-Measure-of-Dependence-Between-two-Random-Variables/10.1214/aoms/1177729754.full},
    doi = {10.1214/aoms/1177729754},
    number = {4},
    urldate = {2025-10-08},
    journal = {The Annals of Mathematical Statistics},
    publisher = {Institute of Mathematical Statistics},
    author = {Blomqvist, Nils},
    month = dec,
    year = {1950},
    pages = {593--600},
}

@article{schmid_development_1957,
    title = {The development of hierarchical factor solutions},
    volume = {22},
    issn = {1860-0980},
    url = {https://doi.org/10.1007/BF02289209},
    doi = {10.1007/BF02289209},
    language = {en},
    number = {1},
    urldate = {2026-03-06},
    journal = {Psychometrika},
    author = {Schmid, John and Leiman, John M.},
    month = mar,
    year = {1957},
    pages = {53--61},
}

@article{wolff_exploring_2005,
    title = {Exploring item and higher order factor structure with the {Schmid}-{Leiman} solution: {Syntax} codes for {SPSS} and {SAS}},
    volume = {37},
    issn = {1554-3528},
    shorttitle = {Exploring item and higher order factor structure with the {Schmid}-{Leiman} solution},
    url = {https://doi.org/10.3758/BF03206397},
    doi = {10.3758/BF03206397},
    language = {en},
    number = {1},
    urldate = {2026-03-06},
    journal = {Behavior Research Methods},
    author = {Wolff, Hans -Georg and Preising, Katja},
    month = feb,
    year = {2005},
    pages = {48--58},
}

@article{cohen_coefficient_1960,
    title = {A {Coefficient} of {Agreement} for {Nominal} {Scales}},
    volume = {20},
    issn = {0013-1644},
    url = {https://doi.org/10.1177/001316446002000104},
    doi = {10.1177/001316446002000104},
    language = {EN},
    number = {1},
    urldate = {2026-03-06},
    journal = {Educational and Psychological Measurement},
    publisher = {SAGE Publications Inc},
    author = {Cohen, Jacob},
    month = apr,
    year = {1960},
    pages = {37--46},
}

@article{byrt_bias_1993,
    title = {Bias, prevalence and kappa},
    volume = {46},
    issn = {0895-4356},
    url = {https://www.sciencedirect.com/science/article/pii/089543569390018V},
    doi = {10.1016/0895-4356(93)90018-V},
    number = {5},
    urldate = {2026-03-05},
    journal = {Journal of Clinical Epidemiology},
    author = {Byrt, Ted and Bishop, Janet and Carlin, John B.},
    month = may,
    year = {1993},
    pages = {423--429},
}

@article{holley_note_1964,
    title = {A {Note} on the {G} {Index} of {Agreement}},
    volume = {24},
    issn = {0013-1644},
    url = {https://doi.org/10.1177/001316446402400402},
    doi = {10.1177/001316446402400402},
    language = {EN},
    number = {4},
    urldate = {2026-03-06},
    journal = {Educational and Psychological Measurement},
    publisher = {SAGE Publications Inc},
    author = {Holley, J.W. and Guilford, J.P.},
    month = dec,
    year = {1964},
    pages = {749--753},
}

@article{hamann_merkmalsbestand_1961,
    title = {Merkmalsbestand und {Verwandtschaftsbeziehungen} der {Farinosae}: {Ein} {Beitrag} zum {System} der {Monokotyledonen}},
    volume = {2},
    issn = {0511-9618},
    shorttitle = {Merkmalsbestand und {Verwandtschaftsbeziehungen} der {Farinosae}},
    url = {https://www.jstor.org/stable/3995266},
    number = {5},
    urldate = {2026-03-06},
    journal = {Willdenowia},
    publisher = {Botanischer Garten und Botanisches Museum, Berlin-Dahlem},
    author = {Hamann, Ulrich},
    year = {1961},
    pages = {639--768},
}

@article{cheetham_binary_1969,
    title = {Binary ({Presence}-{Absence}) {Similarity} {Coefficients}},
    volume = {43},
    issn = {0022-3360},
    url = {https://www.jstor.org/stable/1302424},
    number = {5},
    urldate = {2026-03-06},
    journal = {Journal of Paleontology},
    publisher = {Paleontological Society},
    author = {Cheetham, Alan H. and Hazel, Joseph E.},
    year = {1969},
    pages = {1130--1136},
}

@article{fuxman_bass_using_2013,
    title = {Using networks to measure similarity between genes: association index selection},
    volume = {10},
    issn = {1548-7091},
    shorttitle = {Using networks to measure similarity between genes},
    url = {https://pmc.ncbi.nlm.nih.gov/articles/PMC3959882/},
    doi = {10.1038/nmeth.2728},
    number = {12},
    urldate = {2026-03-06},
    journal = {Nature methods},
    author = {Fuxman Bass, Juan I and Diallo, Alos and Nelson, Justin and Soto, Juan M and Myers, Chad L and Walhout, Albertha J M},
    month = dec,
    year = {2013},
    pages = {1169--1176},
}

@misc{zhang_intermodel_2023,
    title = {The {InterModel} {Vigorish} for {Model} {Comparison} in {Confirmatory} {Factor} {Analysis} with {Binary} {Outcomes}},
    url = {https://osf.io/preprints/psyarxiv/tv9bd_v1/},
    doi = {10.31234/osf.io/tv9bd_v1},
    urldate = {2026-03-10},
    publisher = {PsyArXiv},
    author = {Zhang, Lijin and Rahal, Charles and Kanopka, Klint and Ulitzsch, Esther and Zhang, Zhiyong and Domingue, Benjamin},
    month = sep,
    year = {2023},
}

@article{pearson_i_1900,
    title = {I. {Mathematical} contributions to the theory of evolution. —{VII}. {On} the correlation of characters not quantitatively measurable},
    volume = {195},
    issn = {0264-3952},
    url = {https://doi.org/10.1098/rsta.1900.0022},
    doi = {10.1098/rsta.1900.0022},
    number = {262-273},
    urldate = {2026-03-16},
    journal = {Philosophical Transactions of the Royal Society of London, Series A: Containing Papers of a Mathematical or Physical Character},
    author = {Pearson, Karl},
    month = jan,
    year = {1900},
    pages = {1--47},
}

@article{thomson_hierarchy_1916,
    title = {A {Hierarchy} {Without} a {General} {Factor}},
    volume = {8},
    copyright = {1916 The British Psychological Society},
    issn = {2044-8295},
    url = {https://onlinelibrary.wiley.com/doi/abs/10.1111/j.2044-8295.1916.tb00133.x},
    doi = {10.1111/j.2044-8295.1916.tb00133.x},
    language = {en},
    number = {3},
    urldate = {2026-03-16},
    journal = {British Journal of Psychology, 1904-1920},
    author = {Thomson, Godfrey H.},
    year = {1916},
    note = {\_eprint: https://bpspsychub.onlinelibrary.wiley.com/doi/pdf/10.1111/j.2044-8295.1916.tb00133.x},
    pages = {271--281},
}

@book{ll_thurstone_vectors_1935,
    title = {The {Vectors} {Of} {Mind} {Multiple} {Factor} {Analysis} {For} {The} {Isolation} {Of} {Primary} {Traits}},
    url = {http://archive.org/details/vectorsofmindmul010122mbp},
    language = {eng},
    urldate = {2026-02-22},
    publisher = {The University Of Chicago Press},
    author = {{L.L. Thurstone}},
    collaborator = {{Universal Digital Library}},
    year = {1935},
}

@article{rodriguez_evaluating_2016,
    address = {US},
    title = {Evaluating bifactor models: {Calculating} and interpreting statistical indices},
    volume = {21},
    issn = {1939-1463},
    shorttitle = {Evaluating bifactor models},
    doi = {10.1037/met0000045},
    number = {2},
    journal = {Psychological Methods},
    publisher = {American Psychological Association},
    author = {Rodriguez, Anthony and Reise, Steven P. and Haviland, Mark G.},
    year = {2016},
    pages = {137--150},
}

@article{yang_contraction_2019,
    title = {Contraction and uniform convergence of isotonic regression},
    volume = {13},
    issn = {1935-7524, 1935-7524},
    url = {https://projecteuclid.org/journals/electronic-journal-of-statistics/volume-13/issue-1/Contraction-and-uniform-convergence-of-isotonic-regression/10.1214/18-EJS1520.full},
    doi = {10.1214/18-EJS1520},
    language = {en},
    number = {1},
    urldate = {2026-03-18},
    journal = {Electronic Journal of Statistics},
    publisher = {Institute of Mathematical Statistics and Bernoulli Society},
    author = {Yang, Fan and Barber, Rina Foygel},
    month = jan,
    year = {2019},
    pages = {646--677},
}

% \begin{refsection}[secondary]
%     \nocite{*} 
%     \printbibliography[title={Dataset References}]
% \end{refsection}

\appendix

\begin{figure}
    \centering
    \includegraphics[width=1\linewidth]{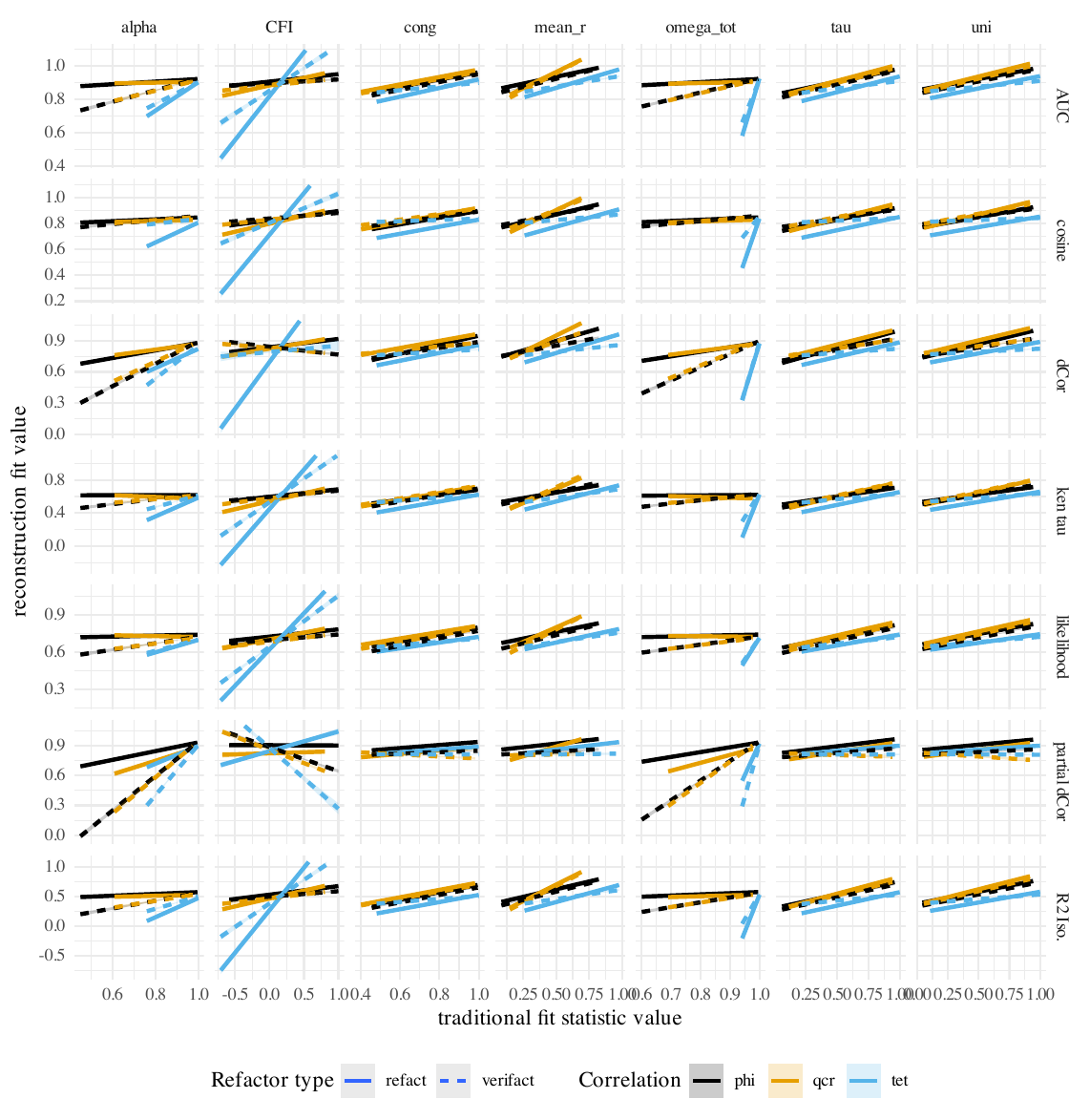}
    \caption{\textbf{Simulation Study I: Panel of Crossed Traditional and Refactor metrics} Traditional Measures of Unidimensionality vs. Refactor and Verifactor Reconstruction Metrics using Tetrachoric, Pearson, and Quadrant Correlations. \footnotesize{Aside from average interitem correlation and proportion of variance accounted for by the factor model, there is no meaningful relationship between the traditional measures of unidimensionality explored in this paper and the ability of the relationship to reconstruct the imaged signal across over 200 publicly available datasets. %Each point represents a separate publicly available dataset. 
    Each column represents a traditional metric of unidimensionality based on factor analyses: alpha = Cronbach's $\alpha$, av\_r = average interitem correlation $\bar{r}_{ij}$; CFI = Comparative Fit Index; rho\_c = a measure of the fit of a congeneric model to the observed correlations $\rho_c$; $\tau_{RC}$ = compares the observed correlations $r_{ij}$ to the mean correlation $\bar{r}_{ij}$ and considers 1 - the ratio of the sum of the squared residuals to the sum of the squared correlations; TLI = Tucker Lewis Index (tli); and u = Revelle and Condon's measure of congeneric unidimensionality, $u_{RC}$ (see \cite{revelle_unidim_2025}. Each row represents Refactor Reconstruction Metrics. X-axes represent the strength of the traditional metric for the dataset. Y-axes are the fit of the out-of-sample reconstructions via Refactor Analyses. Slopes are estimated using robust regression. Given the item-subject random effects used, weak relationships may also be capturing the sensitivity of the scale of the instrument being used to presumed fixed effects of items. }
    }
    \label{fig:sim1full}
\end{figure}

\begin{figure}
    \centering
    \includegraphics[width=1\linewidth]{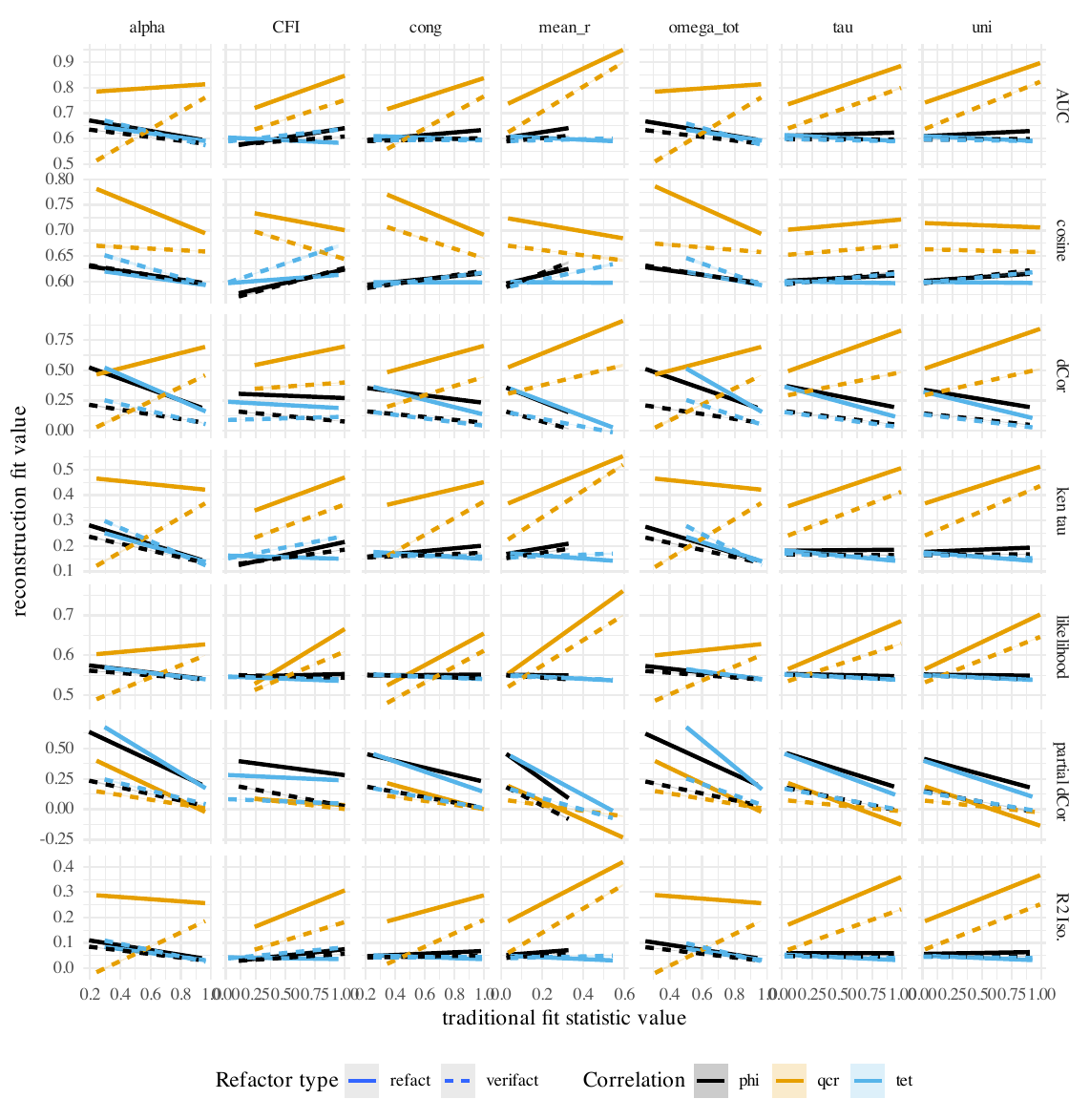}
    \caption{\textbf{Simulation Study II: Panel of Crossed Traditional and Refactor metrics} 
    Traditional Measures of Unidimensionality vs. Refactor and Verifactor Reconstruction Metrics using Tetrachoric, Pearson, and Quadrant Correlations. \footnotesize{Aside from average interitem correlation and proportion of variance accounted for by the factor model, there is no meaningful relationship between the traditional measures of unidimensionality explored in this paper and the ability of the relationship to reconstruct the imaged signal across over 200 publicly available datasets. %Each point represents a separate publicly available dataset. 
    Each column represents a traditional metric of unidimensionality based on factor analyses: alpha = Cronbach's $\alpha$, av\_r = average interitem correlation $\bar{r}_{ij}$; CFI = Comparative Fit Index; rho\_c = a measure of the fit of a congeneric model to the observed correlations $\rho_c$; $\tau_{RC}$ = compares the observed correlations $r_{ij}$ to the mean correlation $\bar{r}_{ij}$ and considers 1 - the ratio of the sum of the squared residuals to the sum of the squared correlations; TLI = Tucker Lewis Index (tli); and u = Revelle and Condon's measure of congeneric unidimensionality, $u_{RC}$ (see \cite{revelle_unidim_2025}. Each row represents Refactor Reconstruction Metrics. X-axes represent the strength of the traditional metric for the dataset. Y-axes are the fit of the out-of-sample reconstructions via Refactor Analyses. Slopes are estimated using robust regression. Given the item-subject random effects used, weak relationships may also be capturing the sensitivity of the scale of the instrument being used to presumed fixed effects of items. 
    }}
    \label{fig:sim2full}
\end{figure}

\begin{figure}
    \centering
    \includegraphics[width=1\linewidth]{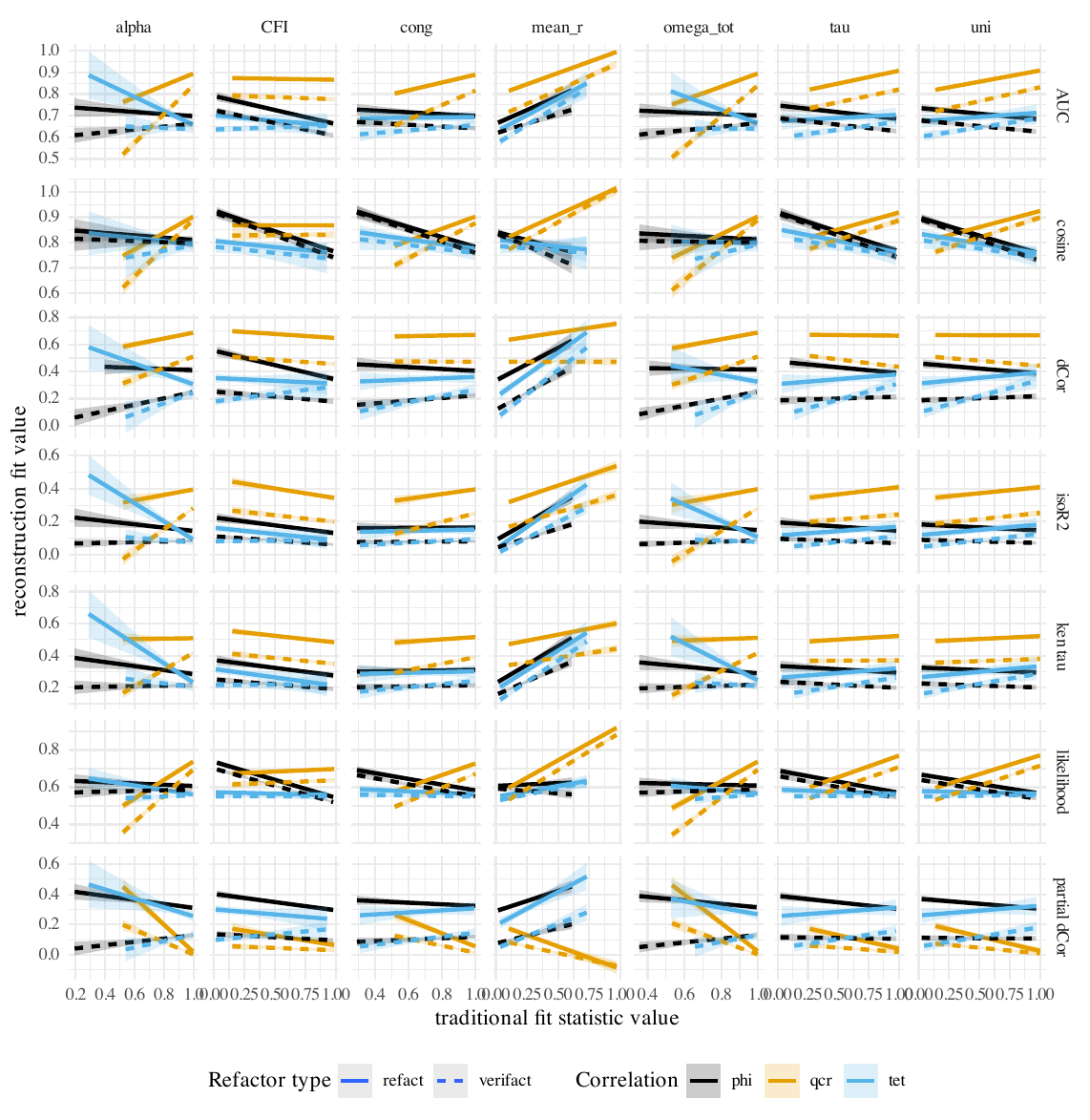}
    \caption{\textbf{Empirical Study: Panel of Crossed Traditional and Refactor metrics} 
    Traditional Measures of Unidimensionality vs. Refactor and Verifactor Reconstruction Metrics using Tetrachoric, Pearson, and Quadrant Correlations. \footnotesize{Aside from average interitem correlation and proportion of variance accounted for by the factor model, there is no meaningful relationship between the traditional measures of unidimensionality explored in this paper and the ability of the relationship to reconstruct the imaged signal across over 200 publicly available datasets. %Each point represents a separate publicly available dataset. 
    Each column represents a traditional metric of unidimensionality based on factor analyses: alpha = Cronbach's $\alpha$, av\_r = average interitem correlation $\bar{r}_{ij}$; CFI = Comparative Fit Index; rho\_c = a measure of the fit of a congeneric model to the observed correlations $\rho_c$; $\tau_{RC}$ = compares the observed correlations $r_{ij}$ to the mean correlation $\bar{r}_{ij}$ and considers 1 - the ratio of the sum of the squared residuals to the sum of the squared correlations; TLI = Tucker Lewis Index (tli); and u = Revelle and Condon's measure of congeneric unidimensionality, $u_{RC}$ (see \cite{revelle_unidim_2025}. Each row represents Refactor Reconstruction Metrics. X-axes represent the strength of the traditional metric for the dataset. Y-axes are the fit of the out-of-sample reconstructions via Refactor Analyses. Slopes are estimated using robust regression. Given the item-subject random effects used, weak relationships may also be capturing the sensitivity of the scale of the instrument being used to presumed fixed effects of items. 
    }}
    \label{fig:empfull}
\end{figure}

\section{Formal Section for Contributions}

\subsection{Refactor: a direct data-level fit criterion}
\label{app:refactor-proof}

This section formalizes the core contribution: Refactor evaluation is a \emph{test of recoverability of $X$ from a rank--$k$ representation learned from an image}, rather than a test of fit on the image itself.

\begin{proposition}[Impossibility of Continuous, Injective Dimensionality Reduction]
A one-to-one function $g: \mathbb{R}^p \to \mathbb{R}^k$ with $k < p$ cannot be continuous. [Proof \ref{prf:inject}]
\end{proposition}
% \begin{toappendix}
\begin{proof}\label{prf:inject}
This is a direct consequence of the Borsuk-Ulam theorem, which states that for any continuous function $g: S^k \to \mathbb{R}^k$, there exists a point $\boldsymbol{x} \in S^k$ such that $g(\boldsymbol{x}) = g(-\boldsymbol{x})$. Consider the sphere $S^k \subset \mathbb{R}^{k+1}$. Since $k < p$, we can embed $S^k$ in the domain $\mathbb{R}^p$. If we assume $g$ is a continuous function from $\mathbb{R}^p \to \mathbb{R}^k$, its restriction to $S^k$ is also continuous. By the Borsuk-Ulam theorem, there must exist antipodal points $\boldsymbol{x}$ and $-\boldsymbol{x}$ on $S^k$ such that $g(\boldsymbol{x}) = g(-\boldsymbol{x})$. As $\boldsymbol{x} \neq -\boldsymbol{x}$, this contradicts the assumption that $g$ is one-to-one. Therefore, no such continuous, injective, one-to-one mapping function can exist.
\end{proof}
% \end{toappendix}

\subsubsection{Rank-1 SVD Formulation}
Let $\boldsymbol{X} \in \mathbb{R}^{n \times p}$ be the data matrix. A rank-1 approximation of $\boldsymbol{X}$ takes the form $\boldsymbol{\hat{X}} = \boldsymbol{u}\boldsymbol{v}^T$, where $\boldsymbol{u} \in \mathbb{R}^n$ and $\boldsymbol{v} \in \mathbb{R}^p$. The Eckart-Young-Mirsky theorem states that the best rank-$k$ approximation to $\boldsymbol{X}$ in the Frobenius norm is given by its truncated Singular Value Decomposition (SVD). For the rank-1 case, this is:
\begin{equation}
\boldsymbol{\hat{X}}^{(1)} = \underset{\text{rank}(\boldsymbol{M})=1}{\arg\min} \|\boldsymbol{X} - \boldsymbol{M}\|_F^2 = \sigma_1 \boldsymbol{u}_1 \boldsymbol{v}_1^T
\end{equation}
where $\sigma_1$ is the largest singular value of $\boldsymbol{X}$, and $\boldsymbol{u}_1$ and $\boldsymbol{v}_1$ are the corresponding left and right singular vectors.

A core insight of Refactor analysis is to recognize that $\boldsymbol{u}_1$ and $\boldsymbol{v}_1$ are also the principal eigenvectors of the Gram-functioning matrix $\boldsymbol{X}\boldsymbol{X}^T$ and the covariance-like matrix $\boldsymbol{X}^T\boldsymbol{X}$, respectively. This establishes a direct link between the conventional "image-based" analysis (which examines $\boldsymbol{X}^T\boldsymbol{X}$) and a full data reconstruction.

\begin{proposition}[The Refactor Reconstruction]
Given a data matrix $\boldsymbol{X}$, let its rank-1 SVD be $\sigma_1 \boldsymbol{u}_1 \boldsymbol{v}_1^T$. The vector $\boldsymbol{u}_1$ represents the component loadings based on the rows (e.g., persons), derived from the factorized eigensystem of the covariance-like matrix $\boldsymbol{X}\boldsymbol{X}^T$. The vector $\boldsymbol{v}_1$ represents the component loadings of the columns (e.g., items), derived from the factorized eigensystem of the covariance-like matrix $\boldsymbol{X}^T\boldsymbol{X}$. The Refactor reconstruction, $\boldsymbol{\hat{X}}$, is the rank-1 matrix formed by the outer product of these two principal vectors, scaled by $\sigma_1$.
\end{proposition}

This process allows for the use of various association matrices. Instead of $\boldsymbol{X}^T\boldsymbol{X}$, one can construct a matrix $\boldsymbol{A}$ whose entries $A_{ij}$ represent a chosen measure of association (e.g., Pearson, tetrachoric, distance correlation) between columns $i$ and $j$ of $\boldsymbol{X}$. The principal eigenvector of $\boldsymbol{A}$ serves as the loading vector $\boldsymbol{v}$, which could also then be used to generate scores $\boldsymbol{u}$ (e.g., via regression) and subsequently reconstruct $\boldsymbol{\hat{X}}$.

\subsubsection{Rank--1 recovery and the necessity of dual (row/column) structure}

\begin{proposition}[Row/column duality for exact rank--$k$ matrices]
\label{prop:rowcol-duality}
Let $S\in\mathbb{R}^{n\times p}$ have rank $k$ with compact SVD $S=U\Sigma V^\top$, where $U\in\mathbb{R}^{n\times k}$ and $V\in\mathbb{R}^{p\times k}$ have orthonormal columns and $\Sigma\in\mathbb{R}^{k\times k}$ is diagonal with positive entries. Then:
\begin{enumerate}\itemsep2pt
\item The rank--$k$ eigenspaces of $SS^\top$ and $S^\top S$ are $\mathrm{span}(U)$ and $\mathrm{span}(V)$, respectively.
\item $S$ is uniquely determined by $(SS^\top,\, S^\top S)$ up to the usual orthogonal indeterminacy: if $S=\widetilde U \Sigma \widetilde V^\top$ is another SVD, then $\widetilde U=UQ$, $\widetilde V=VQ$ for some orthogonal $Q\in\mathbb{R}^{k\times k}$. [Proof \ref{prf:rcdual}]
\end{enumerate}
\end{proposition}
% \begin{toappendix}
\begin{proof}\label{prf:rcdual}
(1) is standard: $SS^\top = U\Sigma^2 U^\top$ and $S^\top S = V\Sigma^2 V^\top$.
(2) follows because the nonzero eigenvalues of $SS^\top$ and $S^\top S$ coincide and equal the squared singular values; their associated eigenspaces are $U$ and $V$ up to a shared orthogonal transform $Q$. Substituting gives the stated indeterminacy class.
\end{proof}
% \end{toappendix}
Proposition~\ref{prop:rowcol-duality} motivates Refactor’s defining move: any claim that ``a single factor explains the data'' is simultaneously a claim about \emph{row structure} (persons/observations align along one latent direction) and \emph{column structure} (items/features align along one latent direction). Image-only diagnostics often quantify only the column-side claim (e.g., via $X^\top X$), \textbf{\textit{while Refactor insists on reconstructing $X$ using both.}}

\subsubsection{Refactor fit as a recoverability functional}

Let $\mathcal{M}_k$ denote a class of rank--$k$ reconstructions induced by a chosen image, estimator, and refactoring map. Define a \emph{Refactor functional} by
\begin{equation}
\mathrm{RF}_m(k;X)\;=\; m(X,\widehat X_k),
\qquad \widehat X_k\in\mathcal{M}_k,
\label{eq:app-refactor-functional}
\end{equation}
where $m(\cdot,\cdot)$ may be a loss (smaller is better) or an association (larger is better). The paper instantiates $m$ with predictive (AUC, cross-entropy), order-preserving (Kendall $\tau$, isotonic $R^2$), geometric (matrix cosine), and dependence-sensitive (bias-corrected distance correlation, generalized coefficient of determination, Ramsay/Yanai coefficients) metrics.

\begin{theorem}[Refactor consistency under correct low-rank structure]
\label{thm:refactor-consistency}
Assume $X=S+E$ where $S$ has rank $k$ and $\|E\|_F$ is small. Suppose the estimator $\textsf{Z}$ applied to the chosen images yields loadings whose column spaces converge to the true row/column spaces of $S$ in the sense that there exist orthogonal $Q_r,Q_c\in\mathbb{R}^{k\times k}$ such that
\[
\|B_k - U Q_r\|_F \to 0,\qquad \|F_k - V Q_c\|_F \to 0,
\]
and suppose the refactoring map $\mathcal{R}_k$ is continuous in $(B_k,F_k)$ and exact on the noiseless target $S$ up to the same orthogonal indeterminacy (i.e., $\mathcal{R}_k(S;UQ,VQ)=S$ for all orthogonal $Q$). Then $\widehat X_k \to S$ in Frobenius norm as $\|E\|_F\to 0$, hence any continuous metric $m$ satisfies $m(X,\widehat X_k)\to m(S,S)$. [Proof \ref{prf:refcons}]
\end{theorem}
% \begin{toappendix}
\begin{proof}\label{prf:refcons}
Write $\widehat X_k=\mathcal{R}_k(X;B_k,F_k)$. By continuity of $\mathcal{R}_k$ and convergence of $(B_k,F_k)$ to $(U,V)$ up to orthogonal transforms, $\widehat X_k$ converges to $\mathcal{R}_k(S;UQ,VQ)=S$ as noise vanishes. Continuity of $m$ completes the claim.
\end{proof}
% \end{toappendix}
Theorem~\ref{thm:refactor-consistency} provides the basic inferential reading: when the data are truly well described by a rank--$k$ correlational signal, Refactor scores approach their optimal values. Conversely, when Refactor scores remain poor even as image-based fit looks adequate, the evidence points to a mismatch between the \emph{associational image} and the \emph{data-level signal} (e.g., nonlinear structure not captured by the chosen association, violations of conditional independence, mixtures, or strong idiosyncratic effects).

\subsection{Verifactor: bi-cross-validated refactoring for random rows and random columns}
\label{app:verifactor-proof}

% \subsection{Related Work}
The original extension from Nonnegative Matrix Factorization (NMF) to Factor analysis was motivated for the removal of
unwanted variation from microarray experiments \citep{owen_bi-cross-validation_2015}. Factor analyses themselves originally stem from the (problematically motivated \cite{oconnor_spearman_2021}) statistical methods from \cite{spearman_general_1904}. Suppose that the ideal data for estimating a model consists of a sample of (X, Z), but the researcher only observes X. Our objective is to identify the latent variable(s) Z under the most general conditions.

We identify latent variables in nonparametric models with nonlinear generating processes based on the so-called Hu-Schennach Theorem \citep{zheng_nonparametric_2025,hu_instrumental_2008,hu_econometrics_2017}, even when confronted with non-negligible noise.

\subsubsection{Self-consistency of rank--$k$ matrices under block holdout}

Partition (after row/column permutation) a matrix $X$ into blocks
\begin{equation}
X=
\begin{pmatrix}
A & B\\
C & D
\end{pmatrix},
\label{eq:app-block-partition}
\end{equation}
where $A\in\mathbb{R}^{r\times s}$ is held out, $D\in\mathbb{R}^{(n-r)\times(p-s)}$ is held in, and $B,C$ are the off-diagonal blocks.

\begin{lemma}[Self-consistency of exact rank--$k$ structure {\citep{owen_bi-cross-validation_2009}}]
\label{lem:self-consistency}
If $X$ has rank $k$ and $D$ has rank $k$, then
\begin{equation}
A \;=\; B\,D^\dagger\,C,
\label{eq:app-self-consistency}
\end{equation}
where $D^\dagger$ denotes the Moore--Penrose pseudoinverse. [Proof \ref{prf:mpinv}]
\end{lemma}
% \begin{toappendix}
\begin{proof}\label{prf:mpinv}
Since $\mathrm{rank}(X)=k$ and $\mathrm{rank}(D)=k$, there exist factorizations
$X=LR$ with $L\in\mathbb{R}^{n\times k}$, $R\in\mathbb{R}^{k\times p}$ such that
$D=L_2R_2$ where $L_2$ and $R_2$ are the submatrices corresponding to the rows/cols of $D$; both are full column/row rank $k$. Then $D^\dagger = R_2^\dagger L_2^\dagger$ (MacDuffee reverse-order identity under full-rank conditions), so
\[B D^\dagger C = (L_1R_2)(R_2^\dagger L_2^\dagger)(L_2R_1)=L_1R_1=A.\]
\end{proof}
% \end{toappendix}
Lemma~\ref{lem:self-consistency} is the foundational guarantee that makes Verifactor interpretable: under an exact rank--$k$ signal model, the held-out block is \emph{determined} by the held-in blocks via a purely algebraic relation. Thus, prediction error in Verifactor is a direct diagnostic for departures from the rank--$k$ hypothesis (and/or instability due to noise).

\subsubsection{Verifactor as out-of-sample recoverability for factor structure}

Let $\Pi$ denote a random partition of rows and columns into folds. For each fold $(i,j)$, let $(A_{ij},B_{ij},C_{ij},D_{ij})$ be the corresponding block decomposition. Verifactor computes loadings (or low-rank structure) on held-in blocks, constructs a refactored predictor $\widehat A_{ij}$ for the held-out block, and aggregates fit:
\begin{equation}
\mathrm{VF}_m(k;X)
=
\frac{1}{|\Pi|}
\sum_{(i,j)\in\Pi}
m\!\left(A_{ij},\ \widehat A_{ij}^{(k)}\right),
\qquad
\widehat A_{ij}^{(k)} := \mathcal{V}_k(B_{ij},C_{ij},D_{ij}),
\label{eq:app-verifactor-functional}
\end{equation}
where $\mathcal{V}_k$ is a BCV-compatible predictor (e.g., $\widehat A = B(\widehat D^{(k)})^\dagger C$, or a model-specific analog using estimated row/column loadings from $D$). In our analysis, rank is fixed at $k=1$ when the inferential target is unidimensionality, which avoids known monotonicity pathologies of one-way deletion CV for matrix factorization while preserving the interpretability of the hypothesis test.

\begin{theorem}[Unbiasedness of BCV residual energy under random row/column sampling]
\label{thm:bcv-unbiased}
Consider a two-way random effects regime where rows and columns are sampled, and let $X=S+E$ with $\mathbb{E}[E\mid S]=0$ and entrywise variances possibly depending on row/column (heteroscedastic). Let $\widehat A_{ij}^{(k)}$ be constructed from held-in blocks only. Then for squared-error loss $m(A,\widehat A)=\|A-\widehat A\|_F^2$,
\begin{equation}
\mathbb{E}_{\Pi,E}\Big[\|A_{ij}-\widehat A_{ij}^{(k)}\|_F^2\Big]
=
\mathbb{E}_{\Pi}\Big[\|S_{A,ij}-\widehat A_{ij}^{(k)}\|_F^2\Big]
+
\mathbb{E}_{\Pi}\Big[\|E_{A,ij}\|_F^2\Big],
\label{eq:app-bcv-decomp}
\end{equation}
i.e., the expected Verifactor error decomposes into approximation error for the signal plus irreducible held-out noise. [Proof \ref{prf:bcv}]
\end{theorem}
% \begin{toappendix}
\begin{proof}\label{prf:bcv}
Because $\widehat A_{ij}^{(k)}$ is measurable w.r.t.\ the $\sigma$-field generated by held-in blocks and $E$ has mean zero conditional on $S$, the cross-term has zero expectation:
\[
\mathbb{E}\langle E_{A,ij},\, S_{A,ij}-\widehat A_{ij}^{(k)}\rangle
=
\mathbb{E}\Big[\mathbb{E}[\langle E_{A,ij},\, S_{A,ij}-\widehat A_{ij}^{(k)}\rangle\mid S,\Pi,\text{held-in}]\Big]=0.
\]
Expanding $\|A-\widehat A\|_F^2=\|S_A-\widehat A\|_F^2+\|E_A\|_F^2+2\langle E_A,S_A-\widehat A\rangle$ yields \eqref{eq:app-bcv-decomp}.
\end{proof}
% \end{toappendix}
Theorem~\ref{thm:bcv-unbiased} explains why Verifactor is especially well aligned with the random-rows/random-columns inferential target: it evaluates the \emph{generalizable} part of the rank--$k$ structure while cleanly separating irreducible noise. This stands in contrast to in-sample reconstruction, where overfitting and axis-specific dependence can inflate fit. Importantly for unidimensionality, Verifactor analysis can illustrate the extent to which the hypothesized relationship is truly unidimensional. Interestingly, unidimensionality as measured by traditional metrics, based on a covariance or correlation matrix, has no significant empirical relationship with this more honest measure of dimenionality (see Figure \ref{fig:fafit_vs_refact}).

\subsection{Correlational appropriateness as a testable modeling choice}
\label{app:correlational-appropriateness}

A distinctive benefit of the Refactor/Verifactor framing is that it treats the \emph{choice of association matrix} as a scientifically meaningful modeling decision. Let $\mathcal{A}^{(1)},\ldots,\mathcal{A}^{(M)}$ denote candidate association operators (Pearson/$\phi$, Kendall $\tau$, tetrachoric, Yule's $Q$, Cram\'er’s $V$, $H_{ij}$, $\kappa$, \dots). Each operator induces a different notion of ``one-factor'' structure because it defines what it means for two variables to be associated.

\begin{proposition}[Association selection via predictive refactoring]
\label{prop:assoc-selection}
Fix $k=1$ (unidimensionality). For each association operator $\mathcal{A}^{(m)}$, construct $\widehat X^{(m)}$ (Refactor) or $\widehat A^{(m)}_{ij}$ (Verifactor) and compute a common set of metrics $\{m_j\}$. If an operator $\mathcal{A}^{(m^\star)}$ better captures the signal-of-interest in $X$ than alternatives, then it will simultaneously (i) increase out-of-sample Verifactor association metrics (e.g., dCor, Kendall $\tau$, AUC) and/or (ii) reduce out-of-sample loss (e.g., cross-entropy, Frobenius loss), relative to other $\mathcal{A}^{(m)}$. [Proof \ref{prf:assoc}]
\end{proposition}
% \begin{toappendix}
\begin{proof}\label{prf:assoc}
Immediate from definitions: all candidates are evaluated on the same held-out data under the same partitioning scheme; differences arise only through the induced low-rank structure and its predictive consequences. Since Verifactor uses held-out blocks, improvements cannot be attributed to in-sample fit to the image alone.
\end{proof}
% \end{toappendix}
Practically, Proposition~\ref{prop:assoc-selection} yields a principled answer to: \emph{``Is a correlational relationship the right signal/abstraction for this dataset? Which correlational notion or relationship is most appropriate for explaining shared signal?''} These questions are central in applied psychometrics (binary/ordinal items; rater-like features) and in machine learning (implicit feedback, pairwise preferences, weak supervision), yet is not addressed by classical image-only fit statistics because they condition on a single, often default, association choice of standard correlation. As an example, a Pearson correlation with for binary data treats the distance between 0 and 1 as a meaningful distance; we ask whether that is a meaningful assumption when representing the latent relationship.

\section{Evaluation Methods}
\label{app:evaluation_motivation}
\subsection{Measures of Predictive Fit}
The first class of metrics assesses predictive accuracy and rank-order fidelity. These include the Area Under the ROC Curve (AUC) for evaluating rank discrimination, an isotonic $R^2$ to quantify the proportion of variance explained under a monotonic transformation, and Kendall's $\tau$ (\cite{kendall_rank_1962}) to measure the preservation of rank-order agreement between observed and predicted values. We also assess the geometric alignment of the overall response patterns via the matrix cosine similarity, and from an information-theoretic perspective, we use binary cross-entropy to quantify information loss, a measure particularly sensitive to the model's ability to predict rare events.

\subsubsection{Other Measures of Structural and Ordinal Similarity}
These metrics assess the overall correspondence between the continuous-valued reconstruction and the binary ground truth.
\begin{itemize}
    \item \textbf{Matrix Cosine Similarity:} We measure the geometric alignment of the two matrices, treated as vectors in $\mathbb{R}^{N \times J}$. The cosine similarity, $\cos(\mathbf{X}, \mathbf{\widehat{X}}) = \frac{\langle \mathbf{X}, \mathbf{\widehat{X}} \rangle_F}{\|\mathbf{X}\|_F \|\mathbf{\widehat{X}}\|_F}$, is insensitive to global scaling and reflects the preservation of the data's overall geometry.
    \item \textbf{Area Under the ROC Curve (AUC):} By treating the elements of $\mathbf{\widehat{X}}$ as a continuous score to predict the binary elements of $\mathbf{X}$, the AUC quantifies the model's ability to correctly rank the order of 1s and 0s. An AUC of 1.0 signifies a perfect ordinal separation.
    \item \textbf{Kendall's Rank Correlation ($\tau$):} We compute Kendall's $\tau$ between the vectorized elements, $\text{vec}(\mathbf{X})$ and $\text{vec}(\mathbf{\widehat{X}})$, to provide a scale-robust measure of the ordinal agreement between the matrices. % by estimating whether lower reconstruction values also lower in .
\end{itemize}

\subsubsection{Measures of Nonlinear Dependence}
To specifically test for the strength of nonlinear relationships captured by our model, we use bias-corrected distance correlation, which detects any functional dependence between the isotonized percentile rankings \citep{heller_consistent_2013} of the variables. The isotonized rankings that minimize residual error in transformation are found by isotonic regression \citep{busing_monotone_2022} $\widehat{X} \rightarrow \widehat{X}^*\in [0,1]$, a transformation that simultaneously results in values that represent probabilities at each monotonic step.
\begin{itemize}
    \item \textbf{Mean Bias-corrected Squared Distance Correlation ($\operatorname{dCor}^2_n$):} We calculate the bias-corrected squared distance correlation \citep{szekely_energy_2013,szekely_energy_2017} between the matrices, $\operatorname{dCor}^2_n(
    \mathbf{X},\mathbf{\widehat{X}^*})$. This metric is zero if and only if the matrices are independent, providing a powerful test of association that is not limited to linear relationships. The bias-corrected measure is orientation-dependent, so we report the mean of this value and of the transposed estimate, $\operatorname{dCor}^2_n(
    \mathbf{X}^\top,{\mathbf{\widehat{X}^{*\top}}})$
    \item \textbf{Mean Bias-corrected (Squared) Partial Distance Correlation ($\operatorname{pdCor}$):} To determine if the reconstruction provides information beyond simple marginal effects, we compute the bias-corrected partial distance correlation \citep{szekely_partial_2014}. We control for a baseline rank-1 matrix $\mathbf{E}$ constructed from the outer product of the row and column marginal probabilities. We refer to this as the ``independence structure'' The resulting metric, $\operatorname{pdCor}(\mathbf{X},\mathbf{\widehat{X}^*};\mathbf{E})$, isolates the explanatory power of the dependence structure captured by the factor loadings.
\end{itemize}

\subsubsection{Probabilistic Goodness-of-Fit}
Finally, we evaluate the reconstruction from a probabilistic standpoint.
\begin{itemize}
    \item \textbf{Likelihood:} The entries of $\mathbf{\widehat{X}}$ are not probabilities.  Using the same isotonic regression used to find the non-decreasing function that best maps the values of $\mathbf{\widehat{X}}$ to a matrix of probabilities $\mathbf{\widehat{X}^*} \in [0, 1]^{N \times J}$. We then calculate the log-likelihood of the observed data $\mathbf{X}$ given these probabilities. To facilitate comparison across datasets of different sizes, we report the geometric mean of the likelihoods, which is equivalent to the exponentiated average per-cell log-likelihood and serves as a principled measure of model fit. This metric can help identify whether a rank-1 representation has a greater than chance probability to result in the original response matrix.
\end{itemize}

\subsection{Recoverability of the response matrix under a rank--1 associative hypothesis}
\label{app:evaluation_target}

Let $X\in\mathbb{R}^{n\times p}$ be observed responses (possibly binary/ordinal). For a fixed association operator $\mathcal{A}$ and a fixed rank $k=1$, Refactor/Verifactor define an estimand that is more primitive than ``fit to correlations'': the extent to which a \emph{rank--1 associative signal} supports \emph{matrix prediction} of $X$.

Formally, an association operator $\mathcal{A}$ induces a hypothesis class of reconstructions
\[
\mathcal{H}_{\mathcal{A}}
=
\Bigl\{
\widehat X = \mathcal{R}_1(X;B_1,F_1):
(B_1,F_1) = \textsf{Z}(\mathcal{A}_r(X),\mathcal{A}_c(X))
\Bigr\},
\]
and Verifactor evaluates the \emph{out-of-sample risk}
\begin{equation}
\mathcal{R}_{\mathcal{A}}(m)
\;:=\;
\mathbb{E}_{\Pi}\,\mathbb{E}\Bigl[
m\bigl(A_{\Pi},\widehat A_{\Pi}^{(\mathcal{A})}\bigr)
\Bigr],
\label{eq:app-assoc-risk}
\end{equation}
where $A_{\Pi}$ is the held-out block under a random row/column partition $\Pi$ and
$\widehat A_{\Pi}^{(\mathcal{A})}$ is predicted from held-in blocks using the rank--1 structure implied by $\mathcal{A}$. This creates a common evaluation currency across (i) association choices, (ii) estimators $\textsf{Z}$, and (iii) reconstruction metrics $m$.

\paragraph{Interpretation.}
Equation~\eqref{eq:app-assoc-risk} turns ``correlational appropriateness'' into a falsifiable claim: the correct associative abstraction is the one that supports stable prediction of held-out entries when both rows and columns are treated as random. This target is distinct from, and not implied by, image fit indices because those indices condition on the association image as sufficient statistics.

\begin{proposition}[Non-equivalence of image fit and recoverability]
\label{prop:non-equivalence}
There exist sequences of data matrices $X^{(t)}$ and rank--1 estimators $\widehat S^{(t)}$ such that image residuals vanish,
\[
\| \mathcal{A}(X^{(t)}) - \mathcal{A}(\widehat S^{(t)})\|_F \to 0,
\]
while out-of-sample recoverability does not improve, e.g.,
\[
\mathbb{E}\big[\|A_{ij}^{(t)}-\widehat A_{ij}^{(t)}\|_F^2\big]\not\to \mathbb{E}\big[\|E_{A,ij}^{(t)}\|_F^2\big],
\]
for BCV-held-out blocks $A_{ij}^{(t)}$ (definitions as in \S\ref{app:verifactor-proof}).
\end{proposition}

% \begin{toappendix}
\begin{proof}[Proof sketch]
Take $X^{(t)}$ generated by a non-factorable but pairwise-matched process: e.g., a mixture of two latent classes with opposite signed loadings, or a thresholded nonlinear model whose induced pairwise association matrix matches that of a one-factor model but whose higher-order dependencies do not factorize. Image-based estimators can reproduce $\mathcal{A}(X^{(t)})$ arbitrarily well, while prediction of held-out blocks fails because the conditional structure needed for completion is absent. BCV isolates this failure by construction because held-out entries are not used to estimate the completion operator.
\end{proof}
% \end{toappendix}

\subsubsection{Matrix prediction as the right evaluation primitive in two-way random-effects regimes}
\label{app:why_matrix_prediction}

In random observations $\times$ random variables regimes \citep{de_boeck_random_2008}, the scientific question concerns generalization to new rows and new columns. This makes the \emph{unit of resampling} (and thus the unit of evaluation) crucial. Entrywise resampling or prediction is generally misaligned because each entry shares its row and column with observed entries, enabling optimistic leakage.

\subsubsection{Axis-respecting generalization and leakage control}
\label{app:axis_respecting}

Let $\Pi=(\mathcal{I},\mathcal{J})$ denote a fold, with held-out rows $\mathcal{I}\subset[n]$ and held-out columns $\mathcal{J}\subset[p]$. The held-out block is
$A=X_{\mathcal{I},\mathcal{J}}$. Verifactor constructs $\widehat A$ from
$D=X_{\mathcal{I}^c,\mathcal{J}^c}$ (and, depending on the completion rule, also $B=X_{\mathcal{I},\mathcal{J}^c}$ and $C=X_{\mathcal{I}^c,\mathcal{J}}$) but \emph{never} from rows in $\mathcal{I}$ together with columns in $\mathcal{J}$.

\begin{proposition}[Two-way holdout targets new-row/new-column prediction]
\label{prop:two_way_targets}
Suppose $X$ is generated under a crossed design with random rows and random columns.
Then the Verifactor risk \eqref{eq:app-assoc-risk} estimates prediction performance for simultaneously new persons and new items, whereas entrywise CV targets prediction for new entries conditional on previously seen persons and items.
\end{proposition}

% \begin{toappendix}
\begin{proof}[Proof sketch]
Under crossed random sampling, generalization requires conditioning on neither the held-out row latent nor the held-out column latent. Entrywise CV conditions on both (since the row and column are observed elsewhere), while two-way holdout removes both simultaneously by excluding entire row/column subsets. Therefore, the induced conditioning sets differ, leading to different estimands.
\end{proof}
% \end{toappendix}

This is the primary sense in which Verifactor is an \emph{improvement} over standard reconstruction assessments: it aligns evaluation with the intended generalization regime, not merely with computational convenience.

\subsubsection{Self-consistency explains why completion is a direct test of rank--1 structure}
\label{app:self_consistency_discussion}

BCV completion is the algebraic identity that must hold for exact low-rank matrices (Lemma~\ref{lem:self-consistency}). Consequently, failure to predict held-out blocks indicates that the rank--1 associative hypothesis is not stable across the matrix.

\begin{corollary}[Verifactor error is a direct diagnostic for departures from rank--1 structure]
\label{cor:direct_diagnostic}
If $S$ is rank--1 and the held-in block has the same rank, then $A=B D^\dagger C$ exactly in the noiseless case. Therefore, persistent Verifactor error beyond irreducible noise indicates at least one of: (i) $S$ is not approximately rank--1, (ii) the chosen association $\mathcal{A}$ does not linearize the signal, (iii) the signal is not shared across folds (instability / nonstationarity across rows or columns).
\end{corollary}

This corollary clarifies the interpretability advantage over many global fit statistics: Verifactor failure is not merely ``misfit'' but misfit in a way that prevents the model from performing the implied generative task of reconstructing unseen portions of the response matrix.

\subsection{Improving matrix reconstruction}
\label{app:why_metric_suite-main}

Matrix reconstruction introduces a second challenge: \emph{how to measure similarity between $X$ and $\widehat X$ in a way that is robust across datasets, link functions, and marginal imbalances.} There is no single universally appropriate scalar because different failure modes of a rank--1 associative model manifest in different ways.

\subsubsection{Pitfall 1: Frobenius error rewards marginal matching and can mask independence}
\label{app:frob_pitfall}

For binary $X$, a predictor can achieve deceptively good squared error by learning row and column marginals alone. Let
% \[
$E := \mathbb{E}(X)\ \approx\ \pi_r \pi_c^\top$,
% \]
where $(\pi_r)_i = \frac{1}{p}\sum_{j}X_{ij}$ and $(\pi_c)_j=\frac{1}{n}\sum_i X_{ij}$.
Then $\widehat X=E$ may have acceptable $\|X-\widehat X\|_F^2$ while capturing no interaction structure. If $X_{ij}\sim\mathrm{Bern}((\pi_r)_i(\pi_c)_j)$ independently conditional on $(\pi_r,\pi_c)$, then $\widehat X=E$ is Bayes optimal under squared error, even though $X$ contains no latent-factor dependence beyond marginals.

Thus, we include dependence-sensitive metrics (bias-corrected squared distance correlation and partial distance correlation), explicitly designed to be $0$ under independence and to remain small when only marginals are matched.

% \subsection{Avoiding common pitfalls in matrix reconstruction}
% \label{app:why_metric_suite}

Matrix reconstruction introduces a second challenge: \emph{how to measure similarity between $X$ and $\widehat X$ in a way that is robust across datasets, link functions, and marginal imbalances.} There is no single universally appropriate scalar because different failure modes of a rank--1 associative model manifest in different geometric and probabilistic ways.

% \subsubsection{Pitfall 1: Frobenius error rewards marginal matching and can mask independence}
% \label{app:frob_pitfall}

For binary $X$, a predictor can achieve deceptively good squared error by learning row and column marginals alone. Let
\[
E := \mathbb{E}(X)\ \approx\ \pi_r \pi_c^\top,
\]
where $(\pi_r)_i = \frac{1}{p}\sum_{j}X_{ij}$ and $(\pi_c)_j=\frac{1}{n}\sum_i X_{ij}$.
Then $\widehat X=E$ may have acceptable $\|X-\widehat X\|_F^2$ while capturing no interaction structure.

\begin{proposition}[Marginal-only predictors can score well under MSE]
\label{prop:marginal_mse}
If $X_{ij}\sim\mathrm{Bern}((\pi_r)_i(\pi_c)_j)$ independently conditional on $(\pi_r,\pi_c)$, then $\widehat X=E$ is Bayes optimal under squared error, even though $X$ contains no latent-factor dependence beyond marginals.
\end{proposition}

% \begin{toappendix}
\begin{proof}
For Bernoulli observations, the conditional mean minimizes squared loss entrywise. Under the stated model, $\mathbb{E}[X_{ij}\mid \pi_r,\pi_c]=(\pi_r)_i(\pi_c)_j$, giving $\widehat X=E$.
\end{proof}
% \end{toappendix}

Hence we include dependence-sensitive metrics (bias-corrected dCor and partial dCor), explicitly designed to be $0$ under independence and to remain small when only marginals are matched.

\subsubsection{Pitfall 2: elementwise correlations conflate calibration with discrimination}
\label{app:corr_pitfall}

Entrywise Pearson correlation between $\mathrm{vec}(X)$ and $\mathrm{vec}(\widehat X)$ can increase even when the predictor is poorly calibrated (e.g., overly confident probabilities) and can be sensitive to monotone nonlinearities. For binary data, two rankers with identical AUC may have very different likelihood and vice versa. Consequently, we separate:

\begin{itemize}\itemsep2pt
\item \textbf{Ordinal/discriminative fit} (AUC, Kendall $\tau_b$): evaluates whether the reconstruction respects rank structure.
\item \textbf{Calibration/probabilistic fit} (geometric mean likelihood / cross-entropy after isotonic calibration): evaluates whether predicted scores can be made probabilistically meaningful.
\end{itemize}

\subsubsection{Pitfall 3: orientation and sign/scale indeterminacy}
\label{app:orientation_pitfall}

Low-rank representations are invariant to sign flips and rotations (for $k>1$). Even for $k=1$, sign is arbitrary. Metrics sensitive to linear scale (e.g., MSE) can therefore reflect arbitrary conventions. We address this by (i) using scale-invariant metrics (cosine similarity, rank metrics), and (ii) applying isotonic regression $\widehat X\mapsto\widetilde X\in[0,1]$ to evaluate probabilistic fit and nonlinear dependence on a common monotone scale.

\begin{proposition}[Isotonic calibration yields a canonical monotone link for evaluation]
\label{prop:isotonic}
Let $\widehat X$ be any real-valued reconstruction and $X\in\{0,1\}^{n\times p}$.
Define $\widetilde X = g^\star(\widehat X)$ where $g^\star$ is the isotonic regression fit minimizing $\|X-g(\widehat X)\|_F^2$ over nondecreasing $g$ applied entrywise.
Then $\widetilde X$ is invariant to strictly monotone reparameterizations of $\widehat X$ and can be interpreted as a calibrated score in $[0,1]$.
\end{proposition}
% \begin{toappendix}
\begin{proof}
Isotonic regression returns the projection onto the cone of monotone functions. If $\varphi$ is strictly monotone, the order of entries of $\widehat X$ and $\varphi(\widehat X)$ is identical, yielding the same feasible set of monotone step functions after pooling ties; thus the projected values coincide up to tie-handling conventions.
\end{proof}
% \end{toappendix}

This justifies using $\widetilde X$ for likelihood and (partial) distance correlation: it isolates whether there exists \emph{any} monotone link under which a rank--1 reconstruction becomes probabilistically and structurally meaningful.

\subsubsection{Pitfall 4: redundancy of classical fit indices and the need for orthogonal diagnostics}
\label{app:redundancy_classical}

Many classical unidimensionality and FA fit indices are strongly correlated in practice because they are functions of residuals on the same association image. This creates the appearance of convergent evidence even when the underlying question (recoverability of $X$) is unanswered. In contrast, the reconstruction metric suite is intentionally \emph{non-collinear}: each metric targets a different failure mode.

\begin{center}
\begin{tabular}{p{0.30\linewidth}p{0.62\linewidth}}
\toprule
Metric family & Detects failure mode \\
\midrule
$\cos_F$ & wrong global geometry / orientation, even if ranking is partly correct \\
AUC, $\tau_b$ & wrong ordinal structure (cannot separate 1s from 0s) \\
$\operatorname{dCor}_n^2$ & dependence mismatch beyond linear association \\
$R^*_{X,\widehat X;\mathbb{E}(X)}$ & ``marginal-only'' success without interaction signal \\
Likelihood / cross-entropy & poor calibration; overconfident wrong predictions; rare-event failure \\
\bottomrule
\end{tabular}
\end{center}

This is why comparing and contrasting metrics is informative rather than redundant: disagreement between metrics is itself a diagnosis of \emph{what kind} of structure the rank--1 associative model is (or is not) capturing.

\subsection{Why these evaluations are especially informative when the data-generating model is unknown}
\label{app:unknown_dgm}

In many applications (psychological instruments across populations; AI benchmark responses across tasks/models), the data-generating mechanism is not specified, and the analyst is deciding whether a rank--1 correlational abstraction is a reasonable simplification. In that setting, Refactor/Verifactor provide two key advantages:

\paragraph{(i) Model criticism without committing to a parametric likelihood}
Image-fit indices often depend on distributional assumptions (e.g., Gaussian latent variables for tetrachorics, asymptotic $\chi^2$ calibrations for fit). Refactor/Verifactor instead rely on prediction, which remains well-defined under heteroscedasticity, non-Gaussianity, and monotone distortions. When paired with metrics like dCor and isotonic likelihood, the evaluation becomes robust to unknown link functions.

\paragraph{(ii) Comparability across association choices}
When the signal-of-interest is ambiguous, multiple associations correspond to multiple scientific hypotheses (e.g., ``binary distances are meaningful'' for $\phi$; ``latent bivariate normal thresholds'' for tetrachoric; ``odds-ratio monotonicity'' for Yule's $Q$; ``monotone scalability'' for Loevinger $H$). Verifactor allows these hypotheses to be compared on equal footing via the same held-out prediction task.

\subsection{Re-evaluating correlational assumptions for dichotomous and sparse matrices}
\label{app:binary_sparse}

Binary matrices are ubiquitous (item responses; benchmark pass/fail; preference clicks). They also create two ubiquitous confounds: \emph{prevalence} (marginal imbalance) and \emph{sparsity} (few ones). Correlation-based images can be dominated by these confounds, and different association measures respond differently.

\begin{proposition}[Association choice changes the implied geometry of rank--1 structure]
\label{prop:assoc_geometry}
Let $X\in\{0,1\}^{n\times p}$. Different association operators $\mathcal{A}$ correspond to different embeddings of rows/columns into Euclidean (or pseudo-Euclidean) geometries. Consequently, the ``best'' rank--1 approximation depends on $\mathcal{A}$, and agreement on image fit does not imply agreement on recoverability of $X$.
\end{proposition}

% \begin{toappendix}
\begin{proof}[Proof sketch]
Each $\mathcal{A}$ defines a bilinear form (or kernel-like matrix) whose leading eigenvector defines the rank--1 direction. Changing $\mathcal{A}$ changes that bilinear form and therefore changes the induced projection of $X$ onto a one-dimensional subspace. Since $\widehat X$ depends on both row and column projections, recoverability changes accordingly.
\end{proof}
% \end{toappendix}

Empirically (Table of examples), this manifests as meaningful reversals: a measure that is strong on $\cos_F$ or likelihood may be weaker on $R^*$ or $\operatorname{dCor}_n^2$, indicating that it either (a) captures global geometry but not residual dependence, or (b) captures dependence beyond marginals but is poorly calibrated for prediction of rare events. Such differences are exactly the nuanced insights that image-based indices cannot provide.

\subsection{Irregular regimes (e.g., $n \ll p$) and modern benchmarks}
\label{app:nllp}

In $n\ll p$ regimes (common when evaluating a small number of models on many tasks, or small sample sizes with many items), correlation matrices can be noisy, high-dimensional, and unstable; fit indices can become artifacts of regularization, smoothing, or the chosen correlation estimator. Refactor/Verifactor mitigate this in two ways:

\begin{enumerate}\itemsep2pt
\item They evaluate predictions on $X$, not only on $p\times p$ images, so the criterion remains meaningful even when the image is ill-conditioned.
\item BCV reduces axis-specific overfitting: a model that fits spurious correlations among a particular set of items will not generalize to held-out items.
\end{enumerate}

\begin{corollary}[BCV regularizes factor selection by enforcing two-way stability]
\label{cor:bcv_regularizes}
When $n\ll p$ or $p\ll n$, in-sample image fit can be inflated by sampling noise in the association image. Verifactor penalizes such inflation because successful prediction requires that the inferred rank--1 structure persists across held-out rows and columns.
\end{corollary}

\subsection{Reconstruction Evaluation Summary}
\label{app:practical_insights_eval}

The proposed evaluation yields a compact set of interpretable outcomes:

\begin{enumerate}\itemsep2pt
\item \textbf{Unidimensionality as predictability:} high Verifactor performance across dependence-sensitive and discriminative metrics supports the rank--1 hypothesis in the sense relevant to generalization.
\item \textbf{Correlation as a testable abstraction:} if all $\mathcal{A}^{(m)}$ perform poorly out-of-sample (especially on dCor and $R^*$), the evidence suggests that the dominant signal is not well represented by any rank--1 correlational relationship (e.g., mixture, conditional dependence, nonlinear interactions).
\item \textbf{Which relationship best captures shared signal:} if some $\mathcal{A}^{(m)}$ dominate, the analyst gains a principled choice of association aligned with stable recovery of $X$, not merely with convenience or tradition.
\item \textbf{Diagnosing the kind of failure:} discrepancies among $\cos_F$, AUC/$\tau_b$, dCor/$R^*$, and likelihood identify whether failure arises from geometry, ranking, dependence structure beyond marginals, or calibration.
\end{enumerate}

These insights directly address the shortcomings motivating this work: traditional fit indices are valuable but often redundant and image-bound; Refactor/Verifactor and the accompanying metric suite provide a predictive, two-way-stable, association-comparable evaluation layer that is especially appropriate in modern random-by-random data settings.

% \section{Evaluation Methods}

\subsection{Refactoring Analyses Evaluate Signal Recovery}
This paper introduces a paradigm shift in model evaluation, termed \textit{Refactoring Analyses}, which refocuses the assessment from the model's abstract image to its direct explanatory power. The central premise, outlined in Figure 1, is to leverage the full structure of the low-rank model to generate a prediction of the original data. This is achieved through a dual analysis of the data matrix, deriving not only the column-space loadings ($\boldsymbol{A}$) from the covariance matrix ($\boldsymbol{X}^T\boldsymbol{X}$) but also the row-space loadings ($\boldsymbol{B}$) from the Gram matrix ($\boldsymbol{X}\boldsymbol{X}^T$). These two matrices, representing the principal axes of the data from both orientations, are then used to "decompress" or refactor a complete, model-based prediction of the original data matrix, $\boldsymbol{\hat{X}}$. This refactoring step transforms the dimensionality reduction model into a generative, predictive one.

The crucial final step of Refactoring Analyses is to quantify the correspondence between the observed data $\boldsymbol{X}$ and its reconstruction $\boldsymbol{\hat{X}}$. This comparison provides a direct and interpretable measure of how well the low-rank structure captures the information in the original data. Because this task requires methods for comparing matrices that may be unfamiliar to many researchers, we introduce a suite of powerful metrics. Moving beyond simple element-wise correlations, these include global measures of matrix association, such as distance correlation and the RV coefficient, as well as more granular diagnostics that assess structural fidelity, such as the preservation of vector orientation as measured by cosine similarity. Together, these tools provide a richer, more diagnostically informative framework for evaluating model adequacy than is possible by inspecting the factor loadings alone.

\paragraph{Verifactoring Analyses evaluates predictive power}
Refactoring allows for even more robust tests of fit. Inspired by bi-cross-validated block prediction \citep{owen_bi-cross-validation_2009, owen_bi-cross-validation_2015}, we also use Verifying Refactor Analyses (Verifactoring) to measure the strength of the underlying relationships captured, by predicting performance on held-out test data. We hold out two “off-diagonal” blocks of a two-way array, learn structure on the complementary diagonal blocks, and evaluate or predict the held-out diagonal. This two-way sample splitting is a form of block cross-validation for matrices where we fit theory based low-rank structures on one diagonal (e.g., A and D)to predict the other (i.e., B and C). By treating rows and columns as stochastic units for each dataset iteration, we get four opportunities to cross-validate structure, not just entries.

\begin{table*}
\centering
\caption{Refactor vs Verifactor Outcome on Human Assessments vs. AI Benchmarks }
\resizebox{\ifdim\width>\linewidth\linewidth\else\width\fi}{!}{
\begin{tabular}{lllllll}
\toprule
Dependent Var.: & $\tau$& AUC& $\cos$& $\mathcal{L}(X\mid\theta)$ & $\operatorname{dCor}_n^2$ & $R_{X,\hat{X};\operatorname{E}(X)}^*$ \\\midrule
Verifactor & -0.023*** (0.001) & -0.016*** (0.0009) & -0.008*** (0.0005) & -0.010*** (0.0006) & -0.051*** (0.001) & -0.021* (0.008) \\
Human $\times$ Verifact & -0.138*** (0.005) & -0.066*** (0.002) & -0.018*** (0.0009) & -0.033*** (0.001) & -0.192*** (0.003) & -0.092*** (0.006) \\
\midrule
idx-table-boot-submatrix & Yes & Yes & Yes & Yes & Yes & Yes \\

% idx-table & Yes & Yes & Yes & Yes & Yes & Yes \\
% table-boot-target & Yes & Yes & Yes & Yes & Yes & Yes \\
n\_items (table) & Yes & Yes & Yes & Yes & Yes & Yes \\
n\_participants (table) & Yes & Yes & Yes & Yes & Yes & Yes \\
% S.E.: Clustered & by: idx & by: idx & by: idx & by: idx & by: idx & by: idx \\
Observations & 92,895 & 92,895 & 92,895 & 92,895 & 91,154 & 89,058 \\
RMSE & 0.09216 & 0.04406 & 0.02079 & 0.02156 & 0.12660 & 0.16705 \\
% R2 & 0.72677 & 0.87259 & 0.97778 & 0.97253 & 0.75780 & 0.51775 \\
% Pseudo R2 & -1.4854 & 23.374 & -2.0494 & -1.5307 & -3.4538 & -2.8950 & 11.588 & -62.167 & -2.5346\\
% Within R2 & 0.36558 & 0.39679 & 0.23592 & 0.42472 & 0.40960 & 0.07719 \\
Adj. R2 & 0.71229 & 0.86584 & 0.97660 & 0.97107 & 0.74473 & 0.49154 \\
Within Adj. R2 & 0.36557 & 0.39677 & 0.23590 & 0.42470 & 0.40959 & 0.07717 \\
% Within Adj. Pseudo R2 & -0.14018 & -0.03985 & -0.30838 & -0.17423 & -0.05799 & -0.12906 & -0.68549 & -0.12152 & -0.30960\\
\bottomrule
\end{tabular}}\label{tab:verifeols}
\caption*{\footnotesize Controlling for each relationship, dataset, bootstrap composition, submatrix composition, item and subject counts with fixed effects, we measure how large the differences are, on average, between Verifactor and Refactor scores on the same matrices and the extent to which humans assessments change the difference in Verifactor outcomes. We see that humans are significantly harder to predict out-of-sample than LLMs on AI Benchmarks and items. Part of this may be due to the nature of historical leaderboard data, where models are representative of capacity at a certain point in technological development, and part may be due a greater diversity of human tasks represented within the datasets.}
\end{table*}

\begin{figure*}[h!]
    \centering
    \includegraphics[width=1\linewidth]{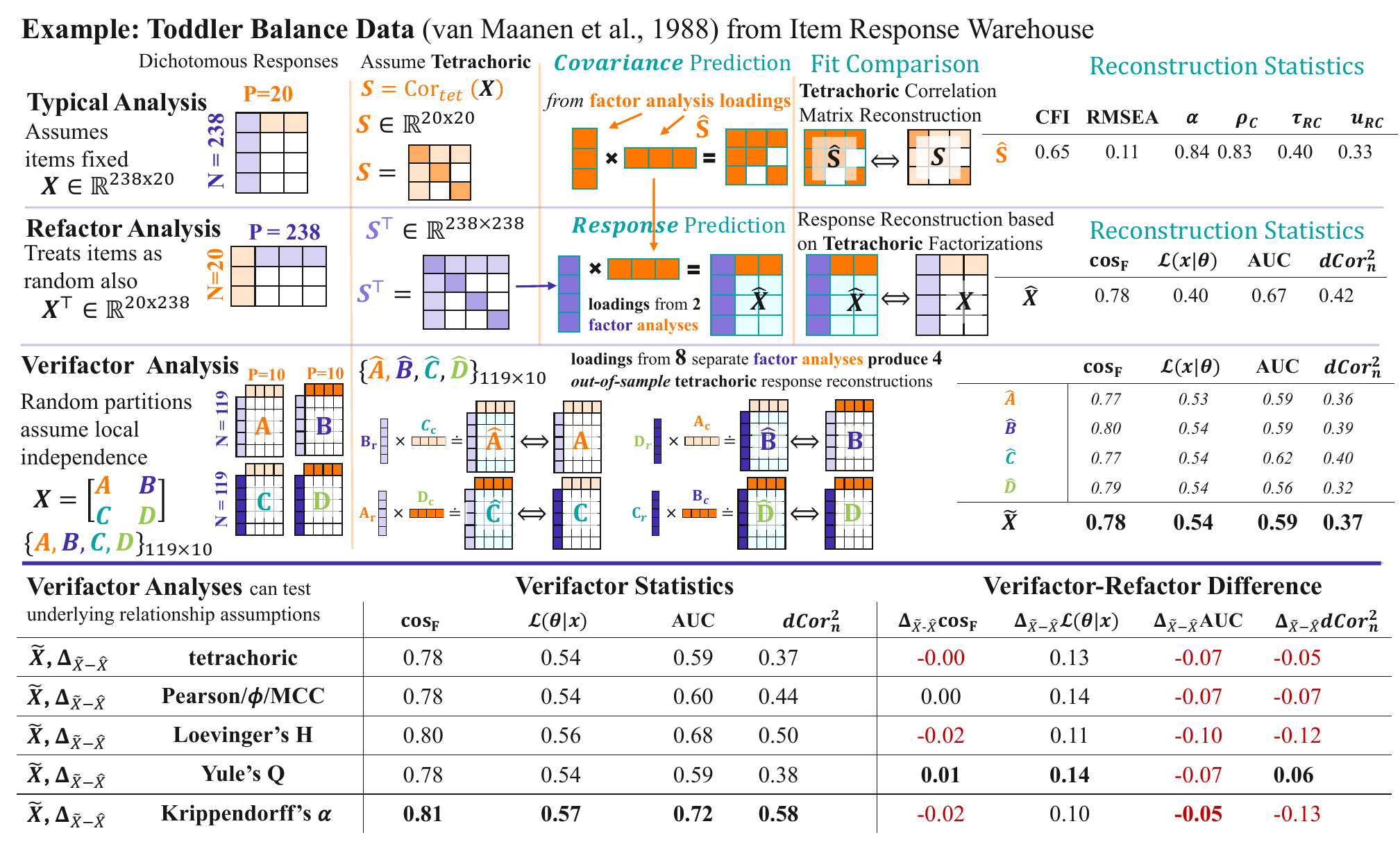}
    \caption{Example Refactor and Verifactor Analyses to test assumptions of underlying relationships as illustrated in Figure \ref{fig:refactor}.}
    \label{fig:ex_verifactor}
\end{figure*}

\subsection{Refactor Metrics}\label{app:eval_metrics}
A core contribution of this work is the intentional selection of a diagnostic suite of metrics for evaluating the reconstruction. A single summary statistic is insufficient to assess the correspondence between the observed data matrix $\boldsymbol{X}$ and its model-based reconstruction $\boldsymbol{\hat{X}}$. Therefore, we propose a multifaceted evaluation that provides a holistic profile of the model's performance. The first class of metrics assesses predictive accuracy and rank-order fidelity. These include the Area Under the ROC Curve (AUC) for evaluating rank discrimination, an isotonic $R^2$ to quantify the proportion of variance explained under a monotonic transformation, and Kendall's $\tau$ (\cite{kendall_rank_1962}) to measure the preservation of rank-order agreement between observed and predicted values. We also assess the geometric alignment of the overall response patterns via the matrix cosine similarity, and from an information-theoretic perspective, we use binary cross-entropy to quantify information loss, a measure particularly sensitive to the model's ability to predict rare events.

% Beyond these foundational measures, we employ advanced metrics to probe for deeper, nonlinear structural correspondence, which is often missed by conventional analyses. A central tool in this endeavor is the bias-corrected distance correlation ($\text{dCor}$), a powerful measure that quantifies both linear and nonlinear dependence between matrices (\cite{szekely_energy_2017}). This is supplemented by partial distance correlation (\cite{szekely_partial_2014}) to specifically isolate dependence structures that are not captured by simple marginal effects. To assess the degree of shared structure and subspace overlap, we utilize Yanai's Generalized Coefficient of Determination (\cite{yanai_proposition_1980}) (GCD) and Ramsay's matrix correlation coefficients ($r_1$ and $r_2$) (\cite{ramsay_matrix_1984}), which quantify the alignment between the principal components and spectral properties of the observed and reconstructed data. For metrics that are dependent on matrix orientation, we ensure a comprehensive and symmetric evaluation by computing each metric on both the reconstruction and its transpose and reporting the average, thereby providing a robust assessment of structural fidelity and providing recommendations and heuristics for baselines. Table \ref{tab:metrics} with brief descriptions of these and of standard factor analyses and software used in estimation.

\subsection{Multi-faceted Reconstruction Evaluation}
To comprehensively assess the quality of the rank-1 reconstruction $\mathbf{\widehat{X}}$, we employ a suite of metrics that probe different aspects of its relationship with the original binary matrix $\mathbf{X}$. This multi-faceted approach ensures a robust and holistic evaluation and avoids pitfalls associated with matrix relationships \citep{josse_measuring_2016,josse_measures_2014,ramsay_matrix_1984}. 

\subsubsection{Measures of Structural and Ordinal Similarity}
These metrics assess the overall correspondence between the continuous-valued reconstruction and the binary ground truth.
\begin{itemize}
    \item \textbf{Matrix Cosine Similarity:} We measure the geometric alignment of the two matrices, treated as vectors in $\mathbb{R}^{N \times J}$. The cosine similarity, $\cos(\mathbf{X}, \mathbf{\widehat{X}}) = \frac{\langle \mathbf{X}, \mathbf{\widehat{X}} \rangle_F}{\|\mathbf{X}\|_F \|\mathbf{\widehat{X}}\|_F}$, is insensitive to global scaling and reflects the preservation of the data's overall geometry.
    \item \textbf{Area Under the ROC Curve (AUC):} By treating the elements of $\mathbf{\widehat{X}}$ as a continuous score to predict the binary elements of $\mathbf{X}$, the AUC quantifies the model's ability to correctly rank the order of 1s and 0s. An AUC of 1.0 signifies a perfect ordinal separation.
    \item \textbf{Kendall's Rank Correlation ($\tau$):} We compute Kendall's $\tau$ between the vectorized elements, $\text{vec}(\mathbf{X})$ and $\text{vec}(\mathbf{\widehat{X}})$, to provide a scale-robust measure of the ordinal agreement between the matrices. % by estimating whether lower reconstruction values also lower in .
\end{itemize}

\subsubsection{Measures of Nonlinear Dependence}
To specifically test for the strength of nonlinear relationships captured by our model, we use bias-corrected distance correlation, which detects any functional dependence between the isotonized percentile rankings \citep{heller_consistent_2013} of the variables. The isotonized rankings that minimize residual error in transformation are found by isotonic regression \citep{busing_monotone_2022} $\widehat{X} \rightarrow \widetilde{X}\in [0,1]$, a transformation that simultaneously results in values that represent probabilities at each monotonic step.
\begin{itemize}
    \item \textbf{Mean Bias-corrected Squared Distance Correlation ($\operatorname{dCor}^2_n$):} We calculate the bias-corrected squared distance correlation \citep{szekely_energy_2013,szekely_energy_2017} between the matrices, $\operatorname{dCor}^2_n(
    \mathbf{X},\mathbf{\widetilde{X}})$. This metric is zero if and only if the matrices are independent, providing a powerful test of association that is not limited to linear relationships. The bias-corrected measure is orientation-dependent, so we report the mean of this value and of the transposed estimate, $\operatorname{dCor}^2_n(
    \mathbf{X}^\top,\mathbf{\widetilde{X}}^\top)$
    \item \textbf{Mean Bias-corrected (Squared) Partial Distance Correlation ($\operatorname{pdCor}$):} To determine if the reconstruction provides information beyond simple marginal effects, we compute the bias-corrected partial distance correlation \citep{szekely_partial_2014}. We control for a baseline rank-1 matrix $\mathbf{E}$ constructed from the outer product of the row and column marginal probabilities. We refer to this as the ``independence structure'' The resulting metric, $\operatorname{pdCor}(\mathbf{X},\mathbf{\widetilde{X}};\mathbf{E})$, isolates the explanatory power of the dependence structure captured by the factor loadings.
\end{itemize}

\subsubsection{Probabilistic Goodness-of-Fit}
Finally, we evaluate the reconstruction from a probabilistic standpoint.
\begin{itemize}
    \item \textbf{Likelihood:} The entries of $\mathbf{\widehat{X}}$ are not probabilities.  Using the same isotonic regression used to find the non-decreasing function that best maps the values of $\mathbf{\widehat{X}}$ to a matrix of probabilities $\mathbf{\widetilde{X}} \in [0, 1]^{N \times J}$. We then calculate the log-likelihood of the observed data $\mathbf{X}$ given these probabilities. To facilitate comparison across datasets of different sizes, we report the geometric mean of the likelihoods, which is equivalent to the exponentiated average per-cell log-likelihood and serves as a principled measure of model fit. This metric can help identify whether a rank-1 representation has a greater than chance probability to result in the original response matrix.
\end{itemize}

\subsubsection{Reinterpreting classical fit indices under a data-reconstruction lens}
\label{app:reinterpret-fit}

Most factor-analytic fit statistics (e.g., CFI, TLI, RMSEA) and many unidimensionality indices (e.g., $\alpha$, average interitem $r$, $\rho_c$, $u$) are functions of \emph{image residuals}---differences between an observed association matrix and its model-implied counterpart. Refactor/Verifactor show that these indices should be interpreted as measuring:

\begin{quote}
\emph{how well a rank--$k$ correlational model reproduces a particular second-order summary of the data},
\end{quote}

not necessarily as measuring:

\begin{quote}
\emph{whether a rank--$k$ correlational model captures the signal that is actually predictive of the response matrix}.
\end{quote}

This clarifies why classical indices can agree strongly with each other yet be weakly related to recoverability: they share the same sufficient statistics (the association image) and therefore cannot detect failures that are invisible at the level of pairwise summaries (mixtures, context effects, conditional dependence, nonlinearities). In our results (Figures~\ref{fig:fafit_vs_refact}--\ref{fig:fafit_vs_refact}), the near-absence of association between classical indices and Verifactor scores indicates that, across diverse public datasets, passing image-fit thresholds often does not entail out-of-sample recoverability of $X$.

\subsubsection{Correlational appropriateness for dichotomous data: correlation is a modeling choice}
\label{app:dichotomous-correlation}

For binary $X$, ``correlation'' is not uniquely defined. Pearson's $\phi$, tetrachoric correlation, Yule's $Q$, quadrant correlation, and information-theoretic measures each encode different invariances and different implied latent-variable stories. Image-based pipelines typically commit to one association matrix (often tetrachoric or $\phi$) and then evaluate fit within that choice. Refactor/Verifactor instead make the choice \emph{empirically comparable} by evaluating each association through the same out-of-sample completion task.

The consequence is practical: Table~(examples) shows that the \emph{best} association depends on both the dataset and the downstream metric. For some datasets, Loevinger's $H$ or Krippendorff's $\alpha$ yields markedly better Verifactor dependence and rank-order recovery than tetrachoric; for others, $\phi$ or quadrant-based measures dominate on AUC or likelihood. This is not a nuisance but an insight: it indicates which associational abstraction best captures the dataset’s stable, generalizable structure.

\begin{proposition}[Association matrices as competing hypotheses]
\label{prop:assoc-hypotheses}
Fix $k=1$. Each association operator $\mathcal{A}^{(m)}$ defines a hypothesis class
$\mathcal{H}^{(m)}:=\{\widehat X=\mathcal{R}_1(X;B_1^{(m)},F_1^{(m)})\}$.
Under BCV, comparing $\mathrm{VF}_m(1;X)$ across $m$ is a valid model comparison for these hypothesis classes with respect to the chosen loss/metric (Theorem~\ref{thm:bcv-unbiased}).
\end{proposition}

Thus, ``correlational appropriateness'' becomes testable: the question is no longer whether a correlation matrix can be formed, but whether the induced rank--1 correlational geometry predicts held-out responses.

\subsubsection{Random-items versus fixed-items: when classical conclusions can be biased}
\label{app:fixed-vs-random}

Many psychometric applications implicitly treat items as fixed (a specific instrument) and persons as random. AI benchmark settings often invert or complicate this: tasks/prompts/items can be treated as random draws from an evolving distribution; additionally, $n\ll p$ (few models, many tasks) or $n\gg p$ (many users, few items) are common. Image-based indices computed on $p\times p$ associations can become unstable or even ill-defined when $p$ is large relative to $n$, and may conflate item difficulty (marginals) with genuine shared structure.

Refactor/Verifactor help separate these issues because they evaluate reconstructions of $X$ itself and can explicitly control for marginals or fixed effects before measuring dependence structure (as in the partialed metrics described in the caption of Figure~\ref{fig:fafit_vs_refact}). In particular, the partial distance-correlation metric
$R^*_{X,\widehat X;\mathbb{E}(X)}$
quantifies dependence beyond what is explained by the independence model $\mathbb{E}(X)$ derived from row/column marginals, preventing spurious ``fit'' driven by prevalence alone.

\begin{corollary}[Marginal-driven fit is detectable]
\label{cor:marginal-detect-app}
If a reconstruction $\widehat X$ succeeds primarily by matching row/column marginals (e.g., predicting item means and person means) but not interaction structure, then metrics sensitive to residual dependence (e.g., bias-corrected dCor and $R^*$) will remain low even when likelihood or cosine similarity is moderately high.
\end{corollary}

This clarifies why multiple reconstruction metrics are necessary: a single scalar can be dominated by easy-to-predict marginals, while BCV dependence metrics detect whether the rank--1 model captures the \emph{interaction} signal that unidimensionality claims typically intend.

\subsubsection{Why a suite of reconstruction metrics is not redundant}
\label{app:metric-suite}

Classical fit indices are often highly correlated because they are deterministic transforms of the same image residuals. Refactor/Verifactor metrics are deliberately heterogeneous because reconstruction quality is multi-faceted:

\begin{itemize}\itemsep2pt
\item \textbf{Discrimination (AUC)} asks whether $\widehat X$ ranks positives above negatives (crucial in sparse binary matrices).
\item \textbf{Calibration/monotonic signal (isotonic $R^2$, Kendall $\tau$)} asks whether the model preserves ordinal structure even under nonlinear distortions.
\item \textbf{Geometric fidelity (matrix cosine)} emphasizes angle/orientation preservation of response patterns.
\item \textbf{Nonlinear dependence (bias-corrected dCor, partial dCor)} detects shared structure not reducible to linear correlation or marginals.
\item \textbf{Information-theoretic loss (cross-entropy / likelihood)} penalizes confident wrong predictions and is sensitive to rare events.
\end{itemize}

The metaregression table (bootstrapped Verifactor results) illustrates that associations can rank differently depending on which aspect of fit is most salient for the scientific question (e.g., discrimination vs.\ dependence vs.\ calibration). This is desirable: it forces explicit alignment between the evaluation criterion and the intended meaning of ``unidimensional signal'' in a given domain.

\section{Formal Treatment of Image-based fit, circularity, and data-level recoverability}\label{app:circularity}

\subsection{Image-first estimation as a composition of operators}
\label{app:composition}

Let $X\in\mathbb{R}^{n\times p}$ be an observed response matrix and let
$\mathcal{A}_c:\mathbb{R}^{n\times p}\to\mathbb{R}^{p\times p}$
be an association operator producing an item-image $A_c=\mathcal{A}_c(X)$.
Let $\mathcal{E}$ be an estimator mapping the image to a rank--1 representation,
e.g.,
\[
\mathcal{E}(A_c) = \widehat v \in \mathbb{R}^{p},
\qquad\text{with } A_c \approx \widehat v\,\widehat v^\top.
\]
Let $\mathcal{F}$ be an image-based fit functional (e.g., a CFI/TLI-type index, a residual norm, or a likelihood in the image space), so that classical evaluation returns
\begin{equation}
T_{\mathrm{img}}(X)
\;:=\;
\mathcal{F}\bigl(\mathcal{A}_c(X),\,\mathcal{E}(\mathcal{A}_c(X))\bigr).
\label{eq:image-test}
\end{equation}
By construction, $T_{\mathrm{img}}$ depends on $X$ \emph{only through} the image $\mathcal{A}_c(X)$.

In contrast, a data-level (Refactor) evaluation uses both row and column images
$\mathcal{A}_r(X)\in\mathbb{R}^{n\times n}$ and $\mathcal{A}_c(X)\in\mathbb{R}^{p\times p}$,
estimates $(\widehat u,\widehat v)$, reconstructs $\widehat X=\widehat u\,\widehat v^\top$ (possibly with calibration), and evaluates via a matrix metric $m$:
\begin{equation}
T_{\mathrm{data}}(X)
\;:=\;
m\!\left(X,\ \widehat X\right),
\qquad
\widehat X
=
\mathcal{R}\!\left(\widehat u,\widehat v\right),
\quad
(\widehat u,\widehat v)=\mathcal{E}_r(\mathcal{A}_r(X))\times\mathcal{E}_c(\mathcal{A}_c(X)).
\label{eq:data-test}
\end{equation}
Equation~\eqref{eq:data-test} is a test of the rank--1 hypothesis where it lives: in the original response matrix.

\subsection{A precise sense in which image-based evaluation can be self-confirming}
\label{app:selfconfirming}

Image-based evaluation is not ``wrong''; it answers a different question. The methodological vulnerability arises when the scientific claim concerns $X$ (e.g., ``a single latent trait explains the responses'') but the test statistic is effectively a property of $\mathcal{A}(X)$.

\begin{proposition}[Image-based fit cannot distinguish datasets sharing the same image]
\label{prop:image-equivalence}
Fix an association operator $\mathcal{A}_c$ and consider two matrices $X$ and $X'$ such that
$\mathcal{A}_c(X)=\mathcal{A}_c(X')$.
Then any image-based test statistic of the form \eqref{eq:image-test} satisfies
$T_{\mathrm{img}}(X)=T_{\mathrm{img}}(X')$,
even if $X$ and $X'$ differ substantially entrywise and have different recoverability by rank--1 reconstructions.
\end{proposition}

\begin{proof}
Immediate from \eqref{eq:image-test}, since $\mathcal{A}_c(X)$ is the only argument passed to $\mathcal{F}$.
\end{proof}

Proposition~\ref{prop:image-equivalence} formalizes the key limitation: once the association image is fixed, \emph{all} information not preserved by $\mathcal{A}_c$ is invisible to the evaluation. For binary and ordinal data, different $\mathcal{A}_c$ preserve different information; consequently, high image-based fit may reflect the internal coherence of the chosen association geometry rather than data-level recoverability.

\subsection{Recoverability as a necessary implication of a rank--1 signal model}
\label{app:recoverability}

To connect the rank--1 hypothesis to prediction, consider a signal-plus-noise model:
\begin{equation}
X = S + E,\qquad S = u v^\top,\quad \mathrm{rank}(S)=1.
\label{eq:signal-plus-noise}
\end{equation}
A rank--1 model is useful for unidimensional measurement only insofar as it yields a reconstruction $\widehat X$ close to $X$ (or $S$).

\begin{proposition}[Recoverability is a necessary condition for unidimensional adequacy]
\label{prop:recoverability-necessary}
If $X$ admits a rank--1 signal decomposition \eqref{eq:signal-plus-noise} with small noise (in the sense that $\|E\|$ is small in an application-relevant norm), then there exists a rank--1 matrix $\widehat X$ such that $\|X-\widehat X\|$ is small in the same norm. Conversely, if \emph{no} rank--1 matrix provides small reconstruction error (in-sample or out-of-sample), then either (i) the signal is not approximately rank--1, or (ii) the chosen association operator does not capture the signal necessary to estimate $(u,v)$, or (iii) the relevant structure is non-correlational (e.g., mixture, conditional dependence, nonlinear interactions).
\end{proposition}

\begin{proof}
The first statement is immediate by choosing $\widehat X=S$. The contrapositive gives the second statement.
\end{proof}

This proposition motivates Refactor: it evaluates the rank--1 premise by directly measuring whether the implied low-rank structure can recover the data matrix.

\begin{proposition}[Impossibility of Continuous, Injective Dimensionality Reduction]
A one-to-one function $g: \mathbb{R}^p \to \mathbb{R}^k$ with $k < p$ cannot be continuous.
\end{proposition}\label{prop:dim_reduction}
\begin{proof}
This is a direct consequence of the Borsuk-Ulam theorem, which states that for any continuous function $g: S^k \to \mathbb{R}^k$, there exists a point $\boldsymbol{x} \in S^k$ such that $g(\boldsymbol{x}) = g(-\boldsymbol{x})$. Consider the sphere $S^k \subset \mathbb{R}^{k+1}$. Since $k < p$, we can embed $S^k$ in the domain $\mathbb{R}^p$. If we assume $g$ is a continuous function from $\mathbb{R}^p \to \mathbb{R}^k$, its restriction to $S^k$ is also continuous. By the Borsuk-Ulam theorem, there must exist antipodal points $\boldsymbol{x}$ and $-\boldsymbol{x}$ on $S^k$ such that $g(\boldsymbol{x}) = g(-\boldsymbol{x})$. As $\boldsymbol{x} \neq -\boldsymbol{x}$, this contradicts the assumption that $g$ is one-to-one. Therefore, no such continuous, injective, one-to-one mapping function can exist.
\end{proof}

\begin{theorem}[Impossibility of Continuous]
For $n \ge 2$, there is no continuous injective map $f: I^n \to \mathbb{R}$.
\end{theorem}

\begin{proof}
Suppose, for the sake of contradiction, that there exists a continuous map $f: I^n \to \mathbb{R}$ that is injective.

Consider the restriction of $f$ to the boundary of the hypercube, denoted by $\partial I^n$. It is a well-known topological result that the boundary of an $n$-dimensional hypercube is homeomorphic to the $(n-1)$-sphere: $\partial I^n \cong S^{n-1} $. Let $h: S^{n-1} \to \partial I^n$ be this homeomorphism. We can define a new continuous composition map $g$ such that: $g: S^{n-1} \to \mathbb{R}, \quad g = f \circ h $. According to the \textbf{Borsuk--Ulam Theorem}, for any continuous map $g: S^k \to \mathbb{R}^k$, there exists a point $x \in S^k$ such that: $g(x) = g(-x)$. In our case, $k=1$ (since $n \ge 2$, $n-1 \ge 1$). Thus, there exists a pair of antipodal points $x, -x \in S^{n-1}$ such that: $f(h(x)) = f(h(-x))$. 

Since $x$ and $-x$ are distinct points on the sphere $S^{n-1}$, and $h$ is a homeomorphism (and thus injective), it follows that $h(x)$ and $h(-x)$ are two distinct points in the hypercube $I^n$. Let:
\[ p_1 = h(x), \quad p_2 = h(-x) \implies p_1 \neq p_2 \]

However, the equality $f(p_1) = f(p_2)$ contradicts our initial assumption that $f$ is injective. Therefore, no such continuous injective map can exist.
\end{proof}

\section{Simulations}\label{apx:sim_plots}

\subsection{Simulation configuration details}
Here we provide a summary of the main simulation from \cite{revelle_unidim_2025}, where they build on a model with a higher level (g) factor arising from the composite of three lower level factors \citep{jensen_what_1994}. From their article:

\begin{displayquote}
% \cite{jensen_what_1994} give a nice example of a higher level (g) factor arising from the composite of three lower level factors. 
Given a three-factor model for nine items with loadings as shown in Table 6 (Part A), and factor loadings on ``g'' of 0.9, 0.8, and 0.7 result in factor correlations of 0.72, 0.63, and 0.56...
% The 9 x 9 correlation matrix may be solved using a higher-order solution or rotated using the Schmid and Leiman (1957) oblique factor solution to produce the factor solution shown in Part B of Table 6. 
% The uw statistic for this set of nine items is 0.84, with @, of 0.69 and CFI of 0.93. 
This model was extended to the case of 18, 27, and 36 items by duplicating the loadings shown in Table 6 two, three, or four times. To examine the effect of factor structure, factor intercorrelations, and sample size, the higher-order loadings were set by specifying identical higher-order loadings varying from 0 to 1. Data were generated for 100, 200, 500, and 1,000 simulated cases with minor noise factors \citep{maccallum_representing_1991,maccallum_factor_2007} with random loadings of —0.2, 0, or 0.2. The data for nine, 18, 27, and 36 items were generated. 

\end{displayquote}

% , with either continuous or categorical items (with five categories). In the categorical case, polychoric correlations were found although we did not expect much effect of the categorical versus continuous distinction (Rhemtulla et al., 2012). Results for the u, @,, and CFI statistics are shown as a function of the number of items (9, 18, 27, and 36) and the general factor loading for both the continuous and categorical cases (Figure 2). Figure 3 shows the effects of the sample size (100, 200, 500, and 1,000) on these estimates. To examine the extent these results generalize across the size of group factors, we redid these analyses with smaller group factors (loadings of 0.6, 0.5, and 0.4) for the case of categorical variables (Figure 4) and sample size (Figure 5).

\section{Additional Empirical Validation and Results}\label{sec:addl_gram_res}

To empirically ground our framework, we also conducted a comprehensive analysis across 50 diverse datasets, directly comparing a suite of established, image-based unidimensionality metrics with the reconstruction fidelity metrics from our Refactoring Analyses. These are in Figures \ref{fig:univsrefact} and \ref{fig:corrpanel}. The central finding of this investigation is striking: we find no meaningful correlation between a model's adequacy as assessed by conventional image-based metrics and its ability to reconstruct the original data matrix. In other words, a model judged to be strongly unidimensional based on its factor loadings or the eigenvalues of its covariance matrix does not necessarily yield a reconstruction that accurately captures the signal in the observed responses. Furthermore, we find that the conventional metrics are highly intercorrelated, suggesting that they provide redundant information about the model's image, while offering little insight into its capacity to represent the underlying data.

This profound disconnect suggests that conclusions drawn from analyzing a model's image may not generalize to the preimage data matrix, $\boldsymbol{X}$. The assumptions inherent in constructing the Gram matrix (e.g., that signal is best captured by linear correlations) may inadvertently focus the analysis on a limited or even misleading form of variation. The robustness of this conclusion is a key contribution; our findings hold consistently across an additional extensive battery of image-based and reconstruction-based metrics, tested on 32 different Gramian constructions.

\subsection{Why Verifactor is not ``factor scores with CV'': preventing target leakage}
\label{sec:target-leakage}

A common reaction is to ask whether Refactor/Verifactor merely restate cross-validation of factor scores. They do not. Many pipelines compute factor scores for all persons using the full item set (or vice versa) and then evaluate prediction on the same $X$, or they cross-validate only one axis (rows or columns). Both practices can induce \emph{target leakage} in two-way arrays: the held-out cell $X_{ij}$ shares its row and column with many observed cells, so using the full row or full column to compute scores can implicitly expose information about $X_{ij}$.

Verifactor prevents this by holding out \emph{entire row subsets and column subsets} simultaneously. The predicted block $A$ shares neither the held-out rows nor the held-out columns with the block $D$ used to learn the low-rank structure (Lemma~\ref{lem:self-consistency}), aligning the evaluation with two-way generalization in crossed random designs.

\begin{corollary}[Leakage control under crossed holdout]
\label{cor:leakage}
Let $(i,j)$ index a BCV fold with held-out block $A_{ij}$ defined by row subset $\mathcal{I}$ and column subset $\mathcal{J}$. Then $\widehat A_{ij}$ computed from $D_{ij}$ is conditionally independent of $E_{A_{ij}}$ given $S$ under the model assumptions of Theorem~\ref{thm:bcv-unbiased}. In particular, any dependence of $\widehat A_{ij}$ on $A_{ij}$ must be mediated through the signal $S$ rather than direct reuse of held-out rows/columns.
\end{corollary}

This distinction matters in precisely the regimes emphasized by De~Boeck \citep{de_boeck_random_2008}: when persons and items are random, a convincing unidimensionality claim requires that the latent structure generalizes to new persons and new items, not merely to new entries within seen persons or items.

\subsection{Random-items versus fixed-items: when classical conclusions can be biased}
\label{sec:fixed-vs-random}

Many psychometric applications implicitly treat items as fixed (a specific instrument) and persons as random. AI benchmark settings often invert or complicate this: tasks/prompts/items can be treated as random draws from an evolving distribution; additionally, $n\ll p$ (few models, many tasks) or $n\gg p$ (many users, few items) are common. Image-based indices computed on $p\times p$ associations can become unstable or even ill-defined when $p$ is large relative to $n$, and may conflate item difficulty (marginals) with genuine shared structure.

Refactor/Verifactor help separate these issues because they evaluate reconstructions of $X$ itself and can explicitly control for marginals or fixed effects before measuring dependence structure (as in the partialed metrics described in the caption of Figure~\ref{fig:fafit_vs_refact}). In particular, the partial distance-correlation metric
$R^*_{X,\widehat X;\mathbb{E}(X)}$
quantifies dependence beyond what is explained by the independence model $\mathbb{E}(X)$ derived from row/column marginals, preventing spurious ``fit'' driven by prevalence alone.

\begin{corollary}[Marginal-driven fit is detectable]
\label{cor:marginal-detect}
If a reconstruction $\widehat X$ succeeds primarily by matching row/column marginals (e.g., predicting item means and person means) but not interaction structure, then metrics sensitive to residual dependence (e.g., bias-corrected dCor and $R^*$) will remain low even when likelihood or cosine similarity is moderately high.
\end{corollary}

This clarifies why multiple reconstruction metrics are necessary: a single scalar can be dominated by easy-to-predict marginals, while BCV dependence metrics detect whether the rank--1 model captures the \emph{interaction} signal that unidimensionality claims typically intend.

% \subsubsection{Why a suite of reconstruction metrics is not redundant}
% \label{sec:metric-suite}

% Classical fit indices are often highly correlated because they are deterministic transforms of the same image residuals. Refactor/Verifactor metrics are deliberately heterogeneous because reconstruction quality is multi-faceted:

% \begin{itemize}\itemsep2pt
% \item \textbf{Discrimination (AUC)} asks whether $\widehat X$ ranks positives above negatives (crucial in sparse binary matrices).
% \item \textbf{Calibration/monotonic signal (isotonic $R^2$, Kendall $\tau$)} asks whether the model preserves ordinal structure even under nonlinear distortions.
% \item \textbf{Geometric fidelity (matrix cosine)} emphasizes angle/orientation preservation of response patterns.
% \item \textbf{Nonlinear dependence (bias-corrected dCor, partial dCor)} detects shared structure not reducible to linear correlation or marginals.
% \item \textbf{Information-theoretic loss (cross-entropy / likelihood)} penalizes confident wrong predictions and is sensitive to rare events.
% \end{itemize}

The metaregression table (bootstrapped Verifactor results) illustrates that associations can rank differently depending on which aspect of fit is most salient for the scientific question (e.g., discrimination vs.\ dependence vs.\ calibration). This is desirable: it forces explicit alignment between the evaluation criterion and the intended meaning of ``unidimensional signal'' in a given domain.

\begin{table}%{llrrrrrr}
\centering
\caption{Examples of differences of Verifactor reconstruction fit estimates across 6 datasets, estimated under 8 different assumed underlying relationships.  }
\footnotesize
\begin{tabular}{llrrrrrr}
\toprule
table & idx & $\cos_F$ & $\operatorname{dCor^2_n}$ & $\mathcal{L}(X|\theta)$ & $R^*_{X,\hat{X};\operatorname{E}(X)}$ & AUC & $\tau_b$ \\
\midrule
PMT\_Trzcinska\_2023\_PMT & acc & 0.74 & 0.15 & 0.51 & 0.07 & 0.59 & 0.13\\
PMT\_Trzcinska\_2023\_PMT & bsym\_U & 0.74 & 0.07 & 0.51 & 0.06 & 0.53 & 0.04\\
PMT\_Trzcinska\_2023\_PMT & hi & \textbf{0.76} & 0.45 & 0.52 & 0.47 & 0.63 & \textbf{0.19}\\
PMT\_Trzcinska\_2023\_PMT & krip & \textbf{0.76} & \textbf{0.46} & \textbf{0.53} & \textbf{0.48} & \textbf{0.64} & \textbf{0.19}\\
PMT\_Trzcinska\_2023\_PMT & phi & \textbf{0.76} & 0.45 & 0.52 & 0.47 & 0.63 & 0.18\\
PMT\_Trzcinska\_2023\_PMT & qcr & 0.75 & 0.44 & 0.52 & 0.46 & 0.62 & 0.17\\
PMT\_Trzcinska\_2023\_PMT & tet & \textbf{0.76} & 0.45 & 0.52 & \textbf{0.48} & 0.63 & \textbf{0.19}\\
PMT\_Trzcinska\_2023\_PMT & yule\_q & \textbf{0.76} & 0.44 & 0.52 & 0.46 & 0.63 & 0.18\\
\addlinespace
brand\_raffaelli\_2024\_recognition\_20 & acc & 0.74 & 0.38 & 0.52 & 0.36 & 0.58 & 0.11\\
brand\_raffaelli\_2024\_recognition\_20 & bsym\_U & 0.72 & 0.07 & 0.50 & 0.05 & 0.51 & 0.02\\
brand\_raffaelli\_2024\_recognition\_20 & hi & 0.88 & 0.97 & 0.70 & 0.97 & 0.91 & \textbf{0.58}\\
brand\_raffaelli\_2024\_recognition\_20 & krip & 0.88 & 0.97 & 0.70 & 0.97 & 0.91 & \textbf{0.58}\\
brand\_raffaelli\_2024\_recognition\_20 & phi & 0.88 & 0.97 & 0.70 & 0.97 & 0.91 & \textbf{0.58}\\
brand\_raffaelli\_2024\_recognition\_20 & qcr & 0.88 & 0.97 & 0.70 & 0.97 & 0.91 & \textbf{0.58}\\
brand\_raffaelli\_2024\_recognition\_20 & tet & 0.88 & 0.97 & 0.70 & 0.97 & 0.91 & \textbf{0.58}\\
brand\_raffaelli\_2024\_recognition\_20 & yule\_q & 0.88 & 0.97 & 0.69 & 0.97 & 0.91 & \textbf{0.58}\\
\addlinespace
ccapsvtskhpacr\_mercedes\_2023\_physical & acc & 0.62 & 0.43 & 0.68 & 0.06 & 0.78 & \textbf{0.32}\\
ccapsvtskhpacr\_mercedes\_2023\_physical & bsym\_U & 0.49 & 0.04 & 0.65 & -0.01 & 0.59 & -0.02\\
ccapsvtskhpacr\_mercedes\_2023\_physical & hi & 0.40 & 0.06 & 0.67 & 0.04 & 0.53 & -0.01\\
ccapsvtskhpacr\_mercedes\_2023\_physical & krip & 0.48 & 0.12 & 0.61 & 0.09 & 0.52 & 0.00\\
ccapsvtskhpacr\_mercedes\_2023\_physical & phi & 0.54 & 0.17 & 0.64 & \textbf{0.10} & 0.60 & -0.03\\
ccapsvtskhpacr\_mercedes\_2023\_physical & qcr & \textbf{0.67} & \textbf{0.47} & 0.71 & -0.01 & \textbf{0.84} & -0.39\\
ccapsvtskhpacr\_mercedes\_2023\_physical & tet & 0.56 & 0.42 & \textbf{0.76} & 0.09 & 0.83 & -0.31\\
ccapsvtskhpacr\_mercedes\_2023\_physical & yule\_q & 0.51 & 0.08 & 0.63 & 0.06 & 0.56 & -0.01\\
\addlinespace
legalbench & acc & \textbf{0.87} & 0.66 & \textbf{0.63} & 0.12 & \textbf{0.84} & \textbf{0.46}\\
legalbench & hi & 0.83 & 0.55 & 0.57 & \textbf{0.70} & 0.69 & 0.26\\
legalbench & krip & 0.82 & 0.61 & 0.56 & 0.68 & 0.68 & 0.25\\
legalbench & phi & 0.84 & 0.58 & 0.58 & 0.42 & 0.75 & 0.34\\
legalbench & qcr &\textbf{ 0.87} & \textbf{0.67} & 0.62 & 0.20 & 0.83 & 0.44\\
legalbench & tet & 0.84 & 0.57 & 0.58 & 0.39 & 0.75 & 0.34\\
legalbench & yule\_q & 0.84 & 0.52 & 0.57 & 0.35 & 0.73 & 0.32\\
\addlinespace
project\_kids\_topel & acc & \textbf{0.96} & \textbf{0.39} & \textbf{0.77} & -0.01 & \textbf{0.78} & \textbf{0.23}\\
project\_kids\_topel & bsym\_U & \textbf{0.96} & 0.18 & 0.75 & 0.16 & 0.56 & 0.05\\
project\_kids\_topel & krip & 0.95 & 0.19 & 0.74 & \textbf{0.20} & 0.53 & 0.02\\
project\_kids\_topel & phi & 0.95 & 0.16 & 0.75 & 0.15 & 0.60 & 0.09\\
project\_kids\_topel & qcr & \textbf{0.96} & \textbf{0.39} & \textbf{0.77} & 0.00 & \textbf{0.78} & \textbf{0.23}\\
project\_kids\_topel & yule\_q & 0.95 & 0.14 & 0.74 & 0.13 & 0.57 & 0.06\\
\addlinespace
project\_kids\_wj\_lwid\_grade & acc & 1.00 & 0.16 & \textbf{0.98} & \textbf{0.08} & 0.74 & 0.04\\
project\_kids\_wj\_lwid\_grade & bsym\_U & 1.00 & \textbf{0.39} & 0.97 & -0.23 & \textbf{0.87} & \textbf{0.09}\\
project\_kids\_wj\_lwid\_grade & krip & 1.00 & 0.11 & \textbf{0.98} & -0.01 & 0.65 & 0.02\\
project\_kids\_wj\_lwid\_grade & phi & 1.00 & 0.16 & \textbf{0.98} & 0.04 & 0.73 & 0.04\\
project\_kids\_wj\_lwid\_grade & qcr & 1.00 & 0.17 & \textbf{0.98} & 0.05 & 0.73 & 0.04\\
project\_kids\_wj\_lwid\_grade & yule\_q & 1.00 & 0.13 & \textbf{0.98} & 0.04 & 0.72 & 0.04\\
\bottomrule
\end{tabular}\caption*{\footnotesize
table column displays the dichotomous dataset. 
"idx" column captures the types of association measures assumed in each rank-1 model. acc = interitem/interobservation acc (which is a dot product in this case), bsym\_U = symmetrized information theoretic uncertainty, hi = Loevinger's H, krip = Krippendorff's alpha, phi = Pearson correlation on binary data, qcr = Mosteller's quadrant correlation/quadrant count ratio, tet = tetrachoric correlation, and yule\_q = Yule's Q. Metrics used for reconstruction include: $\cos_F$: the matrix cosine using Frobinius norm; $\operatorname{dCor^2_n}$: mean bias-corrected squared distance correlation (mean from measuring a reconstruction and its transposed dual);  $\mathcal{L}(X|\theta)$: Geometric mean of likelihood; $R^*_{X,\hat{X};\operatorname{E}(X)}$: the mean unbiased partial distance correlation, where the independence matrix (from the scaled out product of the marginals) is partialled out to measure the dependence structure; AUC: the area under the ROC curve ;  $\tau_b$: the vectorized Kendall's tau correlation between the two matrices
}
\end{table}

\begin{table}[ht]
\centering
\resizebox{\textwidth}{!}{%
\begin{tabular}{p{2.3cm} p{3.7cm} p{3.4cm} p{3.0cm} }
\hline
\textbf{Metric} & \textbf{Formula / Definition} & \textbf{Intended Measure} & \textbf{Assumptions} \\
\hline
% Image residual & $\|A-\widehat A\|_F$ & Entrywise error in image fit & Quadratic loss; scale-sensitive & Simple, direct & Sensitive to scale; dominated by large entries & Null: $\|A\|_F$ \\
% Image $R^2$ & $1 - \|A-\widehat A\|_F^2/\|A-\bar A\mathbf{11}^\top\|_F^2$ & Variance explained in image & Linear variance decomposition & Easy to interpret & Depends on centering choice & Global mean image \\
$u_{rc}$ & $\rho_c\tau_{RC} $  & Holistic unidimensionality& Image-dependent\\
$\rho_c$&$(F_o - F_m) / F_o$  & Image to factor model fit.& Congeneric fit on image represents dimensional fit\\
$\tau_{RC}$ & $1 - \sum(A- \bar A)^2/\sum(A-\bar a)^2$  & Correlational residuals& Approx. multivariate normality \\
% CFI & $1 - \frac{\max(\chi^2_m-\text{df}_m,0)}{\max(\chi^2_b-\text{df}_b,0)}$ & Incremental fit vs. baseline model & Approx. multivariate normality \\
% TLI & Adjusted CFI penalized by df & Parsimonious comparative fit & Same as CFI \\
ECV & $\sum \lambda^{\rm common}_k/\sum \lambda^{\rm total}_i$ & Proportion variance due to common factors & Factor model specified \\
Cronbach’s $\alpha$ & $\tfrac{p}{p-1}(1-\sum\sigma_i^2/\sigma^2_{\rm total})$ & Internal consistency reliability & Tau-equivalence; continuous approx. \\
\hline
AUC & $\Pr(\hat p^{(1)}>\hat p^{(0)})$ & Rank discrimination ability & Rank-based \\

$r_2$ (Ramsay) & $1 - \sum(X-\widehat X)^2/\sum(X-\bar x)^2$ & Global variance explained & Quadratic error \\
Yanai GDC & $1-d(X,\widehat X)/d(X,\widehat X^{\rm base})$ & Normalized similarity of matrices & Choice of dissimilarity metric \\
Cross-entropy & $-\tfrac{1}{nm}\sum[X\log\hat p+(1-X)\log(1-\hat p)]$ & Probabilistic prediction loss & Valid prob. outputs \\
$\text{dCor}^2_n$ & $\mathrm{dCor}(X,\widehat X)$ & Nonlinear dependence & None; metric-based \\
Partial dCor & $\mathrm{dCor}(X,\widehat X|M)$ & Dependence beyond marginals & Valid $M$ specified \\
Kendall (similar to ) & Rank correlation & Monotone assoc. of entries & Monotonicity \\
$r_1$ (Ramsay) & $\operatorname{corr}(\text{vec}(X),\text{vec}(\widehat X))$ & Linear assoc. of entries & Linear relation \\
Frobenius cosine & $\cos\theta=\frac{\langle X-\bar X,\widehat X-\overline{\widehat X}\rangle_F}{\|X-\bar X\|_F\|\widehat X-\overline{\widehat X}\|_F}$ & Geometric similarity of matrices & Centering choice \\
KL divergence & $\text{KL}(P_X||P_{\hat X})$ & Distributional fidelity & Valid prob. est. \\
Rare-event recall & $\sum_{(i,j)\in \mathcal R}\mathbf{1}\{X_{ij}=1,\hat{X}_{ij}=1\}/\sum_{(i,j)\in \mathcal R}\mathbf{1}\{X_{ij}=1\}$ & Recovery of rare signals & Definition of rarity \\
\hline
\end{tabular}%
}
\caption{Examples of image-fit and preimage-fit metrics for evaluating reconstructions of binary data matrices. Each metric’s intent, assumptions, and formulae are included. All image-dependent measures of unidimensionality were estimated using \texttt{psych} \cite{revelle_psych_2024}, } AUC used the \texttt{pROC} package \cite{robin_proc_2011}, distance measures used \texttt{energy} \cite{szekely_energy_2013} , Ramsay and Yanai's measures used \texttt{MatrixCorrelation} \cite{indahl_similarity_2018}, and Kendall's tau used \texttt{pcaPP} \cite{filzmoser_pcapp_2024}. The RFI metrics are supplemented by partial distance correlation (\cite{szekely_partial_2014}) to specifically isolate dependence structures that are not captured by simple marginal effects. To assess the degree of shared structure and subspace overlap, we utilize Yanai's Generalized Coefficient of Determination (\cite{yanai_proposition_1980}) (GCD) and Ramsay's matrix correlation coefficients ($r_1$ and $r_2$) (\cite{ramsay_matrix_1984}), which quantify the alignment between the principal components and spectral properties of the observed and reconstructed data. For metrics that are dependent on matrix orientation, we ensure a comprehensive and symmetric evaluation by computing each metric on both the reconstruction and its transpose and reporting the average, thereby providing a robust assessment of structural fidelity and providing recommendations and heuristics for baselines. Table \ref{tab:metrics} with brief descriptions of these and of standard factor analyses and software used in estimation..\label{tab:metrics}
\end{table}\label{tab:metric}

\begin{table}
\centering
\caption{Metaregression of Bootstrapped Verifactor Tables for all Datasets}
\resizebox{\ifdim\width>\linewidth\linewidth\else\width\fi}{!}{
\begin{tabular}{lllllll}
\toprule
idx & $\cos_F$ & $\operatorname{dCor^2_n}$ & $\mathcal{L}(X|\theta)$ & $R^*_{X,\hat{X};\operatorname{E}(X)}$ & AUC & $\tau_b$\\
\midrule
% AC1 & 0.82*** (0.01) & 0.5*** (0.01) & 0.65*** (0.01) & 0.04*** (0.01) & 0.77*** (0.01) & 0.21*** (0.01)\\
% acc & 0.81*** (0.01) & 0.47*** (0.01) & 0.65*** (0.01) & 0.07*** (0.01) & 0.77*** (0.01) & 0.32*** (0.01)\\
$U_{sym}$ & 0.77*** (0.01) & 0.2*** (0.01) & 0.6*** (0.01) & 0.11*** (0.01) & 0.62*** (0.01) & 0.08*** (0.01)\\
% dcorU & 0.76*** (0.01) & 0.17*** (0.01) & 0.59*** (0.01) & 0.08*** (0.01) & 0.59*** (0.01) & 0.08*** (0.01)\\
% E & 0.83*** (0.01) & 0.53*** (0.01) & 0.66*** (0.01) & 0.05*** (0.01) & 0.79*** (0.01) & 0.35*** (0.01)\\
\addlinespace
$\tau_{gk}$ & \textit{0.78}*** (0.01) & \textit{0.35}*** (0.01) & \textit{0.61}*** (0.01) & \textit{0.17}*** (0.01) & 0.63*** (0.01) & 0.09*** (0.01)\\
$H_i$ & 0.77*** (0.01) & 0.21*** (0.01) & 0.6*** (0.01) & \textbf{0.23}*** (0.01) & 0.58*** (0.01) & 0.03*** (0.01)\\
$\alpha_{krip}$ & 0.77*** (0.01) & 0.33*** (0.01) & 0.59*** (0.01) & 0.2*** (0.01) & 0.58*** (0.01) & 0.06*** (0.01)\\
$\phi$ & \textit{0.78}*** (0.01) & 0.25*** (0.01) & \textit{0.61}*** (0.01) & 0.15*** (0.01) & \textit{0.64}*** (0.01) & \textit{0.11}*** (0.01)\\
$q^\prime$ & \textbf{0.82***} (0.01) & \textbf{0.49}*** (0.01) & \textbf{0.65}*** (0.01) & 0.06*** (0.01) & \textbf{0.77}*** (0.01) & \textbf{0.21}*** (0.01)\\
% \addlinespace
% star\_q & 0.82*** (0.01) & 0.49*** (0.01) & 0.65*** (0.01) & 0.09*** (0.01) & 0.77*** (0.01) & 0.23*** (0.01)\\
$\rho_{tet}$ & \textit{0.78}*** (0.01) & 0.25*** (0.01) & \textit{0.61}*** (0.01) & 0.16*** (0.01) & 0.63*** (0.01) & \textit{0.11}*** (0.01)\\
$Q_{yule}$ & 0.77*** (0.01) & 0.21*** (0.01) & 0.6*** (0.01) & 0.12*** (0.01) & 0.61*** (0.01) & 0.09*** (0.01)\\
\bottomrule
\end{tabular}}
\caption*{Multivariate metaregression is performed across each of 6 Verifactor unidimensionality metrics as dependent variables, with Verifactor various . Bolded numbers represent the best performance for that metric}
\end{table}

\section{Supplemental Figures}\label{apx:supp_figures}

\begin{figure*}[th]
    \centering
    \includegraphics[width=0.8\linewidth]{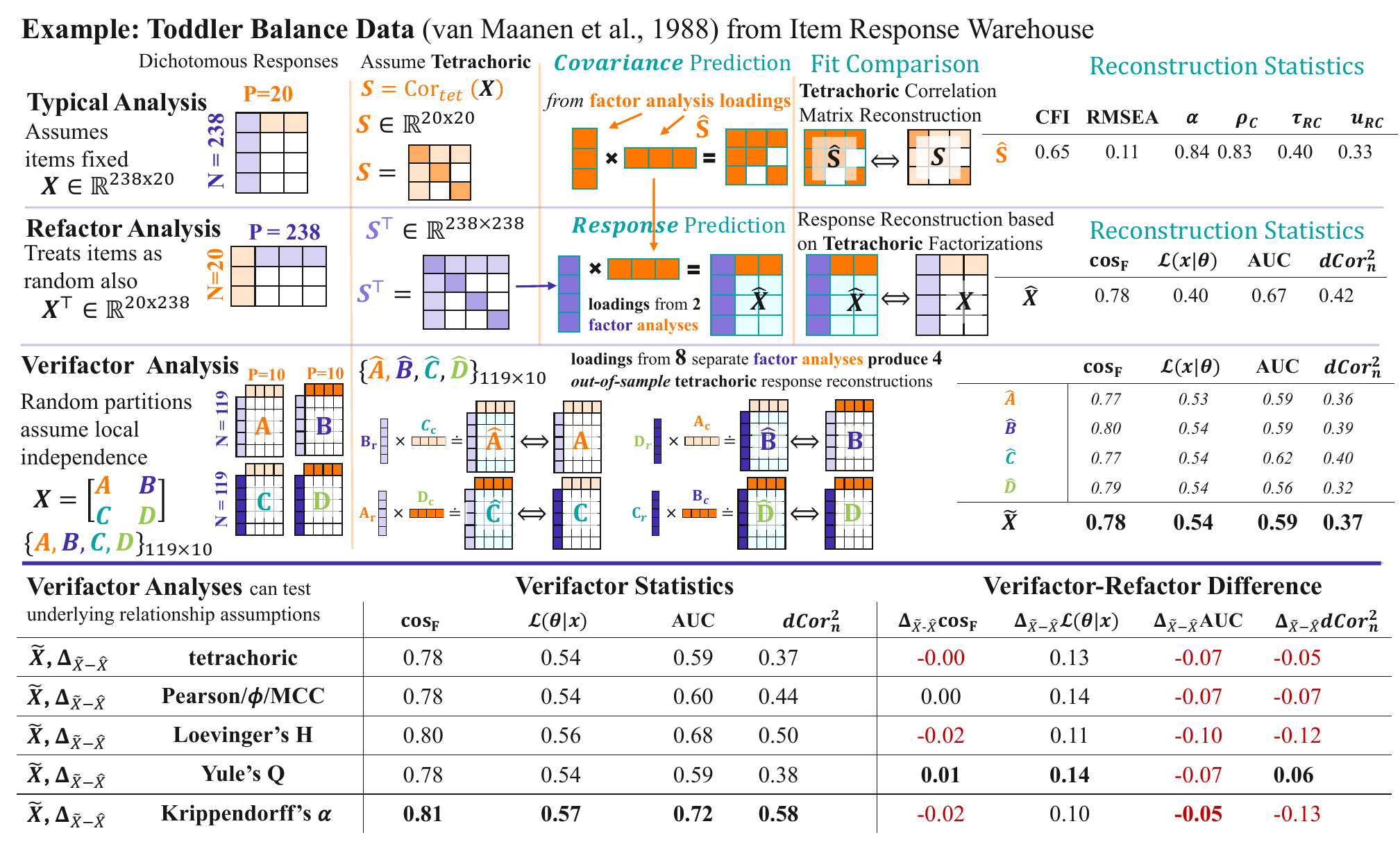}
    \caption{Example Refactor and Verifactor Analyses to test assumptions of underlying relationships as illustrated in Figure \ref{fig:refactor}.}
    \label{fig:ex_verifactor_apx}
\end{figure*}

\begin{figure*}[h!]
    \centering
    \includegraphics[width=0.8\linewidth]{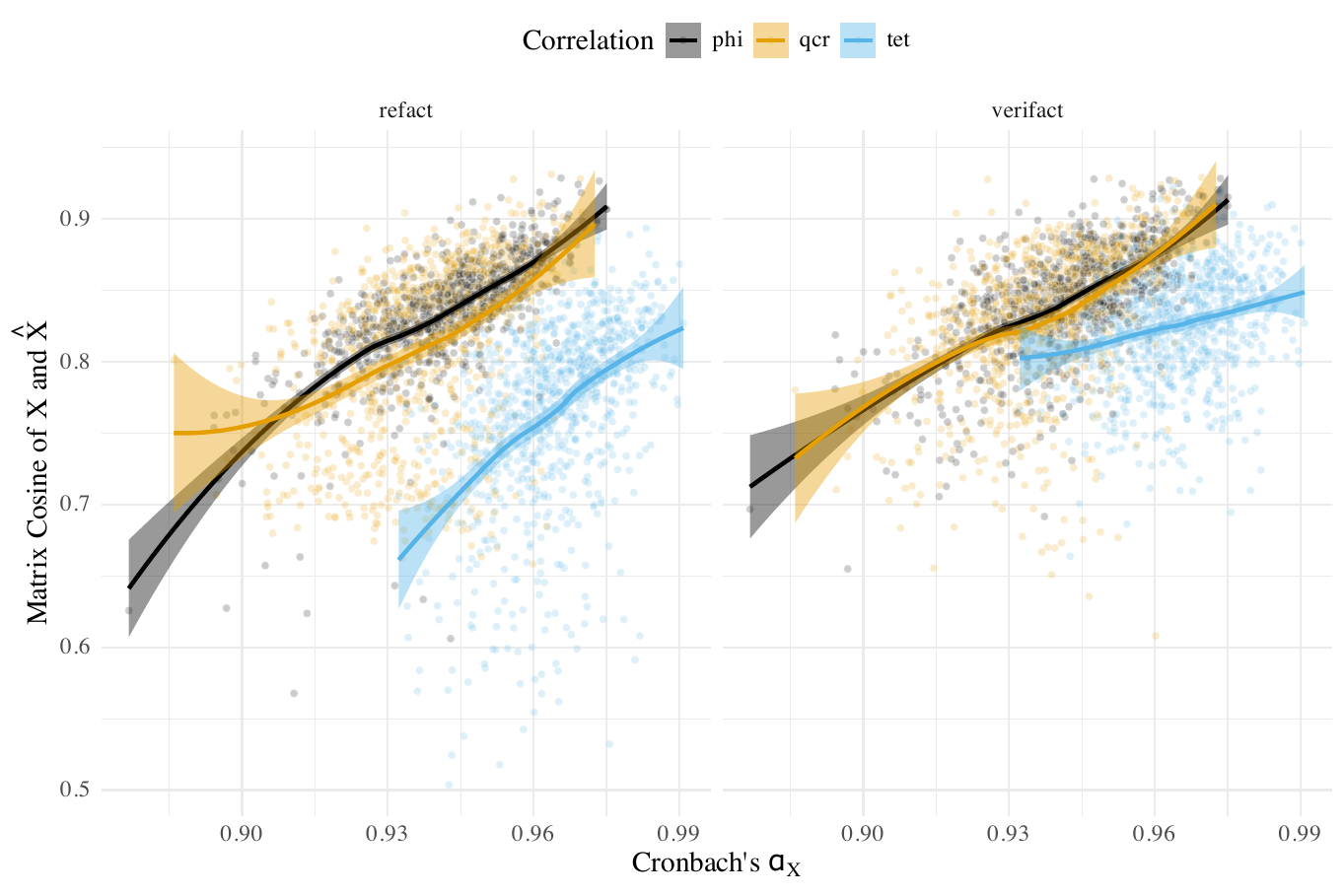}
        \caption{\textbf{Classical Unidimensionality (Cronbach's $\alpha_\mathbf{X}$) vs Refactoring}: matrix cosine between $\mathbf{X}$ and $\widehat{\mathbf{X}}$: $\cos(\mathbf{X}, \mathbf{\widehat{X}}) = \langle \mathbf{X}, \mathbf{\widehat{X}} \rangle_F / \|\mathbf{X}\|_F \|\mathbf{\widehat{X}}\|_F$ across three correlations. Loess fit and confidence intervals are shown for 1000 replications. Fully crossed linear relationships between classical measures of unidimensionality and refactoring measures are in Figure \ref{fig:basesimpanels} }
    \label{fig:phi_dgm_vs_auc_apx}
\end{figure*}

\begin{figure*}
    \centering
    \includegraphics[width=1\linewidth]{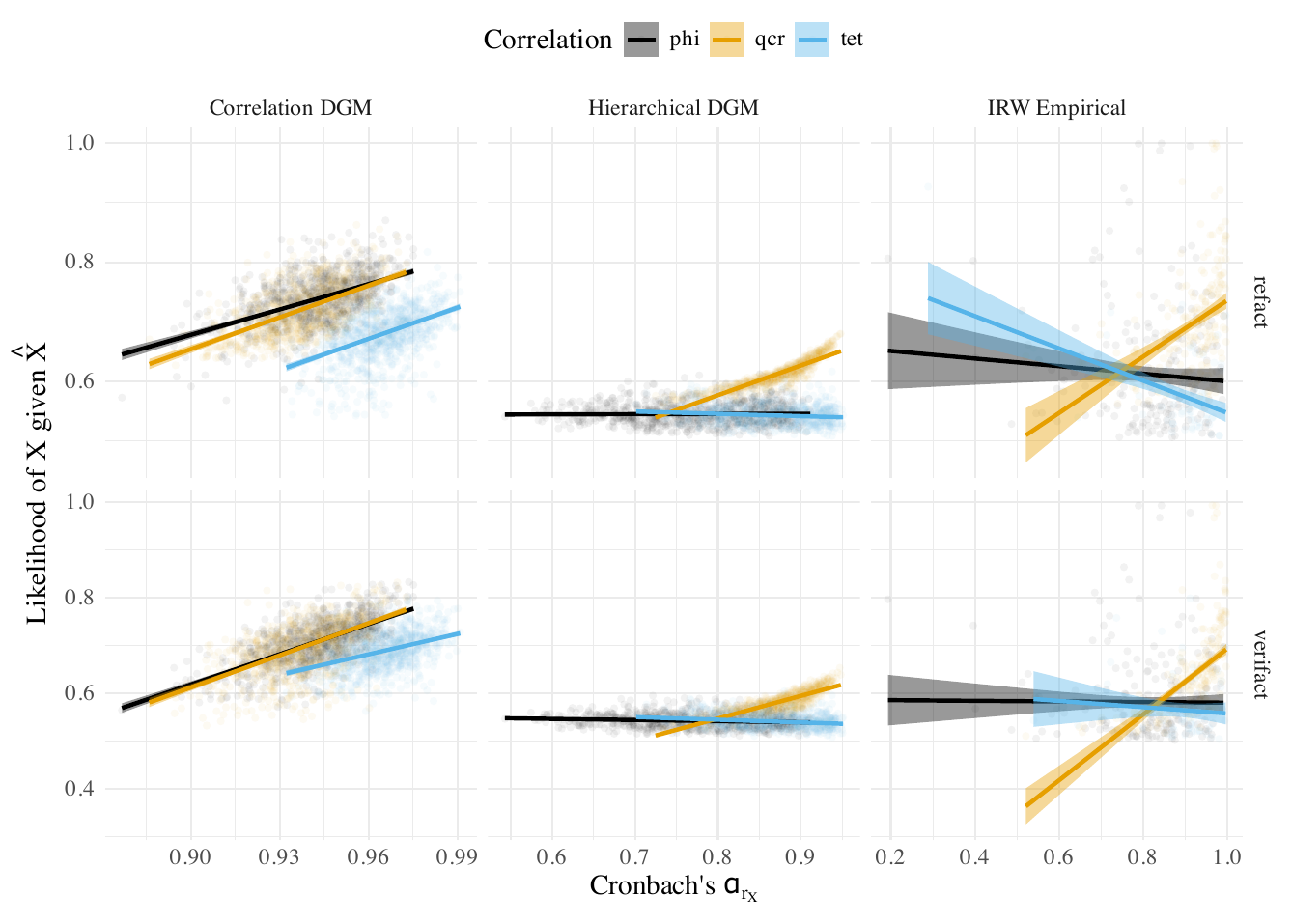}
    \caption{\textbf{Cronbach's $\alpha$ vs Likelihood of the data} given the reconstruction, $\widehat{X}$ across three different conditions. \textbf{(left)} 1000 simulations where the underlying data generating models (DGM) are unidimensional tetrachoric correlations (see Section \ref{sec:simple_sim}. \textbf{(middle)} a hierarchical DGM with minor noise factors, following \cite{revelle_unidim_2025} (see Section \ref{sec:unidim_reprod}. \textbf{(right)} 200 publicly available empirical datasets using the Item Response Warehouse. \textbf{(top)} Refactor Analysis and \textbf{(bottom)} Verifactor out-of-sample bi-cross validated prediction. \textbf{(color)}represents different correlational relationships.}
    \label{fig:alpha_lik}
\end{figure*}

\begin{figure*}[h!]
    \centering
    \includegraphics[width=1\linewidth]{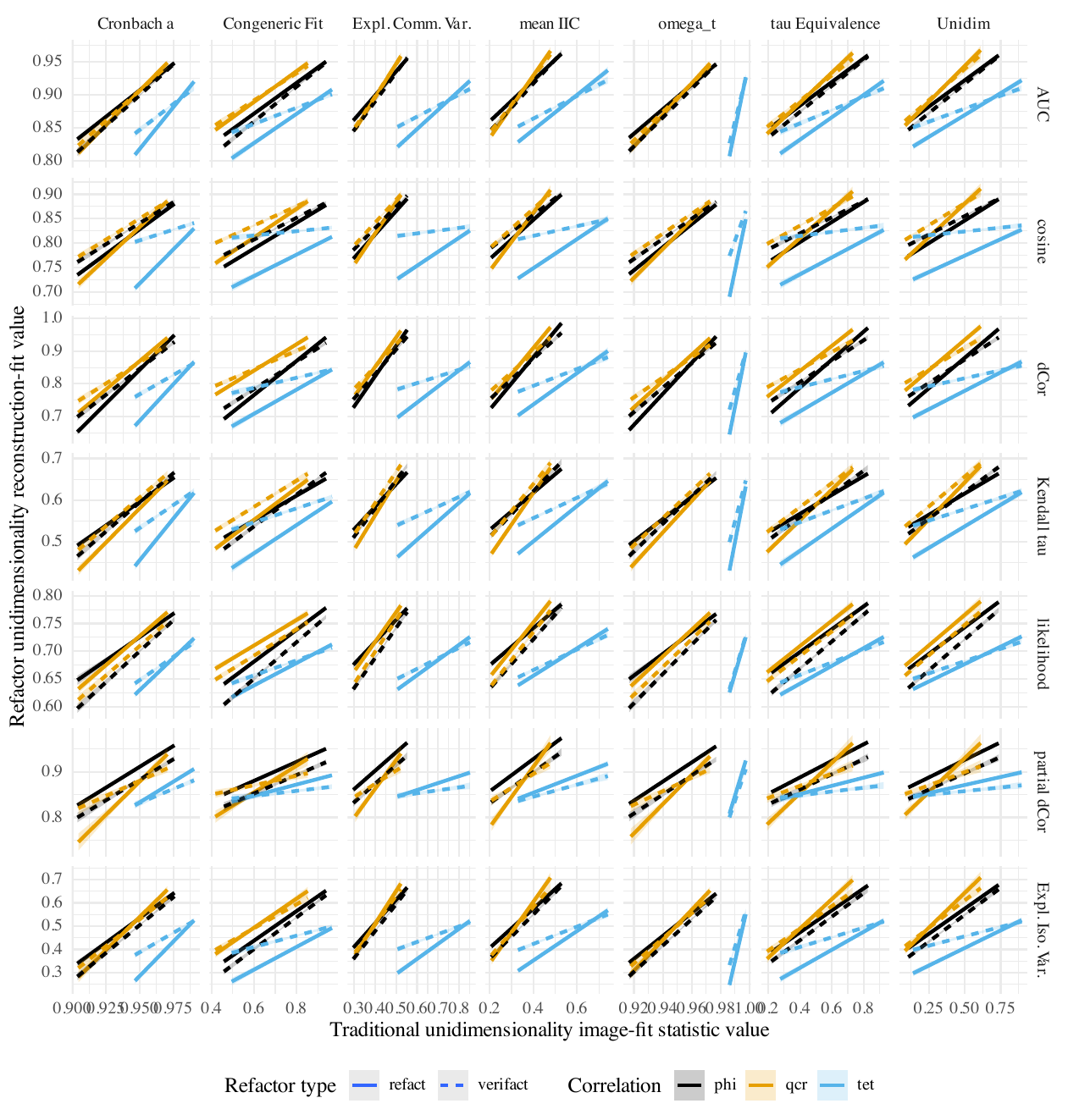}
    \caption{\textbf{Alignment between Classical and Refactor Measures of Unidimensionality: Simulation 1} Classical Measures of Unidimensionality (columns) and their estimated values (x-axis) in the base simulation shown in Figure \ref{fig:corr_refactor_base_sim} vs Refactor evaluation metrics (Section \ref{sec:evalmeths}) of unidimensional fit (rows) and their estimates (y-axes) across three correlations. Robust M-estimator linear fits and confidence intervals are shown for 1000 replications.}
    \label{fig:basesimpanels}
\end{figure*}

\begin{figure*}[h!]
    \centering
    \includegraphics[width=1\linewidth]{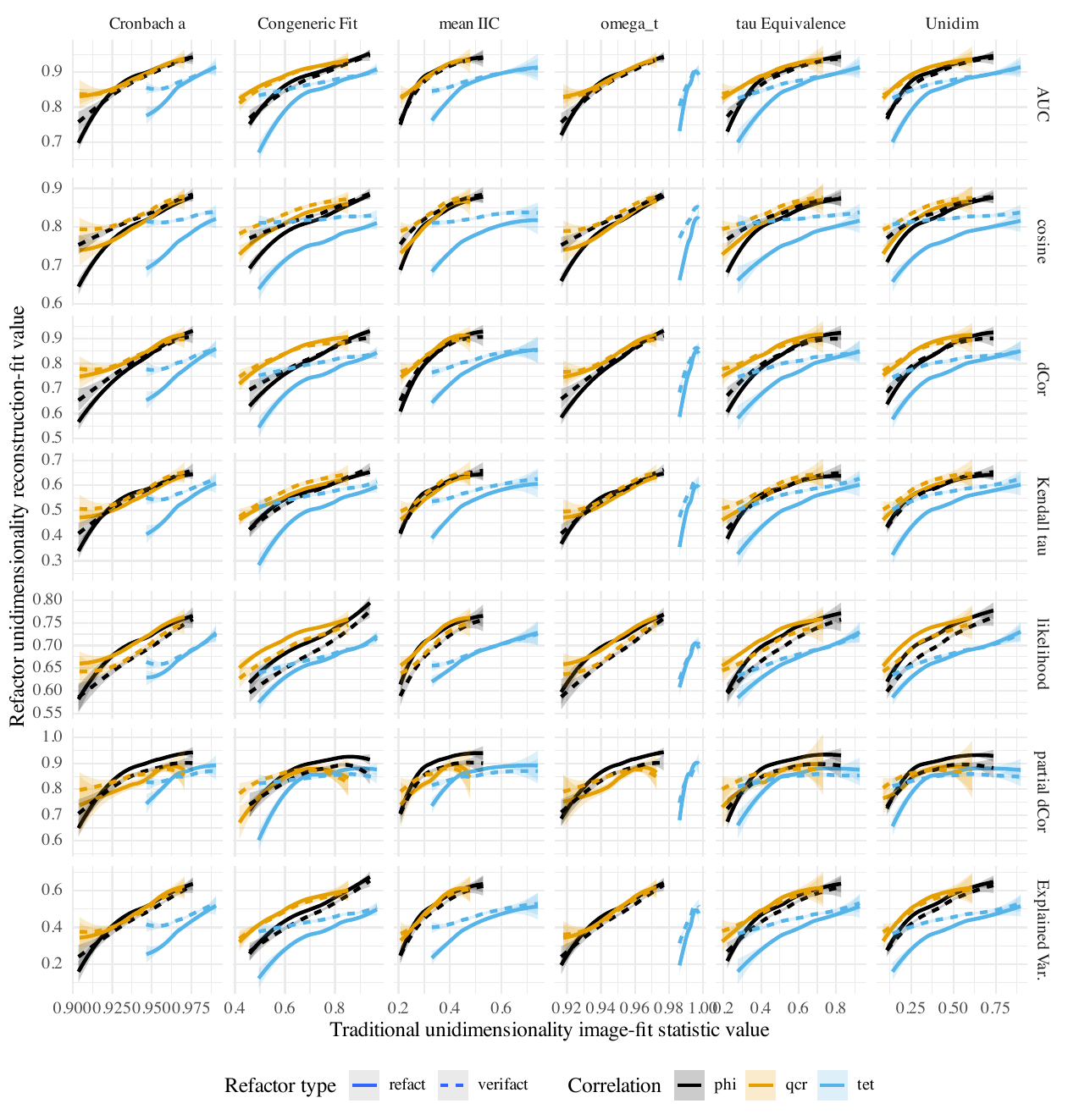}
    \caption{\textbf{Alignment between Classical and Refactor Measures of Unidimensionality: Simulation 1} Classical Measures of Unidimensionality (columns) and their estimated values (x-axis) in the base simulation shown in Figure \ref{fig:corr_refactor_base_sim} vs Refactor evaluation metrics (Section \ref{sec:evalmeths}) of unidimensional fit (rows) and their estimates (y-axes) across three correlations. Loess fits and confidence intervals are shown for 1000 replications.}
    \label{fig:basesimpanelsloess}
\end{figure*}

\begin{figure}[h]
    \centering
    \includegraphics[width=1\linewidth]{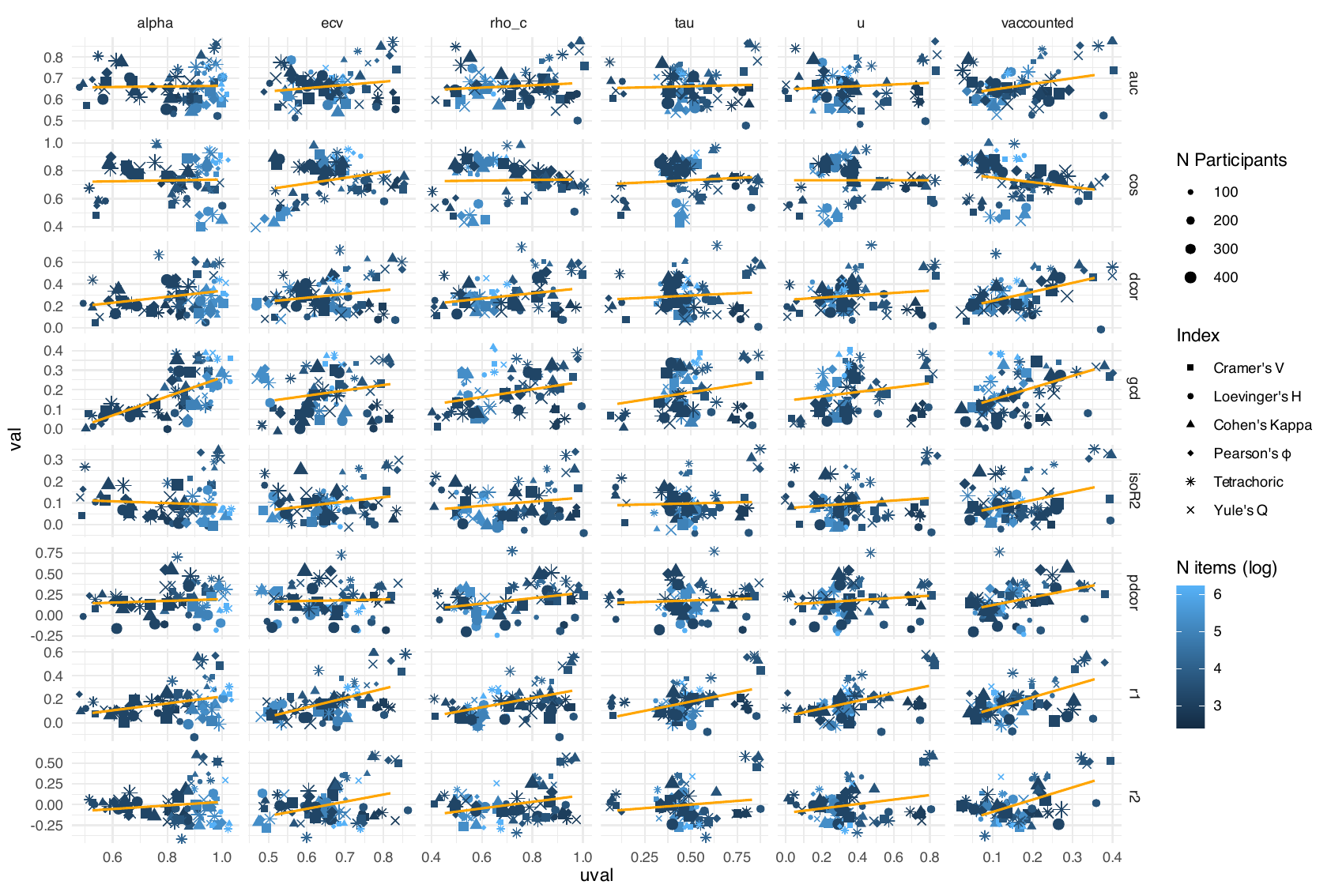}
    \caption{Lack of Relationships between Unidimensional models' abilities to reconstruct the source variation and the image-based measures of dimensionality.}
    \label{fig:univsrefact}
\end{figure}

\begin{figure}[h]
    \centering
    \includegraphics[width=1\linewidth]{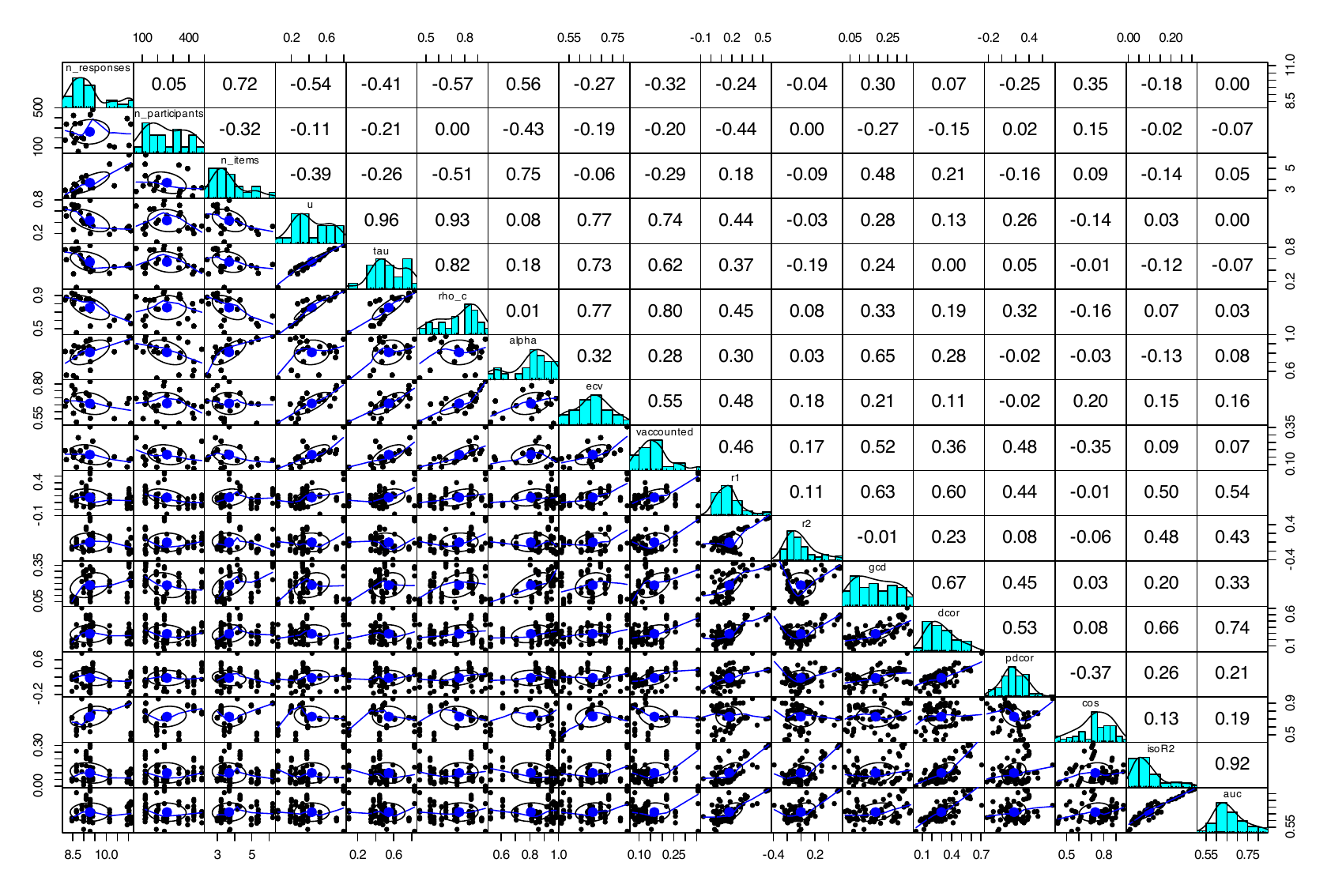}
    \caption{Correlations of the reported image-based metrics, refactor metrics, and dataset properties.}
    \label{fig:corrpanel}
\end{figure}

\begin{figure*}[h]
    \centering
    \includegraphics[width=1\linewidth]{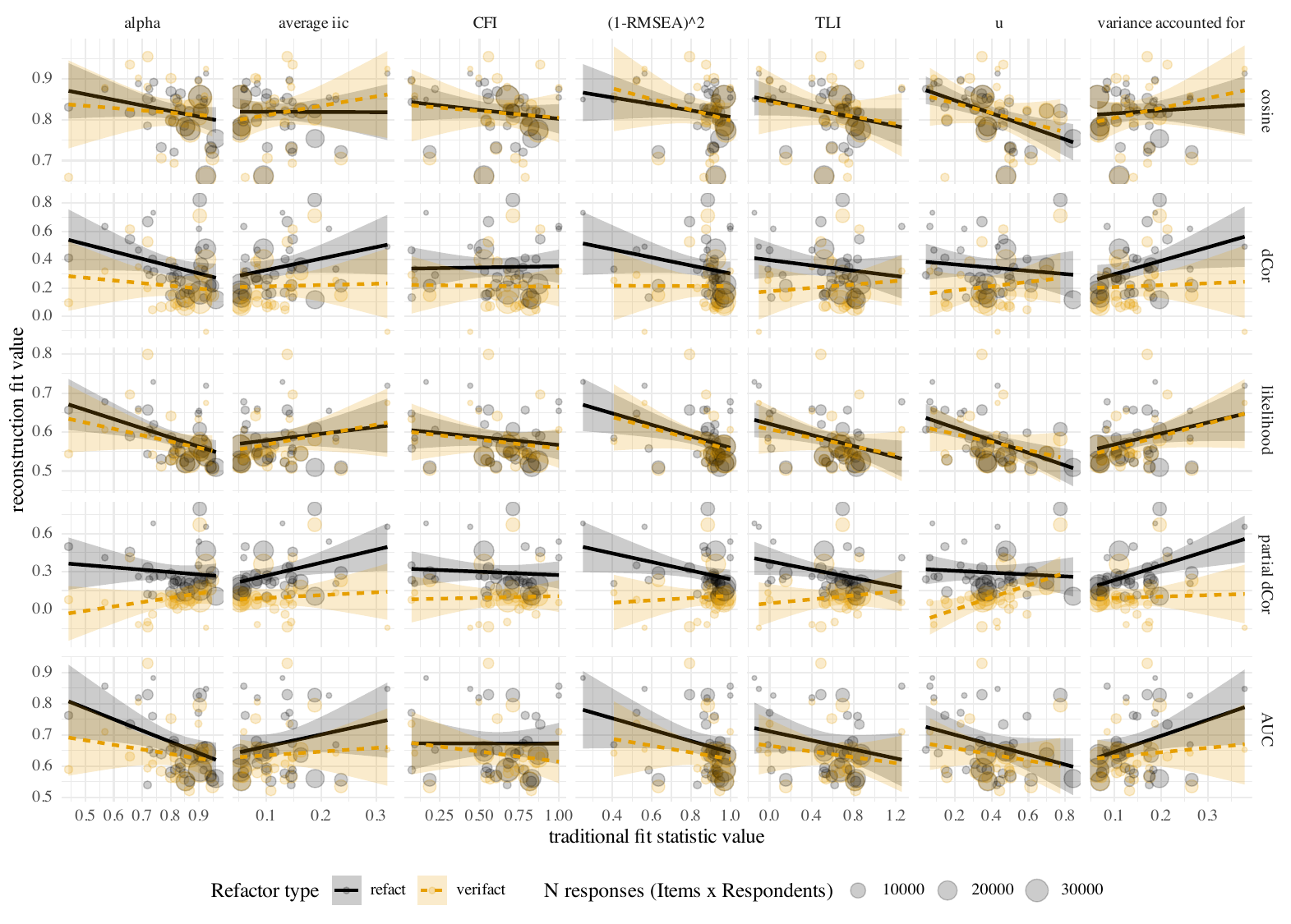}
    \caption{Refactor and Verifactor across traditional and reconstruction fit statistics. Size represents the total number of datapoints. This plot highlights the sensitivity of traditional metrics (such as alpha and RMSEA) to sample size, which would result in overly optimistic interpretations of reliability for an assumed noise structure.}
    \label{fig:refact_vs_verifact}
\end{figure*}

\begin{figure*}[h!]
    \centering
    \includegraphics[width=1\linewidth]{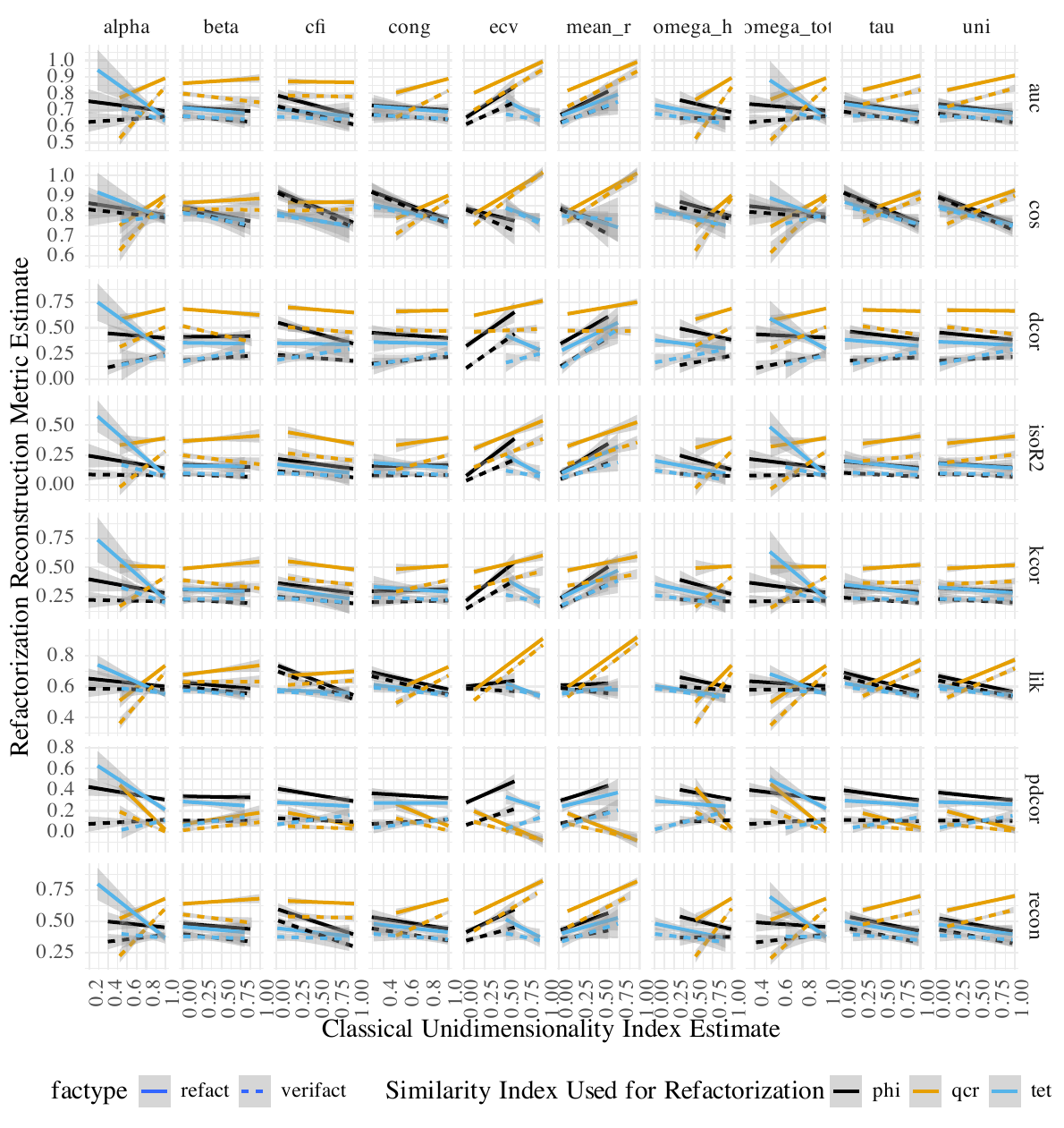}
    \caption{\textbf{Traditional Measures of Unidimensionality vs. Refactor and Verifactor Reconstruction Metrics using Tetrachoric, Pearson, and Quadrant Correlations}. \footnotesize{Aside from average interitem correlation and proportion of variance accounted for by the factor model, there is no meaningful relationship between the traditional measures of unidimensionality explored in this paper and the ability of the relationship to reconstruct the imaged signal across over 200 publicly available datasets. %Each point represents a separate publicly available dataset. 
    Each column represents a traditional metric of unidimensionality based on factor analyses: alpha = Cronbach's $\alpha$, av\_r = average interitem correlation $\bar{r}_{ij}$; CFI = Comparative Fit Index; rho\_c = a measure of the fit of a congeneric model to the observed correlations $\rho_c$; $\tau_{RC}$ = compares the observed correlations $r_{ij}$ to the mean correlation $\bar{r}_{ij}$ and considers 1 - the ratio of the sum of the squared residuals to the sum of the squared correlations; TLI = Tucker Lewis Index (tli); and u = Revelle and Condon's measure of congeneric unidimensionality, $u_{RC}$ (see \cite{revelle_unidim_2025}. Each row represents Refactor Reconstruction Metrics. X-axes represent the strength of the traditional metric for the dataset. Y-axes are the fit of the out-of-sample reconstructions via Refactor Analyses. Slopes are estimated using robust regression. Given the item-subject random effects used, weak relationships may also be capturing the sensitivity of the scale of the instrument being used to presumed fixed effects of items. 
    %The bottommost row controls for these effects via partialing out fixed effects for each item and each subject. After doing so, we see slightly stronger relationships with within-sample Refactor analysis in this row; these relationships, however, mostly vanish under Verifactor analysis, as seen in Figure \ref{fig:fafit_vs_verifact}, and 
    They do not meet what might be expected based on the strength of correlations seen across traditional metrics, as seen in the top left correlations in Figure \ref{fig:pairs_panels_verifact}.}}
    \label{fig:fafit_vs_refact}
\end{figure*}

% \input{apx/drafting/math_proof_notes}
% \subsection{Experimental Design}

% Short summary: the behaviour rests on a single central fact — are A and B (the chosen images) sufficient statistics for the aspects of X you want to predict? If yes (typical for low-rank, bilinear models) the two-image → factor analysis → reconstruction pipeline can recover and predict well, with advantages for privacy and denoising. If no, you will suffer identifiability and bias problems unless you add regularization, extra moments, or partial raw data.

 % Workflow underlying Refactoring Analyses. 
 
 % Any standard low-rank method (e.g., FA, PCA, LDA, MDS) is then applied to these images to obtain rank-r loadings for rows ($B ∈ ℝ^{n×r}$) and columns ($A ∈ ℝ^{p×r}$). Using the dual loadings, we “decompress” to reconstruct predictions of the original matrix, yielding $X̂$ that is consistent with both K and G. Classical practice infers dimensionality from properties of the images and loadings; Refactoring Analyses instead test whether the low-rank model actually captures the signal in X by directly comparing $X$ and $X̂$ with general-purpose metrics (e.g., AUC, $R^2$, distance correlation, cosine similarity, $r, ρ, τ, H, Q$).

\begin{figure*}[ht]
    \centering
    \includegraphics[width=1\linewidth]{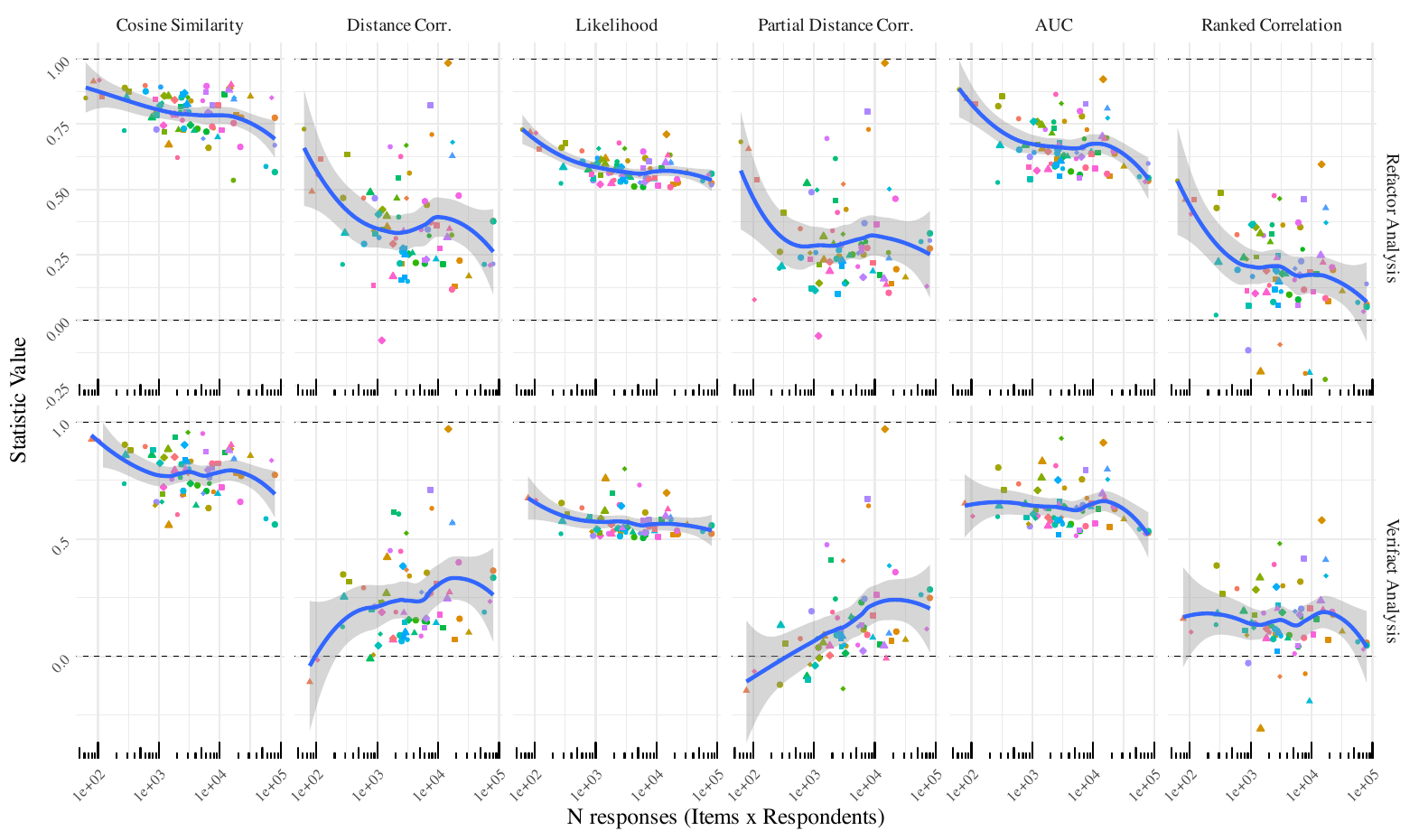}
    \caption{Tetrachoric Correlation with increasing sample size. We see that, for Refactor analysis, the model displays a potentially counterintuitive result: as sample size increases so changes the statistical power with which trivial discrepancies between the data and the model can be identified \citep{steiger_structural_1990}.}
    \label{fig:tet_all_metrics}
\end{figure*}

\begin{figure*}[h!]
    \centering
    \includegraphics[width=1\linewidth]{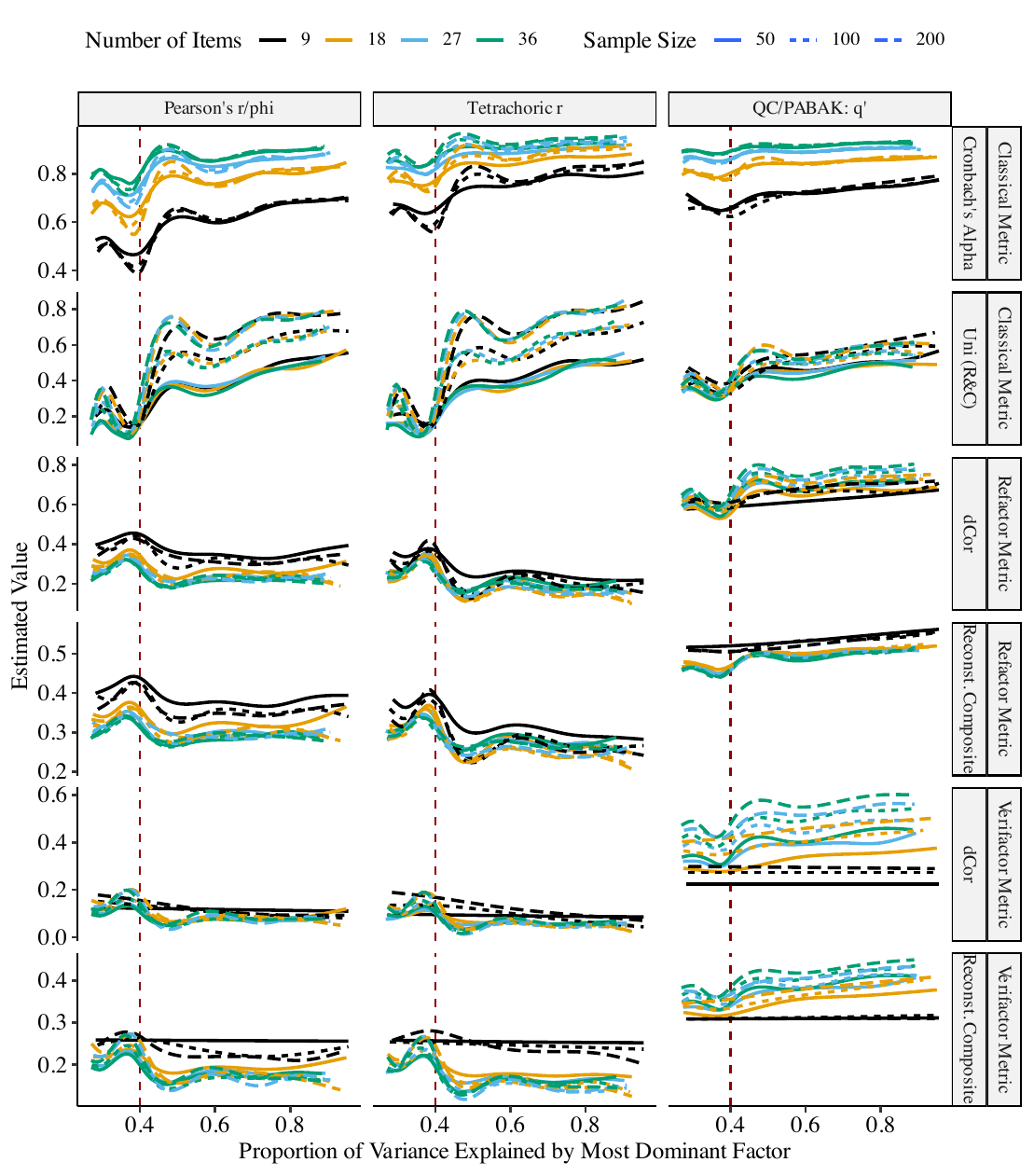}
    \caption{Comparison of Classical and Refactor Metrics of Unidimensionality on Simulated Data (Replication of \cite{revelle_unidim_2025}). X-axis represents the proportion of simulated signal represented by the most dominant dimension: sum of the squared standardized Schmid Leiman loadings $\boldsymbol{\hat{\lambda}}$ used in the hierarchical data generating modelhigher values indicate stronger dominance of a single latent factor. Y-axis represents the value of each respective metric. Additional metrics and comparisons are in Appendix \ref{apx:supp_figures}}.
    \label{fig:sim_phitetq_apx}
\end{figure*}

\begin{figure}[h]
    \centering
    \includegraphics[width=1\linewidth]{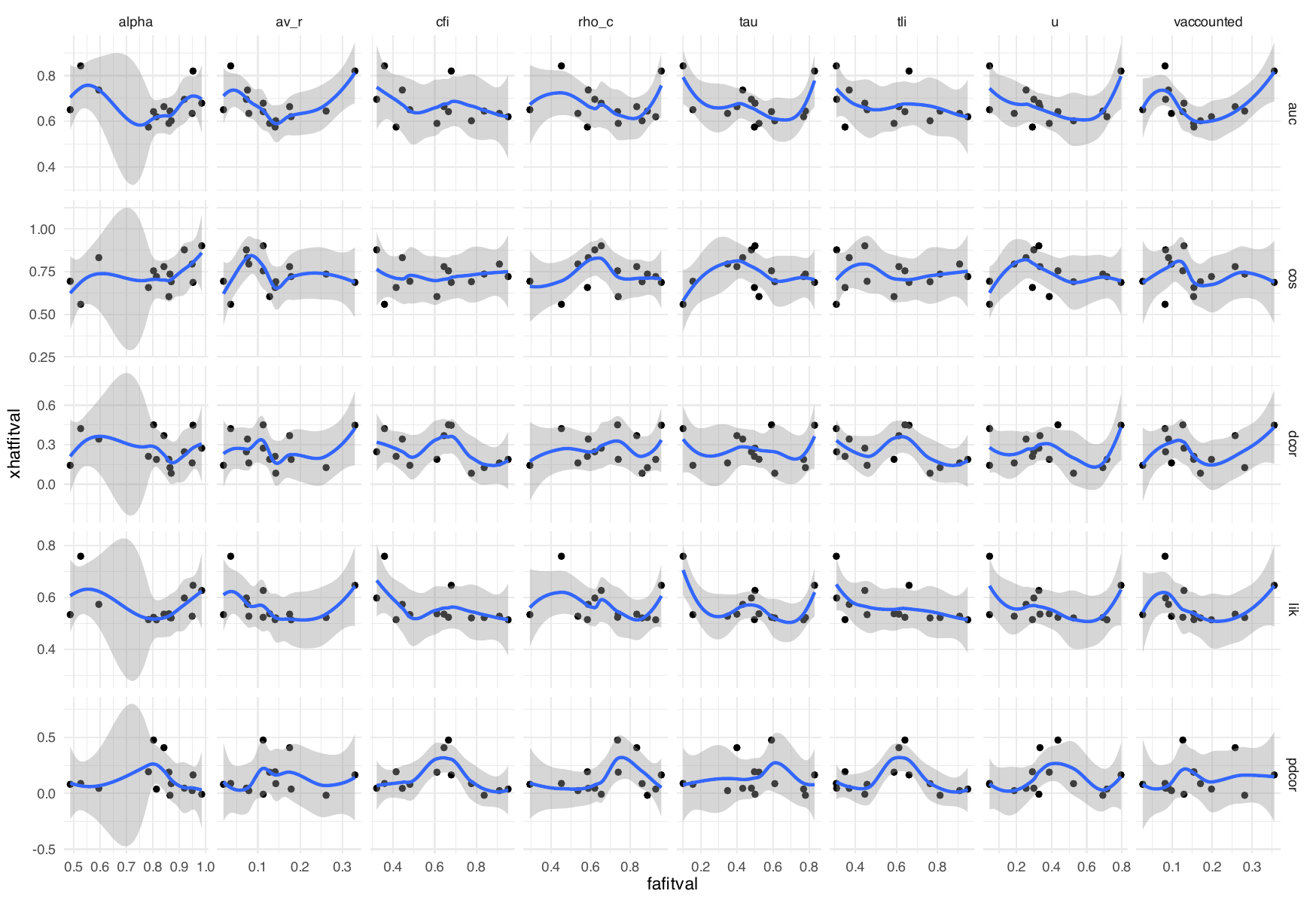}
    \caption{LOESS estimations for Tetrachoric}
    \label{fig:loess_uni_vs_verifact}
\end{figure}

\begin{figure} %{R}{0.38\textwidth}
    \centering
    \includegraphics[width=0.35\linewidth]{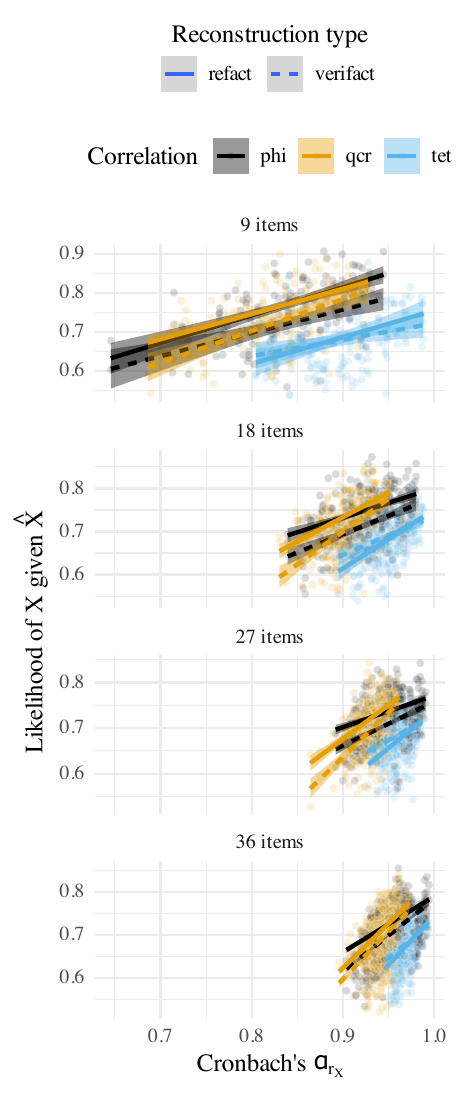}
    \caption{\textbf{Simulation I}: Correlation between Traditional Metrics and rank-1 recoverability} 
        \caption*{\footnotesize \textbf{Classical Unidimensionality (Cronbach's $\alpha_{\mathcal{A}_{X}}$) vs $\boldsymbol{\mathcal{L}(X|\widehat{X})}$ Likelihood of the data given isotonized Refactor reconstruction}: %matrix cosine between $\mathbf{X}$ and $\widehat{\mathbf{X}}$: $\cos(\mathbf{X}, \mathbf{\widehat{X}}) = \langle \mathbf{X}, \mathbf{\widehat{X}} \rangle_F / \|\mathbf{X}\|_F \|\mathbf{\widehat{X}}\|_F$ 
    across three correlations. Robust $M$-estimator linear fit and confidence intervals are shown for 200 replications at each item set size. Cronbach's $\alpha$ displays its sensitivity to number of items. Likelihood is the geometric mean of pointwise likelihood. Response matrix sizes match those of Simulation II. Fully crossed linear relationships between classical measures of unidimensionality and refactoring measures are in Figure \ref{fig:basesimpanels}} 
    \label{fig:sim1_alpha_lik}
\end{figure}

\begin{figure}[h]
    \centering
    \includegraphics[width=0.8\linewidth]{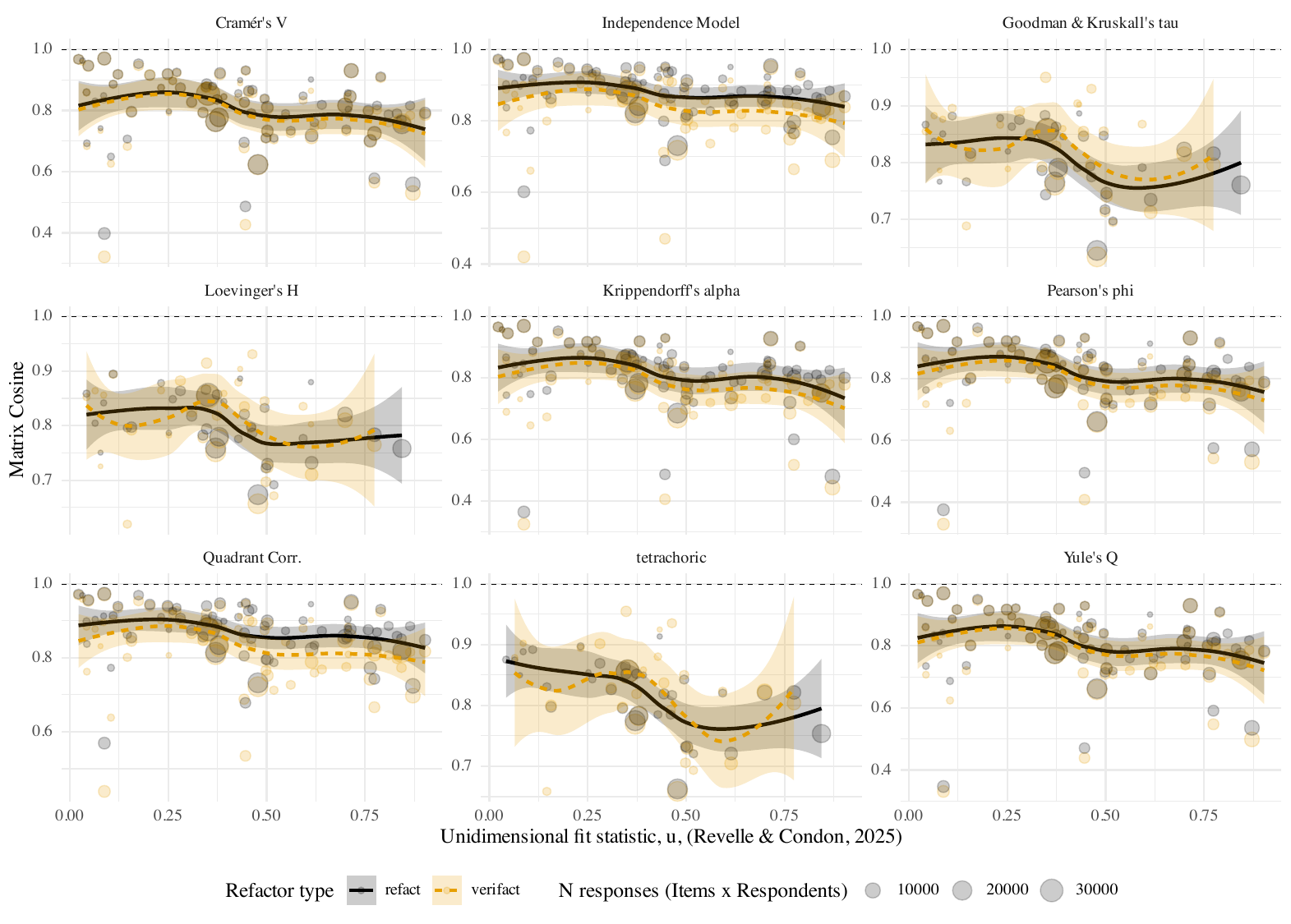}
    \caption{Comparison of 9 different indices of association $u_{RC}$ \citep{revelle_unidim_2025} vs Matrix Cosine}
    \label{fig:placeholder1}
\end{figure}

\begin{figure}[h]
    \centering
    \includegraphics[width=0.8\linewidth]{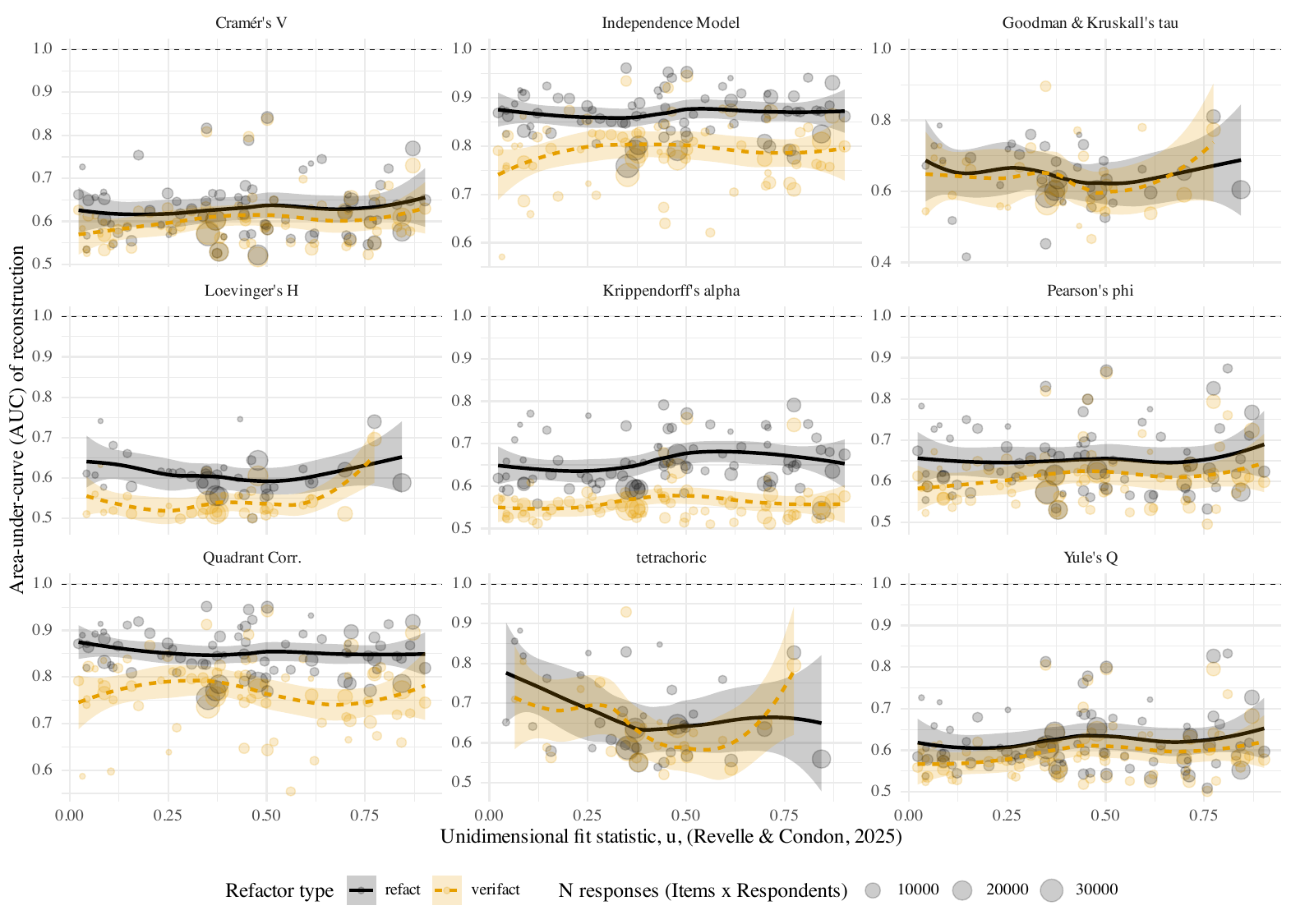}
    \caption{Comparison of 9 different indices of association $u_{RC}$ \citep{revelle_unidim_2025} vs AUC}
    \label{fig:placeholder2}
\end{figure}

\begin{figure}[h]
    \centering
    \includegraphics[width=1\linewidth]{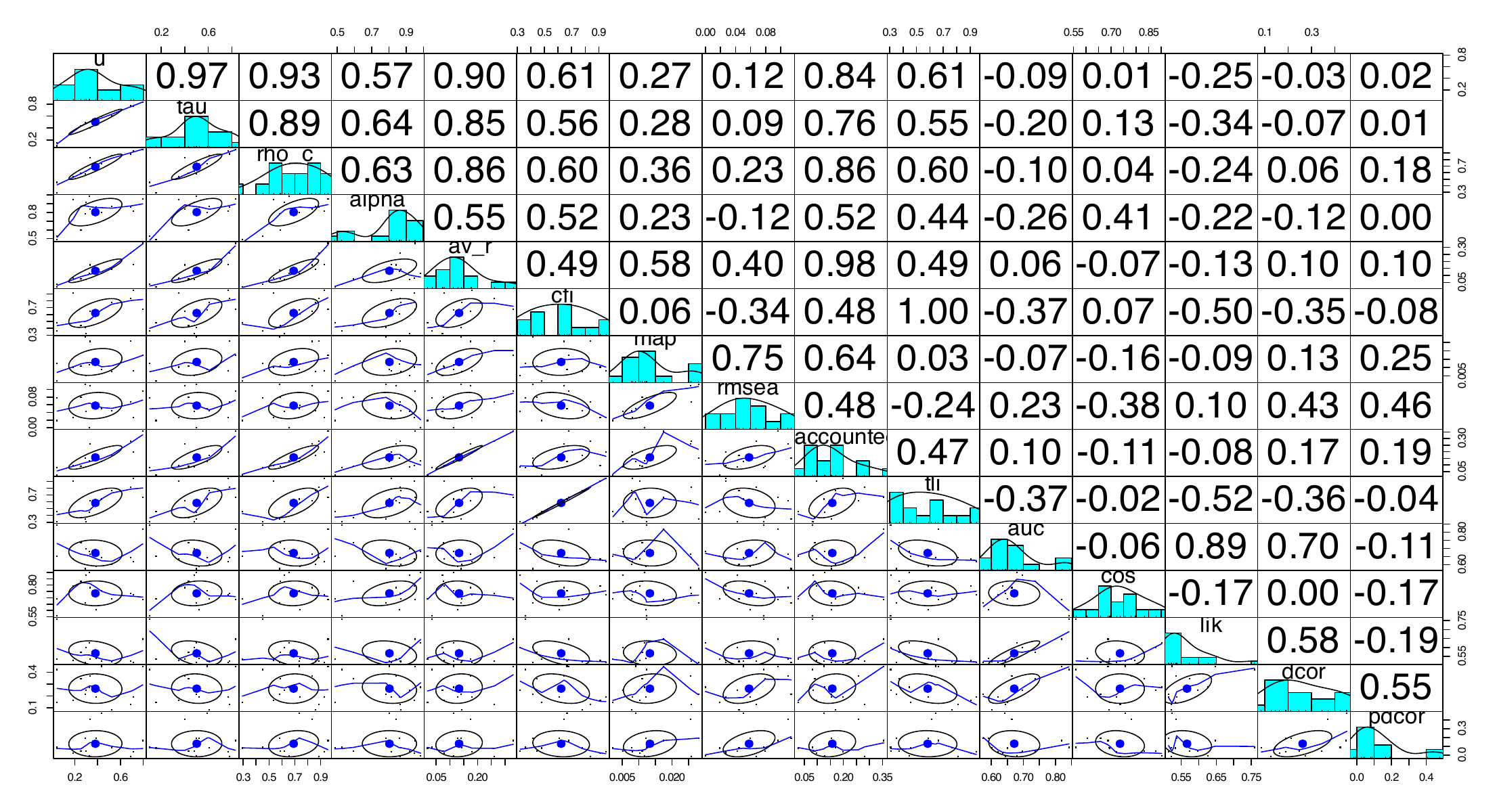}
    \caption{Enter Caption}
    \label{fig:pairs_panels_verifact}
\end{figure}

\begin{figure}[h!]
    % \centering
\begin{subfigure}{0.49\textwidth}%[h!]
    \includegraphics[width=0.999\linewidth]{figs/example_refactor_sim.pdf}
    \caption{Example refactor simulation where the data generating model reflects tetrachoric correlations}
    \label{fig:phi_corr_sim}
\end{subfigure}
\begin{subfigure}{0.49\textwidth} %[h!]
    \includegraphics[width=0.999\linewidth]{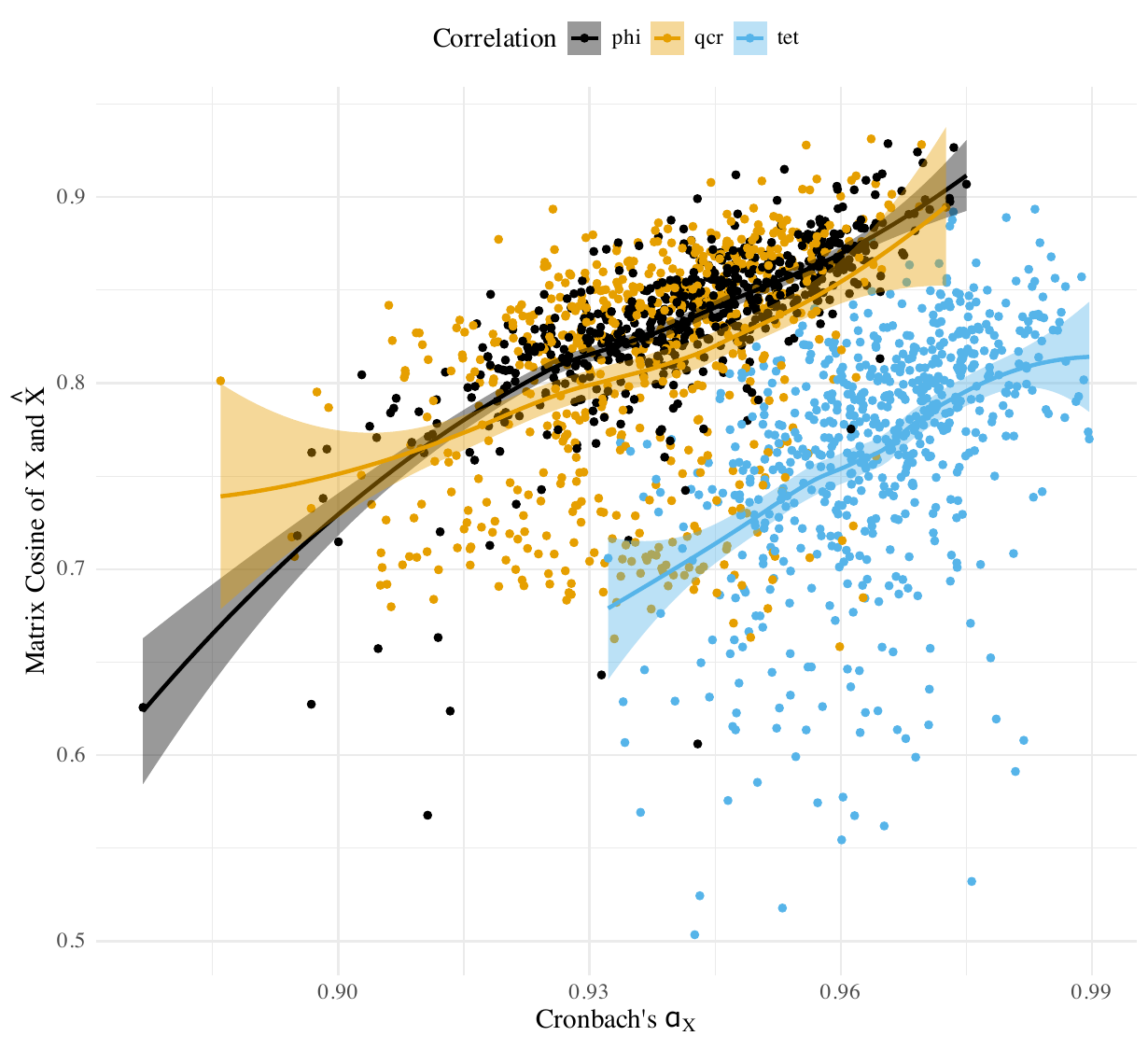}
    \caption{Classical Unidimensionality (Cronbach's $\alpha_\mathbf{X}$) vs Refactoring matrix cosine between $\mathbf{X}$ and $\widehat{\mathbf{X}}$: $\cos(\mathbf{X}, \mathbf{\widehat{X}}) = \langle \mathbf{X}, \mathbf{\widehat{X}} \rangle_F / \|\mathbf{X}\|_F \|\mathbf{\widehat{X}}\|_F$  based on across three correlations. Loess fit and confidence intervals are shown for 700 replications. Fully crossed linear relationships between classical measures of unidimensionality and refactoring measures are in Figure \ref{fig:basesimpanels} }
    \label{fig:phi_dgm_vs_auc}
\end{subfigure}
    \caption{Simulations demonstrating the Refactoring metrics and the }
    \label{fig:corr_refactor_base_sim}
\end{figure}

%While a minor relationship was noted between Cronbach's alpha and Yanai's GCD, it proved to be an artifact of the number of items (see \cite{sijtsma_use_2009}) rather than a substantive link. Given the foundational role of the unidimensionality assumption, these results underscore a critical need to shift the focus of model evaluation from the properties of an abstract data \textit{image} to the model's direct, demonstrable capacity to account for the signal in the original observations.

\end{Backmatter}

\end{document}